\documentclass[a4paper]{article}

\usepackage{yfonts}
\usepackage{marvosym}

\usepackage{amsmath,amssymb,tikz} 
\usepackage{pict2e}
\usepackage{todonotes}
\usetikzlibrary{decorations.markings}

\usetikzlibrary{calc,patterns,angles,quotes}
\usetikzlibrary{arrows,calc,chains, positioning, shapes.geometric,shapes.symbols,decorations.markings,arrows.meta}

\numberwithin{equation}{section}

\def\beq{\begin{equation}}
\def\eeq{\end{equation}}

\def\bit{\begin{itemize}}
	\def\eit{\end{itemize}}

\makeatletter
\def\eqalign#1{\null\vcenter{\def\\{\cr}\openup\jot\m@th
		\ialign{\strut$\displaystyle{##}$\hfil&$\displaystyle{{}##}$\hfil
			\crcr#1\crcr}}\,}
\makeatother

\newcommand{\R}{{\mathbb R}}
\newcommand{\C}{{\mathbb C}}

\newcommand{\al}{\alpha}
\newcommand{\be}{\beta}

\newcommand{\ga}{\gamma}
\newcommand{\Ga}{\Gamma}
\newcommand{\La}{\Lambda}

\newcommand{\ep}{\varepsilon}
\newcommand{\de}{\delta}
\newcommand{\ka}{\kappa}

\newcommand{\om}{\omega}
\newcommand{\Om}{\Omega}

\newcommand{\di}{\displaystyle}

\newcommand{\ovl}{\overline}

\newcommand{\ds}{\displaystyle}
\def\bigO{{\cal O}}
\newenvironment{proof}%
{\rm \trivlist \item[\hskip \labelsep{\bf Proof. }]}%
{\hspace*{\fill}$\Box$\endtrivlist}

\begin{document}
	\tikzset{middlearrow/.style={
			decoration={markings,
				mark= at position 0.6 with {\arrow{#1}} ,
			},
			postaction={decorate}
		}
	}
	
	\newtheorem{theorem}{Theorem}
	\newtheorem{acknowledgement}[theorem]{Acknowledgement}
	\newtheorem{remark}[theorem]{Remark}
	\newtheorem{lemma}[theorem]{Lemma}
	\newtheorem{proposition}[theorem]{Proposition}
	\newtheorem{corollary}[theorem]{Corollary}
	\numberwithin{equation}{section}
	\numberwithin{theorem}{section}
	
	\def\Xint#1{\mathchoice
		{\XXint\displaystyle\textstyle{#1}}%
		{\XXint\textstyle\scriptstyle{#1}}%
		{\XXint\scriptstyle\scriptscriptstyle{#1}}%
		{\XXint\scriptscriptstyle\scriptscriptstyle{#1}}%
		\!\int}
	\def\XXint#1#2#3{{\setbox0=\hbox{$#1{#2#3}{\int}$ }
			\vcenter{\hbox{$#2#3$ }}\kern-.59\wd0}}
	\def\ddashint{\Xint=}
	\def\dashint{\Xint-}
	
	\title{Asymptotics of Hankel determinants with a Laguerre-type or Jacobi-type potential and Fisher-Hartwig singularities}
	\author{Christophe Charlier\footnote{Department of Mathematics, KTH Royal Institute of Technology, Lindstedtsv\"{a}gen 25, SE-114 28 Stockholm, Sweden. e-mail: cchar@kth.se}, Roozbeh Gharakhloo\footnote{Department of Mathematics, Colorado State University, 1874 Campus Delivery, Fort Collins, CO 80523-1874, USA. e-mail: roozbeh.gharakhloo@colostate.edu}}
	
	\maketitle
	
	\tikzset{->-/.style={decoration={
				markings,
				mark=at position #1 with {\arrow{latex}}},postaction={decorate}}}
	
	\tikzset{-<-/.style={decoration={
				markings,
				mark=at position #1 with {\arrowreversed{latex}}},postaction={decorate}}}

\begin{abstract}
We obtain large $n$ asymptotics of $n \times n$ Hankel determinants whose weight has a one-cut regular potential and Fisher-Hartwig singularities. We restrict our attention to the case where the associated equilibrium measure possesses either one soft edge and one hard edge (Laguerre-type) or two hard edges (Jacobi-type). We also present some applications in the theory of random matrices. In particular, we can deduce from our results asymptotics for partition functions with singularities, central limit theorems, correlations of the characteristic polynomials, and gap probabilities for (piecewise constant) thinned Laguerre and Jacobi-type ensembles. Finally, we mention some links with the topics of rigidity and Gaussian multiplicative chaos.
\end{abstract}
\section{Introduction}
Hankel determinants with Fisher-Hartwig (FH) singularities appear naturally in random matrix theory. Among others, they can express correlations of the characteristic polynomial of a random matrix, or gap probabilities in the point process of the thinned spectrum, see e.g. the introductions of \cite{Krasovsky,DIK,Charlier} for more details. In these applications, the size $n$ of an $n \times n$ Hankel determinant is equal to the size of the underlying $n \times n$ random matrices. Large $n$ asymptotics for such determinants have already been widely studied, see e.g. \cite{Krasovsky, ItsKrasovsky, DeiftItsKrasovsky, BerWebbWong, Charlier}. Recent developments in the theory of Gaussian multiplicative chaos \cite{Webb,BerWebbWong} and in the study of the global rigidity of random matrix eigenvalues \cite{ClaeysFahsLambertWebb} provide a renewed interest in these asymptotics. 

\medskip In the present work, we focus our attention on the large $n$ asymptotics of the Hankel determinant
\begin{equation}\label{Hankel introduction}
\det \left( \int_{\mathcal{I}} x^{j+k}w(x)dx \right)_{j,k=0,\ldots,n-1},
\end{equation}
whose weight $w$ is supported on an interval $\mathcal{I}\subset \mathbb{R}$, and is of the form
\begin{equation}\label{weight introduction}
w(x) = e^{-nV(x)}e^{W(x)}\omega(x).
\end{equation}
The function $W$ is continuous on $\mathcal{I}$ and $\omega$ contains the FH singularities. In case $\mathcal{I}$ is unbounded, we also require that $W(x) = \bigO(V(x))$ as $x \to \pm \infty, x \in \mathcal{I}$. The functions $W$ and $\omega$ will be described in more detail below. The potential $V$ is real analytic on $\mathcal{I}$ and, in case $\mathcal{I}$ is unbounded, satisfies $\lim_{x \to \pm \infty, x \in \mathcal{I}}V(x)/\log|x| = + \infty$. Furthermore, we assume that $V$ is one-cut and regular. These properties are described in terms of the equilibrium measure $\mu_{V}$, which is the unique minimizer of the functional
\begin{equation}
\iint \log |x-y|^{-1} d\mu(x)d\mu(y) + \int V(x)d\mu(x)
\end{equation}
among all Borel probability measures $\mu$ on $\mathcal{I}$. One-cut means that the support of $\mu_{V}$ consists of a single interval. For convenience, and without loss of generality, we will assume that this interval is $[-1,1]$. It is known (see e.g. \cite{SaTo}) that $\mu_{V}$ is completely characterized by the Euler-Lagrange variational conditions
\begin{align}
2 \int_{-1}^{1} \log |x-s| d\mu_{V}(s) = V(x) - \ell, & & \mbox{ for } x \in [-1,1], \label{var equality} \\
2 \int_{-1}^{1} \log |x-s| d\mu_{V}(s) \leq V(x) - \ell, & & \mbox{ for } x \in \mathcal{I}\setminus [-1,1], \label{var inequality}
\end{align} 
where $\ell \in \mathbb{R}$ is a constant.
Regular means that the Euler-Lagrange inequality \eqref{var inequality} is strict on $\mathcal{I}\setminus [-1,1]$, and that the density of the equilibrium measure is positive on $(-1,1)$. The three canonical cases are the following:
\begin{enumerate}
\setlength{\itemindent}{0cm}
\item $\mathcal{I} = \mathbb{R}$ and $d\mu_{V}(x) = \psi(x)\sqrt{1-x^{2}}dx$, 
\item \vspace{-0.2cm} $\mathcal{I}= [-1,\infty)$ and $d\mu_{V}(x) = \psi(x)\sqrt{\frac{1-x}{1+x}}dx$,
\item \vspace{-0.2cm} $\mathcal{I} = [-1,1]$ and $d\mu_{V}(x) = \psi(x)\frac{1}{\sqrt{1-x^{2}}}dx$,
\end{enumerate}
where $\psi$ is real analytic on $\mathcal{I}$, such that $\psi(x) > 0$ for all $x \in [-1,1]$. We will refer to these three cases as Gaussian-type, Laguerre-type and Jacobi-type weights, respectively. Note that \eqref{var inequality} is automatically satisfied for Jacobi-type weights, since $\mathcal{I} = [-1,1]$. Well-known examples for potentials of such weights are 
\begin{enumerate}
\item $V(x) = 2x^{2}$ for Gaussian-type weights, with $\ell = 1+2\log 2$ and $\psi(x) = \frac{2}{\pi}$,
\item $V(x) = 2(x+1)$ for Laguerre-type weights, with $\ell = 2+2\log 2$ and $\psi(x) = \frac{1}{\pi}$,
\item $V(x) = 0$ for Jacobi-type weights, with $\ell = 2\log 2$ and $\psi(x) = \frac{1}{\pi}$.
\end{enumerate}
In the language of random matrix theory, the interval $(-1,1)$ is called the bulk, and $\pm 1$ are the edges. An edge is said to be ``soft" if there can be eigenvalues beyond it, and ``hard" if this is impossible. On the level of the equilibrium measure, a soft edge translates into a square root vanishing of $\frac{d\mu_{V}}{dx}$, while a hard edge means that $\frac{d\mu_{V}}{dx}$ blows up like an inverse square root. Thus, there are two soft edges at $\pm 1$ for Gaussian-type weights, one hard edge at $-1$ and one soft edge at $1$ for Laguerre-type weights, and two hard edges at $\pm 1$ for Jacobi-type weights.

\medskip The function $\omega$ that appears in \eqref{weight introduction} is defined by
\begin{equation}\label{weight FH}
\omega(x) = \prod_{j=1}^{m} \omega_{\alpha_{j}}(x)\omega_{\beta_{j}}(x) \times  \left\{ \begin{array}{c l}
\ds 1, & \mbox{for Gaussian-type weights}, \\[0.1cm]
\ds (x+1)^{\alpha_{0}}, & \mbox{for Laguerre-type weights}, \\[0.1cm]
\ds (x+1)^{\alpha_{0}}(1-x)^{\alpha_{m+1}}, & \mbox{for Jacobi-type weights},
\end{array}  \right.
\end{equation}
where
\begin{equation}\label{FH pieces}
\omega_{\alpha_{k}}(x) = |x-t_{k}|^{\alpha_{k}}, \qquad \omega_{\beta_{k}}(x) = \left\{ \begin{array}{l l}
e^{i\pi\beta_{k}}, & \mbox{ if } x < t_{k}, \\
e^{-i \pi \beta_{k}}, & \mbox{ if } x > t_{k},
\end{array}  \right. 
\end{equation}
with 
\begin{equation}
-1 < t_{1} < \ldots < t_{m} < 1.
\end{equation}
The functions $\omega_{\alpha_{k}}$ and $\omega_{\beta_{k}}$ represent the root-type and jump-type singularities at $t_{k}$, respectively. These singularities are named after Fisher and Hartwig, due to their pioneering work in their identification \cite{FisherHartwig}. Since $\omega_{\beta_{k}+1} = - \omega_{\beta_{k}}$, we can assume without loss of generality that $\Re \beta_{k} \in (-\frac{1}{2},\frac{1}{2}]$ for all $k$. Finally, to ensure integrability of the weight (at least for sufficiently large $n$), we require that $\Re \alpha_{k} > -1$ for all $k$ and, in case $\mathcal{I}$ is unbounded, that $W(x) = \bigO(V(x))$ as $x \to \pm \infty, x \in \mathcal{I}$.

\medskip To summarise, the $n \times n$ Hankel determinant given by \eqref{Hankel introduction} depends on $n$, $m$, $V$, $W$, $\vec{t} = (t_{1},\ldots,t_{m})$, $\vec{\beta} = (\beta_{1},\ldots,\beta_{m})$ and $\vec{\alpha}$, where
\begin{equation*}
\vec{\alpha} = \left\{ \begin{array}{l l}
(\alpha_{1},\ldots,\alpha_{m}), & \mbox{if $w$ is a Gaussian-type weight}, \\
(\alpha_{0},\alpha_{1},\ldots,\alpha_{m}), & \mbox{if $w$ is a Laguerre-type weight}, \\
(\alpha_{0},\alpha_{1},\ldots,\alpha_{m},\alpha_{m+1}), & \mbox{if $w$ is a Jacobi-type weight}.
\end{array} \right.
\end{equation*}
This determinant will be denoted by $G_{n}(\vec{\alpha},\vec{\beta},V,W)$, $L_{n}(\vec{\alpha},\vec{\beta},V,W)$ or $J_{n}(\vec{\alpha},\vec{\beta},V,W)$, depending on whether the weight is of Gaussian, Laguerre or Jacobi-type, respectively.

\medskip We mention that the original work \cite{FisherHartwig} of Fisher and Hartwig deals with Toeplitz determinants whose entries are the Fourier coefficients of a symbol defined on the unit circle. Asymptotics of large Toeplitz determinants with Fisher-Hartwig singularities have been conjectured in \cite{FisherHartwig}, and also by Lenard \cite{Lenard}, and then proved by Widom \cite{Widom}, Basor \cite{Basor, Basor2} and B\"ottcher and Silbermann \cite{BS} under certain restrictions on the parameters $\alpha_{j}$'s and $\beta_{j}$'s. The generalization of these results for general values of $\alpha_{j}$'s and $\beta_{j}$'s has been proved by Deift, Its and Krasovsky in \cite{DIK, DeiftItsKrasovsky}. 

\medskip Many authors have contributed over the years to the large $n$ asymptotics of $G_{n}(\vec{\alpha},\vec{\beta},V,W)$ in certain particular cases of the parameters $\vec{\alpha}$, $\vec{\beta}$, $V$ and $W$, and we briefly review these historical developments here. The first result in this direction is due to Johansson \cite{Johansson2},\footnote{In fact the works \cite{BorGui, FRW2017, Johansson2} deal with general $\beta$ ensembles. In this paper we restrict ourselves to $\beta = 2$.} who obtained large $n$ asymptotics of the ratio 
\begin{align*}
\frac{G_{n}(\vec{0},\vec{0},V,W)}{G_{n}(\vec{0},\vec{0},V,0)}.
\end{align*}
Shortly afterward, the existence of a full asymptotic expansion of $G_{n}(\vec{0},\vec{0},V,0)$ was proved in \cite{ErcMcL,BleIts}, in the case where the potential $V$ is a polynomial satisfying some mild regularity assumptions. It was later shown in \cite{ClaeysGravaMcLaughlin} that the assumptions on $V$ made in \cite{ErcMcL} always hold, so that the result of \cite{ErcMcL} actually holds for any one-cut polynomial $V$. A conjecture for the large $n$ asymptotics of the ratio
\begin{align*}
\frac{G_{n}(\vec{\alpha},\vec{0},V,W)}{G_{n + \frac{\mathcal{A}}{2}}(\vec{0},\vec{0},V,W+\frac{\mathcal{A}}{2}V)}, \qquad \mathcal{A}:= \sum_{j=1}^{m}\alpha_{j}
\end{align*}
was formulated by Forrester and Frankel in \cite[Conjecture 8]{ForFra} for general $m$, $V$, and $W$ (we mention that this conjecture, specialized to $m=1$ and $W=0$, was formulated in the earlier work of Br\'{e}zin and Hikami \cite{BrezinHikami}). The conjecture of \cite{ForFra}, specialized to 
\begin{align*}
\frac{\alpha_{1}}{2},\frac{\alpha_{2}}{2},\ldots,\frac{\alpha_{m}}{2} \in \mathbb{N}:=\{0,1,2,\ldots\}, \qquad V(x)=2x^{2}, \qquad \mbox{ and } \qquad W=0,
\end{align*}
was proved rigorously by Garoni in \cite{Garoni}. Simultaneously to \cite{Garoni}, the asymptotics of $G_{n}(\vec{\alpha},\vec{0},2x^{2},0)$ for general values of $\Re \alpha_{j}>-1$ were obtained by Krasovsky in \cite{Krasovsky}. This result was recently generalized by Beresticky, Webb and Wong in \cite{BerWebbWong}, who obtained large $n$ asymptotics for $G_{n}(\vec{\alpha},\vec{0},V,W)$ for general $V$ (not necessarily a polynomial), $W$ and $\vec{\alpha}$. In particular, they were able to prove \cite[Conjecture 8]{ForFra} in its full generality. In a different direction, Its and Krasovsky in \cite{ItsKrasovsky} obtained the asymptotics of $G_{n}(\vec{0},\vec{\beta},2x^{2},0)$, in the case where $m=1$. Finally, the asymptotics of $G_{n}(\vec{\alpha},\vec{\beta},V,W)$ for general values of the parameters were obtained in \cite{Charlier} (see also Theorem \ref{theorem G} below for the precise statement). For another review of these historical developments, see also the introduction of \cite{Charlier}. It is worth to mention that the asymptotics of \cite{Charlier} are only valid for $\Re \beta_{k} \in (-\frac{1}{4},\frac{1}{4})$ and not in the whole strip $\Re \beta_{k} \in (-\frac{1}{2},\frac{1}{2}]$. This is due to technical reasons, and we comment more on that in Remark \ref{remark real part beta} below.

\medskip Much less is known about the large $n$ asymptotics of $L_{n}(\vec{\alpha},\vec{\beta},V,W)$ and $J_{n}(\vec{\alpha},\vec{\beta},V,W)$, and we briefly discuss this here. 

\medskip The quantities $L_{n}((\alpha_{0},0,\ldots,0),\vec{0},V,0)$ and $J_{n}((\alpha_{0},0,\ldots,0,\alpha_{m+1}),\vec{0},V,0)$ (i.e. $W = 0$ and no other singularities than $\alpha_{0}$ and $\alpha_{m+1}$) represent partition functions of certain random matrix ensembles, see also Section \ref{subsection:Applications} below. In some very special cases of $V$ (like $V(x) = 2(x+1)$ for Laguerre-type weights and $V(x) = 0$ for Jacobi-type weights), these Hankel determinants reduce to Selberg integrals and are thus computable explicitly. The existence of expansions to all orders for $L_{n}(\vec{0},\vec{0},V,0)$ and $J_{n}(\vec{0},\vec{0},V,0)$ as $n \to + \infty$ for a general $V$ was proved in \cite{BorGui}. However, to the best of our knowledge, the explicit values of the coefficients appearing in these expansions, which are given by Theorem \ref{theorem L} and Theorem \ref{theorem J} below with $\alpha_{1}=\ldots=\alpha_{m}=0$, $\vec{\beta} = \vec{0}$ and $W = 0$, are already new results. 

\medskip Forrester and Frankel, in \cite[Conjecture 9]{ForFra}, have formulated a conjecture for the asymptotics of the ratios 
\begin{align*}
\frac{L_{n}(\vec{\alpha}_{0},\vec{0},V,W)}{L_{n}(\vec{\alpha}_{0},\vec{0},V,0)} \quad \mbox{and} \quad \frac{L_{n}(\vec{\alpha},\vec{0},V,W)}{L_{n+\frac{\mathcal{A}_{0}}{2}}(\vec{\alpha}_{0},\vec{0},V,W+\frac{\mathcal{A}_{0}}{2}V)}, \qquad \mbox{where} \quad \mathcal{A}_{0}:= \sum_{j=1}^{m}\alpha_{j},
\end{align*}
and where $\vec{\alpha}_{0}:=(\alpha_{0},0,\ldots,0)$. This conjecture, specialized to
\begin{align}\label{ratio Laguerre conjecture}
\frac{\alpha_{1}}{2},\frac{\alpha_{2}}{2},\ldots,\frac{\alpha_{m}}{2} \in \mathbb{N}:=\{0,1,2,\ldots\}, \qquad V(x)=2(x+1), \qquad \mbox{ and } \qquad W=0,
\end{align}
was also verified rigorously by Garoni in \cite{Garoni}. More recently, the asymptotics of 
\begin{align*}
\frac{L_{n}(\vec{\alpha}_{0},\vec{0},2(x+1),W)}{L_{n}(\vec{\alpha}_{0},\vec{0},2(x+1),0)}
\end{align*}
have been obtained in \cite[Proposition 3.9]{FRW2017} using loop equations. To the best of our knowledge, there are no other results available in the literature (prior to this work) for Hankel determinants associated to a Laguerre-type weight with FH singularities in the bulk. 

\medskip There is more known about Jacobi-type weights. 
Asymptotics for $J_{n}((\alpha_{0},0,\ldots,0,\alpha_{m+1}),\vec{0},0,W)$ (i.e. root-type singularities only at the edges) were computed in \cite{KMcLVAV}, however without the constant term. Major progress was achieved in \cite{DIK, DeiftItsKrasovsky}, in which the authors derived large $n$ asymptotics for $J_{n}(\vec{\alpha},\vec{\beta},0,W)$ including the constant term (under weak assumptions on $W$, and for general values of $\vec{\beta}$ such that $\Re \beta_{k} \in (-\frac{1}{2},\frac{1}{2}]$).

\medskip The goal of the present paper is to fill a gap in the literature on the large $n$ asymptotics of Hankel determinants with a one-cut potential and FH singularities. In Theorem \ref{theorem L} and Theorem \ref{theorem J} below, we obtain large $n$ asymptotics for $L_{n}(\vec{\alpha},\vec{\beta},V,W)$ and $J_{n}(\vec{\alpha},\vec{\beta},V,W)$ including the constant term. First, we rewrite (in a slightly different way) the result of \cite{Charlier} in Theorem \ref{theorem G} for the reader's convenience, in order to ease the comparison between the three canonical types of weights. 

\newpage \begin{theorem}[from \cite{Charlier} for Gaussian-type weights]\label{theorem G}\ \\
Let $m \in \mathbb{N}$, and let $t_{j}$, $\alpha_{j}$ and $\beta_{j}$ be such that
\begin{equation*}
-1 < t_{1} < \ldots < t_{m} < 1, \quad \mbox{ and } \quad \Re \alpha_{j} > -1, \quad \Re \beta_{j} \in (-\tfrac{1}{4},\tfrac{1}{4}) \quad \mbox{ for } j=1,\ldots,m.
\end{equation*} 
Let $V$ be a one-cut regular potential whose equilibrium measure is supported on $[-1,1]$ with density $\psi(x)\sqrt{1-x^{2}}$, and let $W: \mathbb{R}\to\mathbb{R}$ be analytic in a neighbourhood of $[-1,1]$, locally H\"{o}lder-continuous on $\mathbb{R}$ and such that $W(x) = \bigO(V(x)), \mbox{ as } |x| \to \infty$. As $n \to \infty$, we have
\begin{equation}\label{asymp thm Gn}
G_{n}(\vec{\alpha},\vec{\beta},V,W) = \exp\left(C_{1} n^{2} + C_{2} n + C_{3} \log n + C_{4} + \bigO \Big( \frac{\log n}{n^{1-4\beta_{\max}}} \Big)\right),
\end{equation}
with $\beta_{\max} = \max \{ |\Re \beta_{1}|,\ldots,|\Re \beta_{m}| \}$ and 
\begin{align}
& C_{1} = - \log 2 - \frac{3}{4} - \frac{1}{2} \int_{-1}^{1} (V(x)-2x^{2})\left( \frac{2}{\pi} + \psi(x) \right)\sqrt{1-x^{2}}dx, \\
& C_{2} = \log(2\pi) - \mathcal{A}\log 2 - \frac{\mathcal{A}}{2\pi}\int_{-1}^{1} \frac{V(x)}{\sqrt{1-x^{2}}}dx + \int_{-1}^{1}W(x)\psi(x)\sqrt{1-x^{2}}dx  \\
& \hspace{1cm} + \sum_{j=1}^{m} \frac{\alpha_{j}}{2}V(t_{j}) + \sum_{j=1}^{m} \pi i \beta_{j} \left( 1 - 2 \int_{t_{j}}^{1}\psi(x)\sqrt{1-x^{2}}dx \right) , \nonumber \\
& C_{3} = - \frac{1}{12} + \sum_{j=1}^{m} \bigg( \frac{\alpha_{j}^{2}}{4} - \beta_{j}^{2} \bigg), \\
& C_{4} = \zeta^{\prime}(-1) - \frac{1}{24}\log \left( \frac{\pi}{2}\psi(-1) \right) - \frac{1}{24}\log \left( \frac{\pi}{2}\psi(1) \right) + \sum_{j=1}^{m} \bigg(\frac{\alpha_{j}^{2}}{4}-\beta_{j}^{2}\bigg)\log\left(\frac{\pi}{2}\psi(t_{j})\right) \nonumber \\
& \hspace{1cm}  + \sum_{1\leq j < k \leq m} \Bigg[ \log  \Bigg(  \frac{\big(1-t_{j}t_{k}-\sqrt{(1-t_{\smash{j}}^{2})(1-t_{k}^{2})}\big)^{2\beta_{j}\beta{k}}}{2^{\frac{\alpha_{j}\alpha_{k}}{2}}|t_{j}-t_{k}|^{\frac{\alpha_{j}\alpha_{k}}{2} + 2\beta_{j}\beta_{k}}} \Bigg) + \frac{i\pi}{2} (\alpha_{k}\beta_{j}-\alpha_{j}\beta_{k})  \Bigg] \nonumber \\
& \hspace{1cm} + \sum_{j=1}^{m} \left( \frac{\alpha_{j}^{2}}{4} \log \big( 2 \sqrt{1-t_{\smash{j}}^{2}} \big) - \beta_{j}^{2} \log \Big( 8(1-t_{j}^{2})^{3/2} \Big) \right) + \mathcal{A}\sum_{j=1}^{m} i \beta_{j}\arcsin t_{j}  \nonumber \\
&  \hspace{1cm}+\sum_{j=1}^{m} \log \frac{G(1+\frac{\alpha_{j}}{2}+\beta_{j})G(1+\frac{\alpha_{j}}{2}-\beta_{j})}{G(1+\alpha_{j})}  \\
& \hspace{1cm}+ \frac{\mathcal{A}}{2\pi} \int_{-1}^{1} \frac{W(x)}{\sqrt{1-x^{2}}}dx - \sum_{j=1}^{m} \frac{\alpha_{j}}{2}W(t_{j}) +\sum_{j=1}^{m}  \frac{i\beta_{j}}{\pi} \sqrt{1-t_{\smash{j}}^{2}} \dashint_{-1}^{1} \frac{W(x)}{\sqrt{1-x^{2}}(t_{j}-x)}dx  \nonumber \\
& \hspace{1cm} + \frac{1}{4\pi^{2}}\int_{-1}^{1}  \frac{W(x)}{\sqrt{1-x^{2}}} \bigg(\dashint_{-1}^{1} \frac{W^{\prime}(y)\sqrt{1-y^{2}}}{x-y}dy \bigg) dx, \nonumber
\end{align}
where $G$ is Barnes' $G$-function, $\zeta$ is Riemann's zeta-function, where we use the notations $\dashint$ for the Cauchy principal value integral, and
\begin{equation}
\mathcal{A} = \sum_{j=1}^{m} \alpha_{j}.
\end{equation}
Furthermore, the error term in \eqref{asymp thm Gn} is uniform for all $\alpha_{k}$ in compact subsets of \\$\{ z \in \mathbb{C}: \Re z >-1 \}$, for all $\beta_{k}$ in compact subsets of  $\{ z \in \mathbb{C}: \Re z \in \big( \frac{-1}{4},\frac{1}{4} \big) \}$, and uniform in $t_{1},\ldots,t_{m}$, as long as there exists $\delta > 0$ independent of $n$ such that
\begin{equation}
\min_{j\neq k}\{ |t_{j}-t_{k}|,|t_{j}-1|,|t_{j}+1|\} \geq \delta.
\end{equation}
\end{theorem}

\newpage
\begin{theorem}[for Laguerre-type weights]\label{theorem L}\ \\
Let $m \in \mathbb{N}$, and let $t_{j}$, $\alpha_{j}$ and $\beta_{j}$ be such that
\begin{equation*}
-1 = t_{0} < t_{1} < \ldots < t_{m} < 1, \quad \mbox{ and } \quad \Re \alpha_{j} > -1, \quad \Re \beta_{j} \in (-\tfrac{1}{4},\tfrac{1}{4}) \quad \mbox{ for } j=0,\ldots,m,
\end{equation*} 
with $\beta_{0} = 0$. Let $V$ be a one-cut regular potential whose equilibrium measure is supported on $[-1,1]$ with density $\psi(x)\sqrt{\frac{1-x}{1+x}}$, and let $W: [-1,+\infty)\to\mathbb{R}$ be analytic in a neighbourhood of $[-1,1]$, locally H\"{o}lder-continuous on $\mathbb{R}^{+}$ and such that $W(x) = \bigO(V(x)), \mbox{ as } x \to +\infty$. As $n \to \infty$, we have
\begin{equation}\label{asymp thm Ln}
L_{n}(\vec{\alpha},\vec{\beta},V,W) = \exp\left(C_{1} n^{2} + C_{2} n + C_{3} \log n + C_{4} + \bigO \Big( \frac{\log n}{n^{1-4\beta_{\max}}} \Big)\right),
\end{equation}
with $\beta_{\max} = \max \{ |\Re \beta_{1}|,\ldots,|\Re \beta_{m}| \}$ and 
\begin{align}
& C_{1} = - \log 2 - \frac{3}{2} - \frac{1}{2} \int_{-1}^{1} (V(x)-2(x+1))\left( \frac{1}{\pi} + \psi(x) \right)\sqrt{\frac{1-x}{1+x}}dx, \\
& C_{2} = \log(2\pi) - \mathcal{A}\log 2 - \frac{\mathcal{A}}{2\pi}\int_{-1}^{1} \frac{V(x)}{\sqrt{1-x^{2}}}dx + \int_{-1}^{1}W(x)\psi(x)\sqrt{\frac{1-x}{1+x}}dx  \\
& \hspace{1cm} + \sum_{j=0}^{m} \frac{\alpha_{j}}{2}V(t_{j}) + \sum_{j=1}^{m} \pi i \beta_{j} \left( 1 - 2 \int_{t_{j}}^{1}\psi(x)\sqrt{\frac{1-x}{1+x}}dx \right) , \nonumber \\
& C_{3} = - \frac{1}{6} + \frac{\alpha_{0}^{2}}{2} + \sum_{j=1}^{m} \bigg( \frac{\alpha_{j}^{2}}{4} - \beta_{j}^{2} \bigg), \\
& C_{4} = 2\zeta^{\prime}(-1) - \frac{1-4\alpha_{0}^{2}}{8}\log \left( \pi\psi(-1) \right) - \frac{1}{24}\log \left( \pi\psi(1) \right) + \sum_{j=1}^{m} \bigg(\frac{\alpha_{j}^{2}}{4}-\beta_{j}^{2}\bigg)\log\left(\pi\psi(t_{j})\right) \nonumber \\
& \hspace{0.8cm} + \frac{\alpha_{0}}{2}\log(2\pi) + \sum_{0\leq j < k \leq m} \Bigg[ \log  \Bigg(  \frac{\big(1-t_{j}t_{k}-\sqrt{(1-t_{\smash{j}}^{2})(1-t_{k}^{2})}\big)^{2\beta_{j}\beta{k}}}{2^{\frac{\alpha_{j}\alpha_{k}}{2}}|t_{j}-t_{k}|^{\frac{\alpha_{j}\alpha_{k}}{2} + 2\beta_{j}\beta_{k}}} \Bigg) + \frac{i\pi}{2} (\alpha_{k}\beta_{j}-\alpha_{j}\beta_{k})  \Bigg] \nonumber \\
& \hspace{0.8cm} + \sum_{j=1}^{m} \left( \frac{\alpha_{j}^{2}}{4} \log \sqrt{\frac{1-t_{\smash{j}}}{1+t_{\smash{j}}}} - \beta_{j}^{2} \log \Big( 4(1-t_{j})^{3/2}(1+t_{j})^{1/2} \Big) \right) + \mathcal{A}\sum_{j=1}^{m} i \beta_{j}\arcsin t_{j}  \nonumber \\
&  \hspace{0.8cm}- \log G (1+\alpha_{0})+\sum_{j=1}^{m} \log \frac{G(1+\frac{\alpha_{j}}{2}+\beta_{j})G(1+\frac{\alpha_{j}}{2}-\beta_{j})}{G(1+\alpha_{j})}   \\
& \hspace{0.8cm}+ \frac{\mathcal{A}}{2\pi} \int_{-1}^{1} \frac{W(x)}{\sqrt{1-x^{2}}}dx - \sum_{j=0}^{m} \frac{\alpha_{j}}{2}W(t_{j}) +\sum_{j=1}^{m}  \frac{i\beta_{j}}{\pi} \sqrt{1-t_{\smash{j}}^{2}} \dashint_{-1}^{1} \frac{W(x)}{\sqrt{1-x^{2}}(t_{j}-x)}dx  \nonumber \\
& \hspace{0.8cm} + \frac{1}{4\pi^{2}}\int_{-1}^{1}  \frac{W(x)}{\sqrt{1-x^{2}}} \bigg(\dashint_{-1}^{1} \frac{W^{\prime}(y)\sqrt{1-y^{2}}}{x-y}dy \bigg) dx, \nonumber
\end{align}
where $G$ is Barnes' $G$-function, $\zeta$ is Riemann's zeta-function, where we use the notations $\dashint$ for the Cauchy principal value integral, and
\begin{equation}
\mathcal{A} = \sum_{j=0}^{m} \alpha_{j}.
\end{equation}
Furthermore, the error term in \eqref{asymp thm Ln} is uniform for all $\alpha_{k}$ in compact subsets of \\$\{ z \in \mathbb{C}: \Re z >-1 \}$, for all $\beta_{k}$ in compact subsets of  $\{ z \in \mathbb{C}: \Re z \in \big( \frac{-1}{4},\frac{1}{4} \big) \}$, and uniform in $t_{1},\ldots,t_{m}$, as long as there exists $\delta > 0$ independent of $n$ such that
\begin{equation}
\min_{j\neq k}\{ |t_{j}-t_{k}|,|t_{j}-1|,|t_{j}+1|\} \geq \delta.
\end{equation}
\end{theorem}

\newpage
\begin{theorem}[for Jacobi-type weights]\label{theorem J}\ \\
Let $m \in \mathbb{N}$, and let $t_{j}$, $\alpha_{j}$ and $\beta_{j}$ be such that
\begin{equation*}
-1 = t_{0} < t_{1} < \ldots < t_{m} < t_{m+1} = 1, \quad \mbox{ and } \quad \Re \alpha_{j} > -1, \quad \Re \beta_{j} \in (-\tfrac{1}{4},\tfrac{1}{4}) \quad \mbox{ for } j=0,\ldots,m+1,
\end{equation*}
with $\beta_{0} = 0 = \beta_{m+1}$. Let $V$ be a one-cut regular potential whose equilibrium measure is supported on $[-1,1]$ with density $\frac{\psi(x)}{\sqrt{1-x^{2}}}$, and let $W: [-1,1]\to\mathbb{R}$ be analytic in a neighbourhood of $[-1,1]$. \\ As $n \to \infty$, we have
\begin{equation}\label{asymp thm Jn}
J_{n}(\vec{\alpha},\vec{\beta},V,W) = \exp\left(C_{1} n^{2} + C_{2} n + C_{3} \log n + C_{4} + \bigO \Big( \frac{\log n}{n^{1-4\beta_{\max}}} \Big)\right),
\end{equation}
with $\beta_{\max} = \max \{ |\Re \beta_{1}|,\ldots,|\Re \beta_{m}| \}$ and 
\begin{align}
& C_{1} = - \log 2 - \frac{1}{2} \int_{-1}^{1} V(x)\left( \frac{1}{\pi} + \psi(x) \right)\frac{dx}{\sqrt{1-x^{2}}}, \\
& C_{2} = \log(2\pi) - \mathcal{A}\log 2 - \frac{\mathcal{A}}{2\pi}\int_{-1}^{1} \frac{V(x)}{\sqrt{1-x^{2}}}dx + \int_{-1}^{1}W(x)\frac{\psi(x)}{\sqrt{1-x^{2}}}dx  \\
& \hspace{1cm} + \sum_{j=0}^{m+1} \frac{\alpha_{j}}{2} V(t_{j}) + \sum_{j=1}^{m} \pi i \beta_{j} \left( 1 - 2 \int_{t_{j}}^{1}\frac{\psi(x)}{\sqrt{1-x^{2}}} dx \right) , \nonumber \\
& C_{3} = - \frac{1}{4} + \frac{\alpha_{0}^{2}+\alpha_{m+1}^{2}}{2} + \sum_{j=1}^{m} \bigg( \frac{\alpha_{j}^{2}}{4} - \beta_{j}^{2} \bigg), \\
& C_{4} = 3\zeta^{\prime}(-1) + \frac{\log 2}{12} - \frac{1 \hspace{-0.05cm}-\hspace{-0.04cm}4\alpha_{0}^{2}}{8}\log \left( \pi\psi(-1) \right) - \frac{1\hspace{-0.05cm}-\hspace{-0.04cm}4\alpha_{m+1}^{2}}{8}\log \left( \pi\psi(1) \right) + \sum_{j=1}^{m} \bigg(\frac{\alpha_{j}^{2}}{4}-\beta_{j}^{2}\bigg)\log\left(\pi\psi(t_{j})\right) \nonumber \\
& \hspace{0.5cm} +  \frac{\alpha_{0}+\alpha_{m+1}}{2}\log(2\pi) + \hspace{-0.15cm} \sum_{0\leq j < k \leq m+1} \hspace{-0.1cm} \Bigg[ \hspace{-0.1cm} \log  \Bigg(  \frac{\big(1-t_{j}t_{k}-\sqrt{(1-t_{\smash{j}}^{2})(1-t_{k}^{2})}\big)^{2\beta_{j}\beta{k}}}{2^{\frac{\alpha_{j}\alpha_{k}}{2}}|t_{j}-t_{k}|^{\frac{\alpha_{j}\alpha_{k}}{2} + 2\beta_{j}\beta_{k}}} \Bigg) + \frac{i\pi}{2} (\alpha_{k}\beta_{j}-\alpha_{j}\beta_{k})  \Bigg] \nonumber \\
& \hspace{0.5cm} + \sum_{j=1}^{m} \left( \frac{\alpha_{j}^{2}}{4} \log \frac{1}{\sqrt{1-t_{\smash{j}}^{2}}} - \beta_{j}^{2} \log \Big( 4 \sqrt{1-t_{\smash{j}}^{2}} \Big) \right) + \mathcal{A}\sum_{j=1}^{m} i \beta_{j}\arcsin t_{j} - \frac{\al^2_0+\al^2_{m+1}}{2}\log 2  \nonumber \\
&  \hspace{0.5cm}- \log G (1+\alpha_{0}) - \log G (1+\alpha_{m+1}) +\sum_{j=1}^{m} \log \frac{G(1+\frac{\alpha_{j}}{2}+\beta_{j})G(1+\frac{\alpha_{j}}{2}-\beta_{j})}{G(1+\alpha_{j})}   \\
& \hspace{0.5cm}+ \frac{\mathcal{A}}{2\pi} \int_{-1}^{1} \frac{W(x)}{\sqrt{1-x^{2}}}dx - \sum_{j=0}^{m+1} \frac{\alpha_{j}}{2}W(t_{j}) +\sum_{j=1}^{m}  \frac{i\beta_{j}}{\pi} \sqrt{1-t_{\smash{j}}^{2}} \dashint_{-1}^{1} \frac{W(x)}{\sqrt{1-x^{2}}(t_{j}-x)}dx  \nonumber \\
& \hspace{0.5cm} + \frac{1}{4\pi^{2}}\int_{-1}^{1}  \frac{W(x)}{\sqrt{1-x^{2}}} \bigg(\dashint_{-1}^{1} \frac{W^{\prime}(y)\sqrt{1-y^{2}}}{x-y}dy \bigg) dx, \nonumber
\end{align}
where $G$ is Barnes' $G$-function, $\zeta$ is Riemann's zeta-function, where we use the notations $\dashint$ for the Cauchy principal value integral, and
\begin{equation}
\mathcal{A} = \sum_{j=0}^{m+1} \alpha_{j}.
\end{equation}
Furthermore, the error term in \eqref{asymp thm Jn} is uniform for all $\alpha_{k}$ in compact subsets of \\$\{ z \in \mathbb{C}: \Re z >-1 \}$, for all $\beta_{k}$ in compact subsets of  $\{ z \in \mathbb{C}: \Re z \in \big( \frac{-1}{4},\frac{1}{4} \big) \}$, and uniform in $t_{1},\ldots,t_{m}$, as long as there exists $\delta > 0$ independent of $n$ such that
\begin{equation}
\min_{j\neq k}\{ |t_{j}-t_{k}|,|t_{j}-1|,|t_{j}+1|\} \geq \delta.
\end{equation}
\end{theorem}
\begin{remark}\label{remark real part beta}
The assumption $\Re \beta_{k} \in (-\frac{1}{4},\frac{1}{4})$ comes from some technicalities in our analysis. Similar difficulties were encountered in \cite{ItsKrasovsky} for $G_{n}(\vec{0},\vec{\beta},2x^{2},0)$ with $m = 1$ (i.e. $\vec{\beta} = \beta_{1}$), and in \cite{DeiftItsKrasovsky} for $J_{n}(\vec{\alpha},\vec{\beta},0,W)$. In \cite{DeiftItsKrasovsky}, the authors overcame these technicalities, and were able to extend their results from $\Re \beta_{k} \in (-\frac{1}{4},\frac{1}{4})$ to $\Re \beta_{k} \in (-\frac{1}{2},\frac{1}{2})$ by using Vitali's theorem. Their argument relies crucially on $w$ being independent of $n$ (which is true only for Jacobi-type weights with $V = 0$) and cannot be adapted straightforwardly to the situation of Theorems \ref{theorem G}, \ref{theorem L} and \ref{theorem J}. However, the method presented in this paper allows in principle, but with significant extra effort, to obtain asymptotics for the whole region $\Re \beta_{k} \in (-\frac{1}{2},\frac{1}{2})$. Finally, extending the result from $\Re \beta_{k} \in (-\frac{1}{2},\frac{1}{2})$ to $\Re \beta_{k} \in (-\frac{1}{2},\frac{1}{2}]$ would rely on so-called FH representations of the weight, see \cite{DIK} for more details.
\end{remark}

\begin{remark}\label{remark: Toeplitz and Hankel}
Starting with a function $f$ defined on the unit circle, the associated Toeplitz determinant is given by
\begin{equation}
\det \left( \frac{1}{2\pi}\int_{-\pi}^{\pi}f(e^{i\theta}) e^{-i(j-k)\theta}d\theta\right)_{j,k=0,\ldots,n-1}.
\end{equation}
Asymptotics of large Toeplitz determinants is another topic of high interest, which presents applications similar to those of Hankel determinants, but for point processes defined on the unit circle instead of the real line. In \cite{DIK}, the authors obtained first large $n$ asymptotics for certain Toeplitz determinants (with the zero potential), and deduced from them large $n$ asymptotics for $J_{n}(\vec{\alpha},\vec{\beta},0,W)$. It is therefore natural to wonder if one can translate the results of Theorems \ref{theorem G}, \ref{theorem L} and \ref{theorem J} into asymptotics for Toeplitz determinants with a one-cut regular potential. We explain here why we believe this is not obvious. 

\vspace{0.3cm}\hspace{-0.55cm}The main tool used in \cite{DIK} is a relation of Szeg\"{o} from \cite{Szego OP}. If 
\begin{equation}\label{transformation symbol weight}
f(e^{i\theta}) = w(\cos \theta) |\sin \theta |,
\end{equation}
we can express orthogonal polynomials on the unit circle associated to $f$ in terms of orthogonal polynomials on the real line associated to $w$. Note that this transformation can only work in all generality from Toeplitz to Hankel, and not the other way around. Indeed, the weight $w$ can be arbitrary, but the function $f$ is of a very particular type (in particular it satisfies $f(e^{i\theta}) = f(e^{-i\theta})$). 

\end{remark}
\begin{remark}\label{Remark: ForFra}
Using Theorem \ref{theorem L}, we find that
\begin{multline}\label{expectation}
\frac{L_{n}(\vec{\alpha}_{0},\vec{0},V,W)}{L_{n}(\vec{\alpha}_{0},\vec{0},V,0)} = \exp\bigg( n \int_{-1}^{1}W(x)\psi(x)\sqrt{\frac{1-x}{1+x}}dx + \frac{1}{4\pi^{2}} \int_{-1}^{1}dx \frac{W(x)}{\sqrt{1-x^{2}}}\int_{-1}^{1}dy \frac{W'(y)\sqrt{1-y^{2}}}{x-y} \bigg) \\ \times \exp \bigg( \frac{\alpha_{0}}{2\pi}\int_{-1}^{1} \frac{W(x)-W(-1)}{\sqrt{1-x^{2}}}dx + o(1)\bigg),
\end{multline}
where $\vec{\alpha}_{0}:=(\alpha_{0},0,\ldots,0)$. These asymptotics, when specialized to $V=2(x+1)$, are consistent with \cite[Proposition 3.9]{FRW2017}. For general $V$ and $W$, \eqref{expectation} shows that a small correction is needed to \cite[Conjecture 9]{ForFra} if $\alpha_{0} \neq 0$. 
\end{remark}
\section{Applications}\label{subsection:Applications}
In this section, we provide several applications of Theorems \ref{theorem G}, \ref{theorem L} and \ref{theorem J} in random matrix theory. For each type of weight, there corresponds a particular type of matrix ensemble. Assume that $V$ is the potential of a Gaussian-type weight. The associated Gaussian-type matrix ensemble consists of the space of $n \times n$ complex Hermitian matrices endowed with the probability measure
\begin{equation}
\frac{1}{\widehat{Z}_{n}^{G}}e^{-n \mathrm{Tr}(V(M))}dM, \qquad dM = \prod_{i=1}^{n} dM_{ii} \prod_{1 \leq i < j \leq n} d \Re M_{ij} d \Im M_{ij},
\end{equation}
with $\widehat{Z}_{n}^{G}$ the normalizing constant. Laguerre-type matrix ensembles are usually defined on $n \times n$ complex positive definite Hermitian matrices. Here we instead assume, for Laguerre-type matrix ensembles, that all matrices have eigenvalues greater than $-1$ (this assumption eases the comparison between the three cases). Such ensembles have a probability measure of the form
\begin{equation}
\frac{1}{\widehat{Z}_{n}^{L}}\det(I+M)^{\alpha_{0}}e^{-n \mathrm{Tr}(V(M))}dM, \qquad \alpha_{0}>-1,
\end{equation}
where $V$ is of Laguerre-type, and $\widehat{Z}_{n}^{L}$ is the normalizing constant. Finally, a Jacobi-type matrix ensemble consists of the space of $n \times n$ Hermitian matrices whose spectrum lies the interval $[-1,1]$, with a probability measure of the form
\begin{equation}
\frac{1}{\widehat{Z}_{n}^{J}}\det(I+M)^{\alpha_{0}}\det(I-M)^{\alpha_{m+1}}e^{-n \mathrm{Tr}(V(M))}dM, \qquad \alpha_{0}, \alpha_{m+1}>-1,
\end{equation}
with a Jacobi-type potential $V$ and $\widehat{Z}_{n}^{J}$ is again the normalizing constant. These three types of matrix ensembles are invariant under unitary conjugation and induce the following probability measures on the eigenvalues $x_{1},\ldots,x_{n}$:
\begin{align}
& \frac{1}{Z_{n}^{G}} \prod_{1 \leq j < k \leq n}(x_{k}-x_{j})^{2} \prod_{j=1}^{n}e^{-nV(x_{j})}dx_{j}, & & x_{1},\ldots,x_{n} \in \mathbb{R}, \label{Gaussian distribution} \\
& \frac{1}{Z_{n}^{L}} \prod_{1 \leq j < k \leq n}(x_{k}-x_{j})^{2} \prod_{j=1}^{n}(1+x_{j})^{\alpha_{0}}e^{-nV(x_{j})}dx_{j}, & & x_{1},\ldots,x_{n} \in [-1,\infty), \label{Laguerre distribution} \\
& \frac{1}{Z_{n}^{J}} \prod_{1 \leq j < k \leq n}(x_{k}-x_{j})^{2} \prod_{j=1}^{n}(1+x_{j})^{\alpha_{0}}(1-x_{j})^{\alpha_{m+1}}e^{-nV(x_{j})}dx_{j},& & x_{1},\ldots,x_{n} \in[-1,1], \label{Jacobi distribution}
\end{align}
where the first, second and third line read for Gaussian, Laguerre, and Jacobi-type matrix ensembles, respectively, and $Z_{n}^{G}$, $Z_{n}^{L}$ and $Z_{n}^{J}$ are the normalizing constants, also called the partition functions.

\vspace{-0.2cm}\paragraph{Partition function asymptotics in the one-cut regime.} By Heine's formula, the partition functions can be rewritten as Hankel determinants of the form \eqref{Hankel introduction} with $W = 0$, $\vec{\beta} = \vec{0}$ and $\alpha_{1}=...=\alpha_{m}=0$ and thus their large $n$ asymptotics can be deduced from Theorems \ref{theorem G}, \ref{theorem L} and \ref{theorem J}. The existence of an all orders expansion as $n \to + \infty$ for $Z_{n}^{G}$ has been proved in some particular cases of $V$ in \cite{BleIts, ErcMcL} using RH methods, and then for $Z_{n}^{G}$, $Z_{n}^{L}$ and $Z_{n}^{J}$ in \cite{BorGui} using loop equations (under the extra condition that $\alpha_{0} = \alpha_{m+1} = 0$). The large $n$ asymptotics for $Z_{n}^{G}$ have been obtained only recently in \cite{BerWebbWong} for general one-cut Gaussian-type potential $V$, and for convenience we rewrite them in \eqref{partition Gaussian} below. To the best of our knowledge, the large $n$ asymptotics of $Z_{n}^{L}$ and $Z_{n}^{J}$ written in Corollary \ref{coro: partition functions} below are already new results. 
\begin{corollary}\label{coro: partition functions}
As $n \to + \infty$, we have
\begin{align}
& Z_{n}^{G} = & & \exp \bigg(- \left( \log 2 + \frac{3}{4} + \frac{1}{2} \int_{-1}^{1}(V(x)-2x^{2})\Big( \frac{2}{\pi} + \psi(x) \Big) \sqrt{1-x^{2}}dx\right)n^{2} \label{partition Gaussian} \\
&     & & \hspace{-0.4cm}+ \log(2\pi) n - \frac{1}{12}\log n + \zeta^{\prime}(-1) - \frac{1}{24}\log \left( \frac{\pi}{2}\psi(-1) \right) - \frac{1}{24}\log \left( \frac{\pi}{2}\psi(1) \right) + \bigO \Big( \frac{\log n}{n}\Big)\bigg), \nonumber \\
& Z_{n}^{L} = & & \exp \bigg( - \bigg( \log 2 + \frac{3}{2} + \frac{1}{2}\int_{-1}^{1}(V(x)-2(x+1))\Big( \frac{1}{\pi}+\psi(x) \Big)\sqrt{\frac{1-x}{1+x}}dx \bigg)n^{2} \\
& & & \hspace{-0.4cm}+ \bigg( \log(2\pi) - \alpha_{0} \log 2 - \frac{\alpha_{0}}{2\pi} \int_{-1}^{1} \frac{V(x)-2(x+1)}{\sqrt{1-x^{2}}}dx  + \frac{\alpha_{0}}{2}(V(-1)-2) \bigg)n +\Big(\frac{\alpha_{0}^{2}}{2}-\frac{1}{6} \Big) \log n \nonumber \\
& & &  \hspace{-0.4cm}+ 2 \zeta^{\prime}(-1) - \frac{1 \hspace{-0.05cm}-\hspace{-0.04cm}4\alpha_{0}^{2}}{8}\log \big( \pi \psi(-1) \big) - \frac{1}{24}\log(\pi \psi(1) + \frac{\alpha_{0}}{2}\log(2\pi) - \log G(1+\alpha_{0}) + \bigO \Big( \frac{\log n}{n}\Big)\bigg), \nonumber \\
& Z_{n}^{J} = & & \exp \bigg( -\bigg( \log 2 + \frac{1}{2}\int_{-1}^{1} V(x) \Big( \frac{1}{\pi}+\psi(x) \Big) \frac{dx}{\sqrt{1-x^{2}}} \bigg)n^{2} \\
& & & \hspace{-0.4cm}+ \hspace{-0.05cm} \bigg( \hspace{-0.05cm} \log(2\pi) - (\alpha_{0}+\alpha_{m+1})\log 2 - \frac{\alpha_{0}+\alpha_{m+1}}{2\pi} \int_{-1}^{1} \frac{V(x)}{\sqrt{1-x^{2}}}dx + \frac{\alpha_{0}}{2}V(-1) + \frac{\alpha_{m+1}}{2}V(1) \bigg)n \nonumber \\
& & & \hspace{-0.4cm}+ \bigg( -\frac{1}{4}+\frac{\alpha_{0}^{2}+\alpha_{m+1}^{2}}{2} \bigg) \log n + 3 \zeta^{\prime}(-1) + \frac{\log 2}{12} - \frac{1 \hspace{-0.05cm}-\hspace{-0.04cm}4\alpha_{0}^{2}}{8}\log \left( \pi\psi(-1) \right) - \frac{1\hspace{-0.05cm}-\hspace{-0.04cm}4\alpha_{m+1}^{2}}{8}\log \left( \pi\psi(1) \right) \nonumber \\
& & &  \hspace{-0.4cm} + \frac{\alpha_{0}+\alpha_{m+1}}{2}\log(2\pi) - \frac{(\alpha_{0}+\alpha_{m+1})^{2}}{2}\log 2 - \log \big( G(1+\alpha_{0})G(1+\alpha_{m+1}) \big) + \bigO \Big( \frac{\log n}{n}\Big)\bigg). \nonumber
\end{align}
\end{corollary} 
\paragraph{Central limit theorems (CLTs).} The function $W$ allows to obtain information about
the global fluctuation properties of the spectrum around the equilibrium measure. In \cite{Johansson2}, Johansson obtained a CLT for Gaussian-type ensembles (and is reproduced in \eqref{CLT johansson} below for convenience). All three types of ensembles considered here fall into the class of ``polynomial ensembles" considered in \cite{BD CLT JAMS}. The general theorem of \cite[Theorem 2.5]{BD CLT JAMS} states that, if the recurrence coefficients for the associated orthogonal polynomials have a limit, then we have
\begin{align*}
\sum_{i=1}^{n} W(x_{i})- \mathbb{E}\bigg[ \sum_{i=1}^{n} W(x_{i}) \bigg] \quad \overset{d}{\longrightarrow} \quad \mathcal{N}(0,\sigma^{2}), \qquad \mbox{as } n \to + \infty,
\end{align*}
for a certain explicitly computable $\sigma^{2}=\sigma^{2}(W)$, where $\overset{d}{\longrightarrow}$ means convergence in distribution, and $\mathcal{N}(0,\sigma^{2})$ is a zero-mean normal random variable with variance $\sigma^{2}$. Since the asymptotics for the recurrence coefficients of the orthogonal polynomials associated to certain particular Laguerre-type and Jacobi-type weights were obtained in \cite{Vanlessen Laguerre} and \cite{KMcLVAV}, respectively, we can obtain CLTs for these particular cases by combining the results of \cite{BD CLT JAMS}, \cite{Vanlessen Laguerre}, and \cite{KMcLVAV}. In Corollary \ref{coro: CLTs} below, we obtain two CLTs for general Laguerre and Jacobi-type ensembles as a rather straightforward consequence of Theorem \ref{theorem L} and Theorem \ref{theorem J}. We also mention that two CLTs for the Laguerre-type potential $V(x) = 2(x+1)$ and the Jacobi-type potential $V(x)=0$ were obtained by Berezin and Bufetov in \cite{BerezinBufetov} independently and simultaneously to the present paper.
\begin{corollary}\label{coro: CLTs}
(a) Let $x_{1},\ldots,x_{n}$ be distributed according to \eqref{Gaussian distribution} and $V$ and $W$ be as in Theorem \ref{theorem G}. As $n \to +\infty$, we have
\begin{align}\label{CLT johansson}
\sum_{i=1}^{n} W(x_{i}) - n \int_{-1}^{1} W(x) \psi(x) \sqrt{1-x^{2}}dx \quad \overset{d}{\longrightarrow} \quad \mathcal{N}(0,\sigma^{2}),
\end{align}
where
\begin{equation}\label{variance}
\sigma^{2} = \frac{1}{2\pi^{2}}\int_{-1}^{1}  \frac{W(x)}{\sqrt{1-x^{2}}} \bigg(\dashint_{-1}^{1} \frac{W^{\prime}(y)\sqrt{1-y^{2}}}{x-y}dy \bigg) dx.
\end{equation}
(b) Let $x_{1},\ldots,x_{n}$ be distributed according to \eqref{Laguerre distribution} and $V$ and $W$ be as in Theorem \ref{theorem L}. As $n \to +\infty$, we have
\begin{align}
\sum_{i=1}^{n} W(x_{i}) - n \int_{-1}^{1} W(x) \psi(x) \sqrt{\frac{1-x}{1+x}}dx \quad \overset{d}{\longrightarrow} \quad \mathcal{N}(\mu_{L},\sigma^{2}),
\end{align}
where $\sigma^{2}$ is given by \eqref{variance} and the mean $\mu_{L}$ is given by
\begin{equation}
\mu_{L} = \frac{\alpha_{0}}{2\pi} \int_{-1}^{1} \frac{W(x)}{\sqrt{1-x^{2}}}dx - \frac{\alpha_{0}}{2}W(-1).
\end{equation}
(c) Let $x_{1},\ldots,x_{n}$ be distributed according to \eqref{Jacobi distribution} and $V$ and $W$ be as in Theorem \ref{theorem J}. As $n \to +\infty$, we have
\begin{align}
\sum_{i=1}^{n} W(x_{i}) - n \int_{-1}^{1} W(x) \frac{\psi(x)}{\sqrt{1-x^{2}}}dx \quad \overset{d}{\longrightarrow} \quad \mathcal{N}(\mu_{J},\sigma^{2}),
\end{align}
where $\sigma^{2}$ is given by \eqref{variance} and the mean $\mu_{J}$ is given by
\begin{equation}
\mu_{J} = \frac{\alpha_{0}+\alpha_{m+1}}{2\pi} \int_{-1}^{1} \frac{W(x)}{\sqrt{1-x^{2}}}dx - \frac{\alpha_{0}}{2}W(-1)- \frac{\alpha_{m+1}}{2}W(1).
\end{equation}
\end{corollary}
\begin{proof}
We only prove the result for Jacobi-type ensembles. The proofs for the other cases are similar. From Heine's formula, we have
\begin{equation}
\mathbb{E}_{J}\Big[ e^{t\sum_{j=1}^{n}W(x_{j})} \Big] = \frac{J_{n}\big( (\alpha_{0},0,...,0,\alpha_{m+1}),\vec{0},V,tW \big)}{J_{n}\big( (\alpha_{0},0,...,0,\alpha_{m+1}),\vec{0},V,0 \big)}, \qquad t \in \mathbb{R},
\end{equation}
where $\mathbb{E}_{J}$ means that the expectation is taken with respect to \eqref{Jacobi distribution}. Let $X_{n}$ be the random variable defined by
\begin{equation}
X_{n} = \sum_{j=1}^{n}W(x_{j}) - n \int_{-1}^{1} W(x) \frac{\psi(x)}{\sqrt{1-x^{2}}}dx.
\end{equation}
Theorem \ref{theorem J} then implies
\begin{equation}
\mathbb{E}_{J}\Big[e^{t X_{n}} \Big] = \exp \left( t \mu_{J} + \frac{t^{2}}{2}\sigma^{2}+\bigO\Big( \frac{\log n}{n} \Big) \right), \qquad \mbox{as } n \to + \infty.
\end{equation}
Thus, for each $t \in \mathbb{R}$, $(X_{n})$ is a sequence of random variables whose moment generating functions converge to $e^{t \mu_{J} + \frac{t^{2}}{2}\sigma^{2}}$ as $n \to + \infty$ (the convergence is pointwise in $t \in \mathbb{R}$). Convergence in distribution follows from well-known convergence theorems (see e.g. \cite{Billingsley}).
\end{proof}
\paragraph{Correlations of the characteristic polynomials.} Let $p_{n}(t)= \prod_{j=1}^{n} (t-x_{j})$ be the characteristic polynomial associated to a matrix from a Gaussian-type, Laguerre-type or Jacobi-type ensemble. Supported by numerical evidence, numerous conjectures in the literature have been formulated about links between $p_{n}(t)$ and the behavior of the Riemann $\zeta$-functions along the critical line (see e.g. \cite{KeaSna}). For Gaussian-type ensembles, correlations with root-type singularities were studied in \cite{Krasovsky} for $V(x) = 2x^{2}$ and in \cite{BerWebbWong} for general $V$. Large $n$ asymptotics for more general correlations with both root-type and jump-type singularities were obtained in \cite{Charlier}. However, the cases of Laguerre or Jacobi-type ensembles were still open. In the same way as in \cite[equation (1.16)]{Charlier}, we can express these correlations in terms of Hankel determinants with FH singularities as follows\footnote{There is an $n$ missing in \cite[equations (1.16) and (1.22)]{Charlier}: $e^{-i \pi \beta_{k}}$ should instead be $e^{-i n \pi \beta_{k}}$ and $s_{k}^{1/2}$ should instead be $s_{k}^{n/2}$. The correct expressions are given by \eqref{correlation finite n} and \eqref{gap prob thinned} of the present work.}:
\begin{equation}\label{correlation finite n}
\mathbb{E}_{D} \bigg[ \prod_{k=1}^{m} |p_{n}(t_{k})|^{\alpha_{k}}e^{2i \beta_{k} \arg p_{n}(t_{k})} \bigg] = \frac{D_{n}(\vec{\alpha},\vec{\beta},V,0)}{Z_{n}^{D}} \prod_{k=1}^{m}e^{-i n \pi \beta_{k}}, \qquad D = G,L,J,
\end{equation}
where $\mathbb{E}_{G}$, $\mathbb{E}_{L}$ and $\mathbb{E}_{J}$ are the expectations taken with respect to \eqref{Gaussian distribution}, \eqref{Laguerre distribution} and \eqref{Jacobi distribution}, respectively, and where
\begin{equation}
\arg p_{n}(t) = \sum_{j=1}^{n} \arg(t-x_{j}), \qquad \mbox{ with } \qquad \arg(t-x_{j}) = \left\{ \begin{matrix}
0, & \mbox{if } x_{j}<t, \\
-\pi, & \mbox{if } x_{j}>t.
\end{matrix} \right.
\end{equation}
Therefore, as an immediate corollary of Theorem \ref{theorem L} and Theorem \ref{theorem J}, we can obtain large $n$ asymptotics for the correlations given in \eqref{correlation finite n} for Laguerre and Jacobi-type ensembles.

\paragraph*{Gap probabilities in piecewise constant thinned point processes.} Given a point process, a constant thinning consists of removing each point independently with a certain probability $s \in [0,1]$. The remaining points, denoted by $y_{1},\ldots,y_{N}$, form a thinned point process, and can be interpreted in certain applications as observed points \cite{BohigasPato1, BohigasPato2}. Probabilities of observing a large gap in the thinned sine, Airy and Bessel point processes, as well as for thinned eigenvalues of Haar distributed unitary matrices, have been studied in \cite{BDIK}, \cite{BB2018}, \cite{BIP2019} and \cite{ChCl2}, respectively. A more general operation, which was considered in \cite{Charlier} for Gaussian-type ensembles, consists of applying a piecewise constant thinning. Large gap asymptotics for the piecewise constant thinned sine, Airy and Bessel point processes were obtained recently in \cite{ChSine}, \cite{ChCl3} and \cite{ChBessel}, respectively. From Theorem \ref{theorem L} and Theorem \ref{theorem J}, we can deduce large gap asymptotics for (piecewise constant) thinned Laguerre and Jacobi-type ensembles. Following \cite{Charlier}, we consider $\mathcal{K} \subseteq \{1,...,m+1\}$. For each $k \in \mathcal{K}$, we remove each point on $(t_{k-1},t_{k})$ with a probability $s_{k} \in (0,1]$. In the same way as shown in \cite[equations (1.20)--(1.22)]{Charlier}, we can express gap probabilities in the piecewise thinned spectrum of Gaussian, Laguerre and Jacobi-type ensembles as follows:
\begin{equation}\label{gap prob thinned}
\mathbb{P}_{D}\Big(\sharp \{ y_{j} \hspace{-0.05cm} \in \hspace{-0.05cm} \bigcup_{k\in \mathcal{K}} (t_{k-1},t_{k}) \} = 0\Big) =  \frac{D_{n}(\vec{\alpha},\vec{\beta},V,0)}{Z_{n}^{D}}\prod_{k \in \mathcal{K}} s_{k}^{n/2}, \qquad D = G,L,J,
\end{equation}
with $\alpha_{1}=...=\alpha_{m}=0$ and $\vec{\beta} = (\beta_{1},...,\beta_{m})$ given by
\begin{equation}
2i \pi \beta_{j} = \log \left( \frac{\widetilde{s}_{j}}{\widetilde{s}_{j+1}} \right), \qquad \widetilde{s}_{j} = \left\{ \begin{matrix}
s_{j}, & \mbox{if } j \in \mathcal{K}, \\
1, & \mbox{if } j \notin \mathcal{K},
\end{matrix} \right.
\end{equation}
and where again $\mathbb{P}_{G}$, $\mathbb{P}_{L}$ and $\mathbb{P}_{J}$ are probabilities taken with respect to \eqref{Gaussian distribution}, \eqref{Laguerre distribution} and \eqref{Jacobi distribution}, respectively.
\paragraph*{Rigidity and Gaussian multiplicative chaos.} Let us consider a sequence of matrices $M_{n}$ taken from either Gaussian, Laguerre, or Jacobi-type ensembles. As $n \to + \infty$, the logarithm of the characteristic polynomial of $M_{n}$ behaves like a log-correlated field. A fundamental tool in describing some properties of the limiting field is a class of random measures, known as Gaussian multiplicative chaos measures. Roughly speaking, these measures are  exponential of the field, however a precise definition is rather subtle.  This subject was introduced by Kahane in \cite{Kahane}, and we refer to \cite{RhodesVargas} for a recent review. For Gaussian-type ensembles, it is known (from \cite{BerWebbWong}) that a sufficiently small power of the absolute value of the characteristic polynomial converges weakly in
distribution to a Gaussian multiplicative chaos measure. Large $n$ asymptotics for Hankel determinants with root-type singularities provide crucial estimates in the proof. Theorem \ref{theorem L} and Theorem \ref{theorem J} provide similar estimates for Laguerre and Jacobi-type ensembles, which could probably be used to prove analogous results for the Laguerre and Jacobi cases. Another related topic is the study of rigidity, which attempts to answer the question: ``How much can the eigenvalues of a random matrix fluctuate?". For Gaussian-type ensembles, this question has been answered in \cite{ClaeysFahsLambertWebb}. This time, it is large $n$ asymptotics for Hankel determinants with jump-type singularities that are crucial in the analysis. In particular, the proof of \cite{ClaeysFahsLambertWebb} relies heavily on Theorem \ref{theorem G} (with $\vec{\alpha} = \vec{0}$). Theorem \ref{theorem L} and Theorem \ref{theorem J} provide similar estimates for Laguerre and Jacobi-type ensembles, which we believe are relevant to prove similar rigidity results for these ensembles.

\subsection*{Outline}
The general strategy of our proof is close to the one done in \cite{Charlier} (itself inspired by the earlier works \cite{Krasovsky,ItsKrasovsky,DIK,BerWebbWong}), and can be schematized as
\begin{equation}\label{lol4}
\begin{array}{l l l l l l l}
L_{n}(\vec{0},\vec{0},2(x+1),0) & \mapsto & L_{n}(\vec{\alpha},\vec{\beta},2(x+1),0) & \mapsto & L_{n}(\vec{\alpha},\vec{\beta},V,0) & \mapsto & L_{n}(\vec{\alpha},\vec{\beta},V,W), \\
 & & J_{n}(\vec{\alpha},\vec{\beta},0,0) &   \mapsto & J_{n}(\vec{\alpha},\vec{\beta},V,0) &  \mapsto & J_{n}(\vec{\alpha},\vec{\beta},V,W).
\end{array}
\end{equation}
In Section \ref{Section:OP and Y}, we recall a well-known correspondence between Hankel determinants and orthogonal polynomials (OPs), and the characterization of these OPs in terms of a Riemann-Hilbert (RH) problem found by Fokas, Its and Kitaev \cite{FokasItsKitaev}, and whose solution is denoted by $Y$. In Section \ref{Section:diff identities LUE}, we derive suitable differential identities, which express the quantities
\begin{equation}\label{lol2}
\begin{array}{l l l}
\partial_{\nu}\log L_{n}(\vec{\alpha},\vec{\beta},2(x+1),0), & \qquad \partial_{s}\log L_{n}(\vec{\alpha},\vec{\beta},V_{s},0), & \qquad \partial_{t}\log L_{n}(\vec{\alpha},\vec{\beta},V,W_{t}), \\
& \qquad \partial_{s}\log J_{n}(\vec{\alpha},\vec{\beta},V_{s},0), & \qquad \partial_{t}\log J_{n}(\vec{\alpha},\vec{\beta},V,W_{t}),
\end{array}
\end{equation}
in terms of $Y$, where $\nu \in \{\alpha_{0},\ldots,\alpha_{m},\beta_{1},\ldots,\beta_{m}\}$, and $s \in [0,1]$ and $t \in [0,1]$ are smooth deformation parameters (more details on these deformations are given in Section \ref{Section: integration in V} and Section \ref{Section: integration in W}). In Section \ref{Section: steepest descent}, we perform a Deift/Zhou steepest descent analysis of the RH problem to obtain large $n$ asymptotics for $Y$. We deduce from them asymptotics for the log derivatives given in \eqref{lol2}, and we also proceed with their successive integrations (represented schematically by an arrow in \eqref{lol4}). These computations are rather long, and we organise them in several sections: Section \ref{Section: FH integration} is devoted to integration in $\vec{\alpha}$ and $\vec{\beta}$, Section \ref{Section: integration in V} to integration in $s$ and Section \ref{Section: integration in W} to integration in $t$. We mention that the method of successive integrations with respect to the parameters of singularities $\alpha_{j}$'s to obtain the asymptotics of a Hankel determinant originated in \cite{Krasovsky}. Each integration only gives us asymptotics for a ratio of Hankel determinants. Therefore, it is important to chose carefully the starting point of integration in the set of parameters $(\vec{\alpha},\vec{\beta},V,W)$. For Laguerre-type weights, we chose this point to be $(\vec{0},\vec{0},2(x+1),0)$ and for Jacobi-type weights, we use the result of \cite{DIK} and chose $(\vec{\alpha},\vec{\beta},0,0)$. We recall the large $n$ asymptotics for $L_{n}(\vec{0},\vec{0},2(x+1),0)$ and for $J_{n}(\vec{\alpha},\vec{\beta},0,0)$ in Section \ref{Section: starting points of integrations}.

\paragraph{Notations.} We will use repetitively through the paper the convention $t_{0} = -1$, $t_{m+1} = 1$, $\beta_{0} = 0$ and $\beta_{m+1} = 0$. Furthermore, for Laguerre-type weights, we define $\alpha_{m+1} = 0$ and for Gaussian-type weights, we define $\alpha_{0} = 0$ and $\alpha_{m+1} = 0$. This allows us for example to rewrite $\omega$ given in \eqref{weight FH} as
\begin{equation}
\omega(x) = \prod_{j=0}^{m+1} \omega_{\alpha_{j}}(x)\omega_{\beta_{j}}(x).
\end{equation}

\section{A Riemann-Hilbert problem for orthogonal polynomials}
\label{Section:OP and Y}

We consider the family of OPs associated to the weight $w$ given in \eqref{weight introduction}. The degree $k$ polynomial $p_{k}$ is characterized by the relations
\begin{equation}\label{orthogonality_relation}
\int_{\mathcal{I}} p_k(x)x^{j}w(x)dx = \kappa_{k}^{-1}\de_{jk}, \qquad j=0,1,2,\ldots,k,
\end{equation}
where $\kappa_{k} \neq 0$ is the leading order coefficient of $p_{k}$. If $\be_j \in i\R$ and $\Re \al_j > -1, \ j=0,\ldots, m+1$, then $w(x) > 0$ for almost all $x \in \mathcal{I}$. In this case, we can rewrite \eqref{orthogonality_relation} as an inner product and it is a simple consequence of Gram-Schmidt that the OPs exist. However, for general values of $\alpha_{j}$ and $\beta_{j}$, the weight $w$ is complex-valued and existence is no more guaranteed. This fact introduces some technicalities in the analysis that are briefly discussed in Section \ref{Section: FH integration}, Section \ref{Section: integration in V} and Section \ref{Section: integration in W}. 

\vspace{0.3cm}\hspace{-0.55cm}We associate to these OPs a RH problem for a $2 \hspace{-0.08cm} \times \hspace{-0.08cm} 2$ matrix-valued function $Y$, due to \cite{FokasItsKitaev}. As mentioned in the outline, it will play a crucial role in our proof.

\subsubsection*{RH problem for $Y$}
\begin{itemize}
\item[(a)] $Y: \C \setminus \mathcal{I} \to \C^{2 \times 2}$ is analytic.
\item[(b)] The limits of $Y(z)$ as $z$ tends to $x \in \mathcal{I} \setminus \{-1,t_1,\ldots, t_m,1\}$ from the upper and lower half plane exist, and are denoted $Y_{\pm}(x)$ respectively. Furthermore, the functions $x \mapsto Y_{\pm}(x)$ are continuous on $\mathcal{I} \setminus \{-1,t_1,\ldots, t_m,1\}$ and are related by
\begin{equation}\label{Y_Jump}
Y_+(x)=Y_-(x) \begin{pmatrix}
1 & w(x) \\
0 & 1
\end{pmatrix},
\qquad x \in \mathcal{I} \setminus \{-1,t_1,\ldots, t_m,1\}.
\end{equation} 
\item[(c)] As $z \to \infty$,  \begin{equation}\label{Y_Asymptotics_infty}
Y(z)=\big( I + \mathcal{O}(z^{-1}) \big) z^{n \sigma_3} , \quad \mbox{where} \quad \sigma_3 = \begin{pmatrix}
1 & 0 \\
0 & -1
\end{pmatrix}.
\end{equation}
\item[(d)] As $z \to t_j$, for $j=0,1,\ldots,m+1$ (with $t_{0}:=-1$ and $t_{m+1}:=1$), we have
\begin{equation}\label{Y_Asymptotics_tj}
Y(z)=\begin{cases}\begin{pmatrix}
\mathcal{O}(1) & \mathcal{O}(1) +\mathcal{O}((z-t_j)^{\al_j})\\
\mathcal{O}(1) & \mathcal{O}(1) + \mathcal{O}((z-t_j)^{\al_j})\end{pmatrix}, & \mbox{if } \Re \al_j \neq 0,  \\
\begin{pmatrix}
\mathcal{O}(1) & \mathcal{O}(\log(z-t_j))\\
\mathcal{O}(1) & \mathcal{O}(\log(z-t_j))\end{pmatrix},  & \mbox{if } \Re \al_j = 0.
\end{cases}
\end{equation}
\end{itemize}
The solution of the RH problem for $Y$ is always unique, exists if and only if $p_{n}$ and $p_{n-1}$ exist, and is explicitly given by
\begin{equation}\label{OPsolution}
Y(z) = \begin{pmatrix}
\di \ka^{-1}_n p_n(z) & \di  \frac{\ka^{-1}_n}{2\pi i} \int_{\mathcal{I}} \frac{p_n(x)w(x)}{x-z}dx \\[0.35cm]
\di	-2\pi i \ka_{n-1} p_{n-1}(z) & \di -\ka_{n-1} \int_{\mathcal{I}} \frac{p_{n-1}(x)w(x)}{x-z}dx
\end{pmatrix}.
\end{equation}
The fact that $Y$ given by \eqref{OPsolution} satisfies the condition (b) of the RH problem for $Y$ follows from the Sokhotski formula and relies on the assumption that $W$ is locally H\"{o}lder  continuous on $\mathcal{I}$ (see e.g. \cite{Gakhov}).

\section{Differential identities}\label{Section:diff identities LUE}

In this section, we express the logarithmic derivatives given in \eqref{lol2} in terms of $Y$.

\subsection{Identity for $\partial_{\nu}\log L_{n}(\vec{\alpha},\vec{\beta},2(x+1),0)$ with $\nu \in \{\alpha_{0},\ldots,\alpha_{m},\beta_{1},\ldots,\beta_{m}\}$}
In this subsection, we specialize to the Laguerre-type weight $w(x) = \omega(x)e^{-2n(x+1)}$.

\medskip Note that the second column of $Y$ blows up as $z \to t_{k}$, $k=0,1,\ldots,m$ as shown in \eqref{Y_Asymptotics_tj}. The terms of order $1$ in these asymptotics will contribute in our identity for $\partial_{\nu}\log L_{n}(\vec{\alpha},\vec{\beta},2(x+1),0)$. To prepare ourselves for that matter, following \cite[eq (3.6)]{Charlier}, for each $k \in \{1,\ldots,m\}$ we define a regularized integral by
\begin{equation}\label{Reg_k}
\mbox{Reg}_k(f) = \lim_{\ep \to 0_{+}} \left[ \al_k \int_{\mathcal{I} \setminus [t_{k}-\epsilon,t_{k}+\epsilon]} \frac{f(x)\om(x)}{x-t_k}dx - f(t_k)\om_{t_k}(t_k)(e^{\pi i \be_k} - e^{-\pi i \be_k})\ep^{\al_k} \right],
\end{equation}
where $\mathcal{I}=[-1,+\infty)$, $f:\mathcal{I}\to \mathbb{C}$ is smooth and such that $f\omega$ is integrable on $\mathcal{I}$, and
\begin{equation}
\omega_{t_{k}}(x) = \prod_{\substack{0 \leq j \leq m \\ j \neq k }} \omega_{\alpha_{j}}(x)\omega_{\beta_{j}}(x).
\end{equation}

\vspace{-0.2cm}\hspace{-0.55cm}For $k = 0$, we define the regularized integral as above, with $e^{\pi i \beta_{k}}$ replaced by $0$ and $e^{-\pi i \beta_{k}}$ replaced by $1$ (we also recall that $t_{0} = -1$), i.e. we have
\begin{equation}\label{Reg_0}
\mbox{Reg}_0(f) := \lim_{\ep \to 0_{+}} \left[ \al_0 \int_{\mathcal{I} \setminus [t_{0},t_{0}+\epsilon]} \frac{f(x)\om(x)}{x-t_0}dx + f(t_0)\om_{-1}(t_0)\ep^{\al_0} \right].
\end{equation}
	
\begin{proposition}\label{reg-integrals}
Let $k \in \{0,1,\ldots,m\}$, and let $f: \mathcal{I}\to \mathbb{C}$ be a smooth function, analytic in a neighborhood of $t_{k}$, and such that $f \omega$ is integrable on $\mathcal{I}$. We have
\begin{equation}\label{Reg_k versus the Cauchy integral}
\mbox{Reg}_k(f)  =\lim_{z \to t_k} \al_k \int_{\mathcal{I}} \frac{f(x)\om(x)}{x-z}dx-\mathcal{J}_k(z)
\end{equation}
where the limit is taken along a path in the upper-half plane which is non-tangential to the real line. If $k \in \{ 1,\ldots,m\}$, $\mathcal{J}_k(z)$ is given by 
\begin{equation}\label{J_k}
\mathcal{J}_k(z)=\begin{cases}
\di \frac{\pi \al_k}{\sin(\pi \al_k)}f(t_k)\om_{t_k}(t_k)(e^{\pi i \be_k} - e^{-\pi i \al_k}e^{-\pi i \be_k}) (z-t_k)^{\al_k}, & \mbox{if }\Re \al_k \leq 0, \ \al_k \neq 0, \\
f(t_k)\om_{t_k}(t_k)(e^{\pi i \be_k} - e^{-\pi i \be_k}), & \mbox{if }\al_k=0, \\
0, & \mbox{if }\Re \al_k > 0.
\end{cases}
\end{equation}
If $k = 0$, we have
\begin{equation}\label{J_0}
\mathcal{J}_0(z)=\begin{cases}
\di -\frac{\pi \al_0 e^{-\pi i \alpha_{0}}}{\sin(\pi \al_0)}f(t_0)\om_{-1}(t_0) (z-t_0)^{\al_0}, & \mbox{if } \Re \al_0 \leq 0, \ \al_0 \neq 0, \\
-f(t_0)\om_{-1}(t_0), & \mbox{if } \al_0=0, \\
0, & \mbox{if } \Re \al_0 > 0.
\end{cases}
\end{equation}
\end{proposition}
\begin{proof}
The proof for $k = 1,\ldots,m$ can be found in \cite[Proposition 3.1]{Charlier} (which is itself based on \cite{Krasovsky}). The case $k = 0$ can be proved as in \cite{Krasovsky}, or by replacing $e^{\pi i \beta_{k}}$ by $0$ and $e^{-\pi i \beta_{k}}$ by 1 in the proof of \cite[Proposition 3.1]{Charlier}. For the convenience of the reader, we also sketch the proof for $k=0$ here. The statement for $\alpha_{0}=0$ and $\Re \alpha_{0} > 0$ is trivial, and we assume from now that $\Re \alpha_{0} \leq 0, \; \alpha_{0}\neq 0$. Let $z$ be fixed, with $\Im z \neq 0$, but close enough to $t_{0}$ so that $z$ lies in the domain of analyticity of $f$. Let $\epsilon>0$ be such that $\epsilon < |z|$, let $r_{0}>t_{0}+\epsilon$ be fixed and such that $f$ is analytic at $r_{0}$, and let $\mathcal{C}$ be a close contour lying in the domain of analyticity of $f$ and surrounding $z$ in the positive direction as shown in Figure \ref{fig:contour C of Krasovsky}.
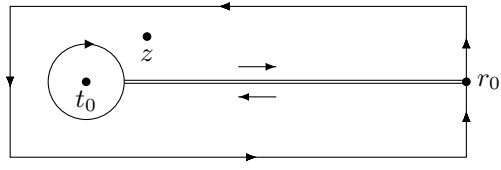
\begin{figure}
\begin{center}
\begin{tikzpicture}
\node at (0,0) {};
\draw[->-=0.203,->-=0.35,->-=0.503,->-=0.67,->-=0.805] (0.5,0.02)--(5,0.02)--(5,1)--(-1,1)--(-1,-1)--(5,-1)--(5,-0.02)--(0.5,-0.02);
\draw[fill] (0,0) circle (0.05);
\node at (0,-0.25) {$t_{0}$};
\draw (0,0) circle (0.5);
\draw[fill] (0.8,0.6) circle (0.05);
\node at (0.8,0.35) {$z$};
\draw[black,arrows={-Triangle[length=0.12cm,width=0.12cm]}]
($(0.1,0.5)$) --  ++(0:0.001);
\draw[->-=1] (2,0.2)--(2.5,0.2);
\draw[-<-=0] (2,-0.2)--(2.5,-0.2);

\node at (5.3,0) {$r_{0}$};
\draw[fill] (5,0) circle (0.05);
\end{tikzpicture}
\end{center}
\caption{The contour $\mathcal{C}$. The circle surrounding $t_{0}$ is part of $\mathcal{C}$ and its radius is equal to $\epsilon > 0$.}\label{fig:contour C of Krasovsky}
\end{figure}
In this proof, we consider the branch of $(x-t_{0})^{\alpha_{0}}$ associated to $\arg(x-t_{0}) \in [0,2\pi)$. By evaluating the integral
\begin{align*}
\int_{\mathcal{C}} \frac{f(x)\omega_{-1}(x)(x-t_{0})^{\alpha_{0}}}{(x-z)(x-t_{0})}dx
\end{align*}
both by residue and by direct parametrization, we obtain
\begin{multline*}
2\pi i f(z) \omega_{-1}(z)(z-t_{0})^{\alpha_{0}-1} = G_{1}(z) + (1-e^{2\pi i \alpha_{0}}) \int_{t_{0}+\epsilon}^{r_{0}} \frac{f(x)\omega_{-1}(x)|x-t_{0}|^{\alpha_{0}}}{(x-z)(x-t_{0})}dx \\ - \frac{1-e^{2\pi i \alpha_{0}}}{\alpha_{0}(z-t_{0})}f(t_{0})\omega_{-1}(t_{0})\epsilon^{\alpha_{0}} + \bigO(\epsilon^{\Re \alpha_{0}+1})
\end{multline*}
as $\epsilon \to 0$, where $G_{1}$ is analytic in a neighbourhood of $t_{0}=-1$. Multiplying by $\frac{\alpha_{0}(z-t_{0})}{1-e^{2\pi i \alpha_{0}}}$, and noting that
\begin{align*}
& \frac{1}{(x-z)(x-t_{0})} = \frac{1}{z-t_{0}}\bigg( \frac{1}{x-z}-\frac{1}{x-t_{0}} \bigg), \\
& \int_{r_{0}}^{+\infty} \frac{f(x)\omega(x)}{x-z}dx = \int_{r_{0}}^{+\infty} \frac{f(x)\omega(x)}{x-t_{0}}dx + (z-t_{0})G_{2}(z),
\end{align*}
where $G_{2}$ is analytic at $z=t_{0}$, we find
\begin{multline*}
\alpha_{0} \int_{\mathcal{I}\setminus[-1,-1+\epsilon]} \frac{f(x)\omega(x)}{x-t_{0}}dx + f(t_{0})\omega_{-1}(t_{0}) \epsilon^{\alpha_{0}} = \alpha_{0} \int_{\mathcal{I}} \frac{f(x)\omega(x)}{x-z}dx  \\ +\frac{\pi\alpha_{0} f(z)\omega_{-1}(z)}{\sin(\pi\alpha_{0})}(z-t_{0})^{\alpha_{0}}e^{- i \pi\alpha_{0}} + \frac{G_{3}(z)(z-t_{0}) + \bigO(\epsilon^{\Re\alpha_{0}+1})}{2i e^{i \pi\alpha_{0}}\sin (\pi\alpha_{0})}\alpha_{0},
\end{multline*}
where $G_{3}$ is analytic at $z=t_{0}$. We obtain the claim after first letting $\epsilon \to 0_{+}$ and then $z \to t_{0}$. 
\end{proof}
Since the second column of $Y(z)$ blows up as $z \to t_{j}$, $j=0,\ldots,m$, we regularize $Y$ at these points using the definitions \eqref{Reg_k} and \eqref{Reg_0} as follows:
\begin{equation}\label{Reg_Solution-(t_j)}
\widetilde{Y}(t_j) := \di \begin{pmatrix}
Y_{11}(t_j) & \mbox{Reg}_j\Big( \di \frac{1}{2 \pi i} Y_{11}(x)e^{-2n(x+1)} \Big)\\[0.25cm]
Y_{21}(t_j) & \mbox{Reg}_j\Big( \di \frac{1}{2 \pi i} Y_{21}(x)e^{-2n(x+1)}\Big)
\end{pmatrix}.
\end{equation}
From Proposition \ref{reg-integrals}, we have
\begin{equation}\label{RHPsolution_versus_Regsolution_k}
\widetilde{Y}_{k2}(t_j) = \lim_{z \to t_j} \al_j Y_{k2}(z) - c_j Y_{k1}(t_j)(z-t_j)^{\al_j}, \ \ \ \ \ k=1,2,
\end{equation}
where the limit is taken along a path in the upper half plane non-tangential to the real line. For $j = 1,\ldots,m$, $c_{j}$ is given by
\begin{equation}\label{c_j}
c_j = \frac{\pi \al_j}{\sin(\pi \al_j)}\frac{e^{-2n(t_j+1)}}{2\pi i}\om_{t_j}(t_j)(e^{\pi i \be_j} - e^{-\pi i \al_j}e^{-\pi i \be_j}),
\end{equation}
and for $j = 0$ we have
\begin{equation}\label{c_0}
c_0 = \frac{\pi \al_0}{\sin(\pi \al_0)}\frac{- e^{-\pi i \al_0}}{2\pi i}\om_{-1}(-1).
\end{equation}
Note that $\det \widetilde{Y}(t_j)$ is not equal to $1$, but instead we have
\begin{equation}\label{determinant of Y tilde}
\det \widetilde{Y}(t_j) = \al_j, \qquad j=0,1,\ldots, m.
\end{equation}
\begin{proposition}\label{prop: diff identity Laguerre}
Let $p_{0},p_{1},\ldots$ be the family of OPs with respect to the weight $w(x) = \omega(x)e^{-2n(x+1)}$, whose leading coefficients are denoted by
\begin{equation}\label{Orthogonal_Polynomials}
p_{k}(x) = \kappa_{k}(x^{k} + \eta_{k}x^{k-1}+\ldots).
\end{equation}
Let $\nu \in \{\alpha_{0},\alpha_{1},\beta_{1},\ldots,\alpha_{m},\beta_{m}\}$ and let $n$, $\vec{\alpha}$ and $\vec{\beta}$ be such that $p_{0},p_{1},\ldots,p_{n}$ exist. We have the following identity: 
\begin{multline}\label{Diff_identity_alphas_betas}
\partial_{\nu} \log L_n(\vec{\al},\vec{\be}, 2(x+1),0) = -(n + \mathcal{A})\partial_{\nu}\log(\ka_n\ka_{n-1})+2n \partial_{\nu} \eta_n \\ +  \sum^{m}_{j=0} \Big( \widetilde{Y}_{22}(t_j)\partial_{\nu}Y_{11}(t_j) - \widetilde{Y}_{12}(t_j)\partial_{\nu}Y_{21}(t_j) + Y_{11}(t_j)\widetilde{Y}_{22}(t_j)\partial_{\nu}\log(\ka_n\ka_{n-1}) \Big),
\end{multline}
where $\mathcal{A} = \sum_{j=0}^{m} \alpha_{j}$.
\end{proposition}
\begin{remark}
We do not need an analogous formula for $\partial_{\nu} \log J_{n}(\vec{\alpha},\vec{\beta},0,0)$ as large $n$ asymptotics of $J_{n}(\vec{\alpha},\vec{\beta},0,0)$ are already known from \cite{DIK}, see the outline.
\end{remark}
\begin{proof}
The proof is an adaptation of \cite[Subsection 3.1]{Charlier} where the author obtained a differential identity for $\partial_{\nu} \log G_n(\vec{\al},\vec{\be}, 2x^{2},0)$ (this proof was itself a generalization of \cite{Krasovsky, ItsKrasovsky}). Here, the proof is even slightly easier, due to the fact that the potential is a polynomial of degree $1$ (and not of degree $2$ as in \cite{Krasovsky,ItsKrasovsky,Charlier}). Since we assume that $p_{0},\ldots,p_{n}$ exist, we can use the following general identity, which was obtained in \cite{Krasovsky}
\begin{equation}\label{general_diff_identity}
\partial_{\nu} \log L_n(\vec{\al},\vec{\be}, 2(x+1),0) = -n \partial_{\nu} \log \ka_{n-1}+\frac{\ka_{n-1}}{\ka_n}(I_1-I_2),
\end{equation}
where 
\begin{equation}\label{J_1 and J_2}
I_1 = \int_{\mathcal{I}} p'_{n-1}(x) \partial_{\nu} p_n(x)w(x)dx, \quad \mbox{and} \quad I_2 = \int_{\mathcal{I}} p'_{n}(x) \partial_{\nu} p_{n-1}(x)w(x)dx. 
\end{equation}
Since $\Re \alpha_{j} > -1$ for all $j = 0,1,\ldots,m$, we first note that
\begin{equation}\label{I_1 as a limit}
I_{1} = \lim_{\epsilon \to 0_{+}} \int_{\mathcal{I}_{\epsilon}} p'_{n-1}(x) \partial_{\nu} p_n(x)w(x)dx,
\end{equation}
where $\mathcal{I}_{\epsilon}$ is the union of $m+1$ intervals given by
\begin{equation*}
\mathcal{I}_{\epsilon} = [t_{0}+\epsilon, t_{1}-\epsilon] \cup [t_{1}+\epsilon,t_{2}-\epsilon] \cup \ldots \cup [t_{m-1}+\epsilon,t_{m}-\epsilon] \cup [t_{m}+\epsilon,\infty).
\end{equation*}
Along each of these $m+1$ intervals, we integrate by parts (for each fixed and sufficiently small $\epsilon$), using 
\begin{equation}\label{w'_L}
w^{\prime}(x)= \bigg( -2n + \sum_{j=0}^{m} \frac{\al_j}{x-t_j} \bigg) w(x), \qquad x \in (-1,\infty)\setminus \{t_{1},\ldots,t_{m}\}.
\end{equation}
Then, we simplify the expression using the orthogonality relations \eqref{orthogonality_relation}, and more precisely using
\begin{align*}
& \int_{\mathcal{I}_{\epsilon}} p_{n-1}(x)\partial_{\nu}p_{n}(x)w(x)dx = \frac{\kappa_{n}}{\kappa_{n-1}}\partial_{\nu}\eta_{n} + o(1), \\
& \int_{\mathcal{I}_{\epsilon}} p_{n-1}(x)\partial_{\nu}p_{n}'(x)w(x)dx = n\frac{\partial_{\nu}\kappa_{n}}{\kappa_{n-1}} + o(1), \\
& \int_{\mathcal{I}_{\epsilon}}p_{n-1}(x)\frac{\partial_{\nu}p_{n}(x)}{x-t_{j}}w(x)dx = \frac{\partial_{\nu}\kappa_{n}}{\kappa_{n-1}} + \partial_{\nu}p_{n}(t_{j}) \int_{\mathcal{I}_{\epsilon}} \frac{p_{n-1}(x)}{x-t_{j}}w(x)dx+o(1), 
\end{align*}
as $\epsilon \to 0$. Finally, we substitute these expression in the limit \eqref{I_1 as a limit}, using \eqref{Reg_k} and \eqref{Reg_0}, and we find
\begin{equation}\label{J1_main}
I_1 =   -(n + \mathcal{A}) \frac{\partial_{\nu} \ka_n}{\ka_{n-1}} +2n \frac{\ka_n}{\ka_{n-1}} \partial_{\nu} \eta_n    - \sum_{j=0}^{m} \partial_{\nu} p_n(t_j) \mbox{Reg}_j\left[p_{n-1}(x)e^{-2n(x+1)}\right].
\end{equation}
We proceed similarly for $I_{2}$. Using that
\begin{align*}
& \int_{\mathcal{I}_{\epsilon}} p_{n}(x)\partial_{\nu}p_{n-1}(x)w(x)dx = o(1), \\
& \int_{\mathcal{I}_{\epsilon}} p_{n}(x)\partial_{\nu}p_{n-1}'(x)w(x)dx = o(1), \\
& \int_{\mathcal{I}_{\epsilon}}p_{n}(x)\frac{\partial_{\nu}p_{n-1}(x)}{x-t_{j}}w(x)dx = \partial_{\nu}p_{n-1}(t_{j}) \int_{\mathcal{I}_{\epsilon}} \frac{p_{n}(x)}{x-t_{j}}w(x)dx+o(1), 
\end{align*}
as $\epsilon \to 0$, we find
\begin{equation}
\label{J2_main}
I_2 =  - \sum_{j=0}^{m} \partial_{\nu} p_{n-1}(t_j) \mbox{Reg}_j\left[p_{n}(x)e^{-2n(x+1)}\right].
\end{equation}
By rewriting first $I_{1}$ and $I_{2}$ in terms of $Y$ and $\widetilde{Y}$, then by substituting these expressions into \eqref{general_diff_identity}, and finally by using \eqref{determinant of Y tilde}, we obtain the claim.
\end{proof}
\subsection{A general differential identity}
We recall here a differential identity that is valid for all three types of weights. In Section \ref{Section: integration in V} and Section \ref{Section: integration in W}, we will use Proposition \ref{prop: general diff identity} below with $\nu = s$ or $\nu = t$ to obtain identities for the quantities in \eqref{lol2} (save the case of $\partial_{\nu}L_{n}(\vec{\alpha},\vec{\beta},2(x+1),0)$ for which we will use Proposition \ref{prop: diff identity Laguerre}).
\begin{proposition}\label{prop: general diff identity}
Let $D_{n}$ be a Hankel determinant whose weight $w$ depends smoothly on a parameter $\nu$. Let us assume that the associated orthonormal polynomials $p_{0}$,\ldots,$p_{n}$ exist. Then we have
\begin{align}
& \label{Diff_identity very general}
\partial_{\nu} \log D_n = \frac{1}{2\pi i} \int_{\mathcal{I}} [Y^{-1}(x)Y'(x)]_{21} \partial_\nu w(x) dx,
\end{align}
where $\mathcal{I}$ is the support of $w$, and $Y$ is given by \eqref{OPsolution}.
\end{proposition} 
\begin{proof}
It suffices to start from the well-known \cite{Szego OP} identity
\begin{equation}\label{det as product}
D_{n} = \prod_{j=0}^{n-1} \kappa_{j}^{-2},
\end{equation}
take the log, differentiate with respect to $\nu$, use the orthogonality relations and finally substitute $Y$ in the expression.
\end{proof}
	
\section{Steepest descent analysis}\label{Section: steepest descent}

In this section we will construct an asymptotic solution to the RH problem for $Y$ through the Deift/Zhou steepest descent method, for Laguerre-type and Jacobi-type weights. The analysis goes via a series of transformations $Y \mapsto T \mapsto S \mapsto R$. The $Y \mapsto T$ transformation of Subsection \ref{Subsection: Y to T} normalizes the RH problem at $\infty$ by means of a so-called $g$-function (whose properties are presented in Subsection \ref{subsection: eq measure}). We proceed with the opening of the lenses $T \mapsto S$ in Subsection \ref{Subsection: T to S}. As a preliminary to the last step $S\mapsto R$, we first construct approximations (called ``parametrices") for $S$ in different regions of the complex plane: a global parametrix in Subsection \ref{subsection:global parametrix}, local parametrices in the bulk around $t_{k}$ in Subsection \ref{subsection: local param t_k}, and local parametrices at the edges $\pm 1$ in Subsection \ref{subsection: local param near 1} and Subsection \ref{subsection: local param near -1}. These parametrices are rather standard: our global parametrix is close to the one done in \cite{Charlier} and local parametrices in the bulk are built out of confluent hypergeometric functions (as in \cite{FouMarSou,ItsKrasovsky,DIK}), local parametrices at soft edges in terms of Airy functions (as in \cite{Deiftetal}) and at a hard edge, in terms of Bessel functions (as in \cite{KMcLVAV}). Finally, the last step $S \mapsto R$ is carried out in Subsection \ref{subsection: small norm}.

\subsection{Equilibrium measure and $g$-function}\label{subsection: eq measure}
It is convenient for us to introduce the notation $\rho$ for the density of $\mu_{V}$:
\begin{equation}\label{def of rho}
d\mu_{V}(x) = \rho(x)dx = \left\{ \begin{array}{l l}
\ds \psi(x) \frac{\sqrt{1-x}}{\sqrt{1+x}}dx, & \mbox{for Laguerre-type weight}, \\[0.3cm]
\ds \psi(x) \frac{1}{\sqrt{1-x^{2}}}dx, & \mbox{for Jacobi-type weight},
\end{array} \right.
\end{equation}
where we recall that by assumption $\psi: \mathcal{I}\to \mathbb{R}$ is analytic and positive on $[-1,1]$. Let $U_{V}$ be the maximal open neighbourhood of $\mathcal{I}$ in which $V$ is analytic, and $U_{W}$ be an open neighbourhood of $[-1, 1]$ in which $W$ is analytic, sufficiently small such that $U_{W} \subset U_{V}$. The so-called $g$-function is defined by	
\begin{equation}\label{g function}
g(z) = \int^{1}_{-1} \log(z-s)\rho(s)ds, \qquad \mbox{for }  z \in \C \setminus (-\infty,1],
\end{equation}
where the principal branch is chosen for the logarithm. The $g$-function is analytic in $\C \setminus (-\infty,1]$ and has the following properties
\begin{align}
& g_{+}(x) + g_{-}(x) = 2 \int^{1}_{-1} \log|x-s|\rho(s)ds, & & x \in \R \label{sum}, \\
& g_{+}(x) - g_{-}(x) = 2 \pi i, & & x \in (-\infty,-1), \label{g+_minus_g- 1} \\
& g_{+}(x) - g_{-}(x) = 2 \pi i \di \int_{x}^{1} \rho(s)ds, & & x \in [-1,1]. \label{g+_minus_g- 2}
\end{align}
The Euler-Lagrange conditions \eqref{var equality}-\eqref{var inequality} can be rewritten in terms of the $g$-function as follows:
\begin{align}
& g_{+}(x) + g_{-}(x) = V(x)-\ell, & & x \in [-1,1], \label{EL =} \\[0.1cm]
& 2g(x) < V(x)-\ell, & & x \in \mathcal{I}\setminus[-1,1]. \label{EL <}
\end{align}
The above inequality is relevant only for Laguerre-type weight (since for Jacobi-type weight $\mathcal{I}\setminus [-1,1]= \emptyset$), and is strict since we assume that $V$ is regular.

\vspace{0.3cm}\hspace{-0.55cm}For $z \in U_{V} \setminus [-1,1]$, we define
\begin{equation}
\widetilde{\rho}(z) = \left\{ \begin{array}{l l}
\ds -i\psi(z) \frac{\sqrt{z-1}}{\sqrt{z+1}}, & \mbox{for Laguerre-type weight}, \\[0.3cm]
\ds i\psi(z) \frac{1}{\sqrt{z^{2}-1}}, & \mbox{for Jacobi-type weight},
\end{array} \right.
\end{equation}
where the principal branches are chosen for $\sqrt{z-1}$ and $\sqrt{z+1}$. Note that for $x \in (-1,1)$ we have $\widetilde{\rho}_{+}(s) = - \widetilde{\rho}_{-}(s) = \rho(s)$. Let us also define 
\begin{equation}\label{def-of-xi}
\xi(z) = -\pi i \int_{1}^{z} \widetilde{\rho}(s) ds, \qquad z \in U_{V} \setminus (-\infty,1),
\end{equation}
where the path of integration lies in $U_{V}\setminus (-\infty,1)$. Since $\xi_{+}(x) +\xi_{-}(x) = 0$ for $x \in (-1,1)$, by \eqref{g+_minus_g- 2} and \eqref{EL =}, we have 
\begin{equation}\label{2xi+}
2 \xi_{\pm}(x) = g_{\pm}(x) - g_{\mp}(x) = 2g_{\pm}(x) - V(x) + \ell.
\end{equation}
By analytic continuation, we have 
\begin{equation}\label{analytic-continuation-of-xi}
\xi(z)=g(z)+\frac{\ell}{2}-\frac{V(z)}{2},  \qquad z \in U_{V} \setminus (-\infty,1).
\end{equation}
Thus, the Euler-Lagrange inequality \eqref{EL <} can be simply rewritten as $2 \xi(x)< 0$ for $x \in \mathcal{I}\setminus [-1,1]$. Furthermore, since $g(z) \sim \log(z)$ as $z \to \infty$, we have that $ (\xi_{+}(x)+\xi_{-}(x))/V(x)\to-1$ as $x \to + \infty$, $x \in \mathcal{I}$. Finally, by a standard and straightforward analysis of $\xi$, we conclude that there exists a small enough neighbourhood of $(-1,1)$ such that, for $z$ in this neighbourhood with $\Im z \neq 0$, we have $\Re{\xi(z)}>0$. 

\vspace{0.3cm}\hspace{-0.55cm}We will also need later large $z$ asymptotics of $e^{ng(z)}$ for the Laguerre-type potential $V(x) = 2(x+1)$. In this case, we recall that $\psi(x) = \frac{1}{\pi}$, and after a straightforward calculation we obtain
\begin{equation}\label{e^ng_asym_Laguerre}
e^{ng(z)} = z^{n} \Big( 1+\frac{n}{2z} + \bigO(z^{-2}) \Big), \qquad \mbox{as } z \to \infty.
\end{equation}

\subsection{First transformation: $Y \mapsto T$}\label{Subsection: Y to T}
We normalize the RH problem for $Y$ at $\infty$ by the standard transformation	
\begin{equation}\label{YtoT}
T(z) := e^{\frac{n \ell}{2} \sigma_3} Y(z) e^{-ng(z) \sigma_3} e^{-\frac{n \ell}{2} \sigma_3}.
\end{equation}
$T$ satisfies the following RH problem.
\subsubsection*{RH problem for $T$}
\begin{itemize}
\item[(a)] $T: \C \setminus \mathcal{I} \to \C^{2 \times 2}$ is analytic.
\item[(b)] The jumps for $T$ follows from \eqref{g+_minus_g- 1}, \eqref{EL =} and \eqref{analytic-continuation-of-xi}. We obtain
\begin{align}
& T_+(x)=T_-(x)\begin{pmatrix}
e^{-2n\xi_{+}(x)} & e^{W(x)}\om(x) \\
0 & e^{2n\xi_{+}(x)}
\end{pmatrix}, & & \mbox{if }x\in (-1,1) \setminus \{ t_1,\cdots,t_m \}, \\
& T_+(x)=T_-(x)\begin{pmatrix}		
1 & e^{W(x)}\om(x) e^{2n\xi(x)} \\
0 & 1
\end{pmatrix}, & & \mbox{if } x \in \mathcal{I}\setminus [-1,1].	
\end{align}
\item[(c)] As $z \to \infty$, $T(z) = I + \bigO(z^{-1})$.
\item[(d)] As $z \to t_j$, for $j=0,1,\ldots,m+1$, we have
\begin{equation}\label{T_Asymptotics_tj}
T(z)=\begin{cases}\begin{pmatrix}
\mathcal{O}(1) & \mathcal{O}(1) +\mathcal{O}((z-t_j)^{\al_j})\\
\mathcal{O}(1) & \mathcal{O}(1) + \mathcal{O}((z-t_j)^{\al_j})\end{pmatrix}, & \mbox{if } \Re \al_j \neq 0,  \\
\begin{pmatrix}
\mathcal{O}(1) & \mathcal{O}(\log(z-t_j))\\
\mathcal{O}(1) & \mathcal{O}(\log(z-t_j))\end{pmatrix},  & \mbox{if } \Re \al_j = 0.
\end{cases}
\end{equation}
\end{itemize}
\subsection{Second transformation: $T \mapsto S$}\label{Subsection: T to S}
In this step, we will deform the contour of the RH problem. Therefore, we first consider the analytic continuations of the functions $\om_{\al_k}$ and $\om_{\be_k}$ from $\mathbb{R}\setminus \{t_{k}\}$ to $\C \setminus \{z: \Re(z)=t_k\}$. They are given by 
\begin{equation}
\om_{\al_k}(z) = \begin{cases}
(t_k-z)^{\al_k}, & \mbox{if } \Re z < t_k, \\
(z-t_k)^{\al_k}, & \mbox{if }\Re z > t_k, \\    
\end{cases} \qquad  \quad \om_{\be_k}(z) = \begin{cases}
e^{i \pi \be_k}, & \mbox{if }\Re z < t_k, \\
e^{-i \pi \be_k}, & \mbox{if }\Re z > t_k. \\    
\end{cases}
\end{equation}
For $k = 0,\ldots,m+1$, we also define 
\begin{equation}
\omega_{t_{k}}(z) = \prod_{\substack{0 \leq j \leq m \\ j \neq k }} \omega_{\alpha_{j}}(z)\omega_{\beta_{j}}(z).
\end{equation}
Note that for $x \in  (-1,1) \setminus \{ t_1,\ldots,t_m \}$ we have the following factorization for $J_T(x)$ :
\begin{multline}\label{factorization of J_T}
\begin{pmatrix}
e^{-2n\xi_{+}(x)} & e^{W(x)}\om(x) \\
0 & e^{2n\xi_{+}(x)}
\end{pmatrix} = \begin{pmatrix}
1 & 0 \\
e^{-W(x)}\om(x)^{-1}e^{-2n\xi_{-}(x)} & 1
\end{pmatrix} \\ \times \begin{pmatrix}
0 & e^{W(x)}\om(x) \\
-e^{-W(x)}\om(x)^{-1} & 0
\end{pmatrix}\begin{pmatrix}
1 & 0 \\
e^{-W(x)}\om(x)^{-1}e^{-2n\xi_{+}(x)} & 1
\end{pmatrix}.
\end{multline}
Let $\gamma_{+}$ and $\gamma_{-}$ be two curves (lying respectively in the upper and lower half plane) that join the points $-1,t_{1},\ldots,t_{m},1$ as depicted in Figure \ref{opening-lenses}. In order to be able to deform the contour of the RH problem, we choose them so that they both lie in $U_{W}$. In the constructions of the local parametrices, they will be required to make angles of $\frac{\pi}{4}$ with $\mathbb{R}$ at the points $t_{1},\ldots,t_{m}$, and angles of $\frac{\pi}{3}$ with $\mathbb{R}$ at the points $\pm 1$, and this is already shown in Figure \ref{opening-lenses}. Also, we denote $\Omega_{\pm}$ for the open regions delimited by $\gamma_{\pm}$ and $\mathbb{R}$, see Figure \ref{opening-lenses}. The next transformation is given by
\begin{figure}
\centering
\begin{tikzpicture}
\fill (0,0) circle (0.07cm);
\fill (3,0) circle (0.07cm);
\fill (7.5,0) circle (0.07cm);
\fill (10,0) circle (0.07cm);
\node at (0,-0.35) {$-1$};
\node at (3,-0.35) {$t_{1}$};
\node at (7.5,-0.35) {$t_{m}$};
\node at (10,-0.35) {$1$};
\draw[->-=0.5,thick] (0,0)--(3,0);
\draw[->-=0.5,thick] (3,0)--(7.5,0);
\draw[->-=0.5,thick] (7.5,0)--(10,0);
\draw[->-=0.5,thick] (10,0)--(12,0);
\draw[->-=0.5,black,thick] (0,0) to [out=60, in=135] (3,0);
\draw[->-=0.5,black,thick] (3,0) to [out=45, in=135] (7.5,0);
\draw[->-=0.5,black,thick] (7.5,0) to [out=45, in=120] (10,0);
\draw[->-=0.5,black,thick] (0,0) to [out=-60, in=-135] (3,0);
\draw[->-=0.5,black,thick] (3,0) to [out=-45, in=-135] (7.5,0);
\draw[->-=0.5,black,thick] (7.5,0) to [out=-45, in=-120] (10,0);
\node at (5,1.2) {$\gamma_{+}$};
\node at (5,-1.2) {$\gamma_{-}$};
\node at (5.2,0.5) {$\Omega_{+}$};
\node at (5.2,-0.5) {$\Omega_{-}$};
\end{tikzpicture}
\caption{The jump contour for the RH problem for $S$ with $m=2$ and a Laguerre-type weight. For Jacobi-type weights, the jump contour for $S$ is of the same shape, except that there are no jumps on $(1,+\infty)$.}
\label{opening-lenses}
\end{figure}
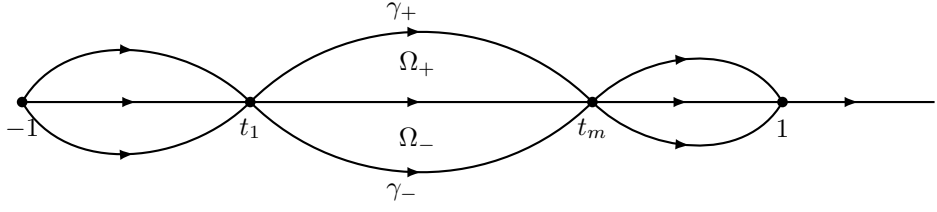
\begin{equation}\label{TtoS}
S(z) = T(z) \begin{cases}
\begin{pmatrix}
1 & 0 \\
-e^{-W(z)}\om(z)^{-1}e^{-2n\xi(z)} & 1
\end{pmatrix}, & \mbox{if } z \in  \Om_{+}, \\
\begin{pmatrix}
1 & 0 \\
e^{-W(z)}\om(z)^{-1}e^{-2n\xi(z)} & 1
\end{pmatrix}, & \mbox{if } z \in \Om_{-}, \\
I, & \mbox{if }z \in \C \setminus \overline{\left( \Om_{+} \cup \Om_{-} \cup (\mathcal{I}\setminus \mathcal{S}) \right)}.
\end{cases}
\end{equation} 
$S$ satisfies the following RH problem.
\subsubsection*{RH problem for $S$}
\begin{itemize}
\item[(a)] $S: \C \setminus (\mathcal{I}\cup \gamma_{+}\cup \gamma_{-}) \to \C^{2 \times 2}$ is analytic.
\item[(b)] The jumps for $S$ follows from those of $T$ and from \eqref{factorization of J_T}. They are given by
\begin{align}
& S_+(z)=S_-(z)\begin{pmatrix}
0 & e^{W(z)}\om(z) \\
- e^{-W(z)}\om(z)^{-1} & 0
\end{pmatrix}, & & \mbox{if }z\in (-1,1) \setminus \{ t_1,\ldots,t_m \}, \label{J_S 1} \\
& S_+(z)=S_-(z)\begin{pmatrix}	
1 & e^{W(z)}\om(z) e^{2n\xi(z)} \\
0 & 1
\end{pmatrix}, & & \mbox{if } z \in \mathcal{I}\setminus [-1,1], \\
& S_+(z)=S_-(z)\begin{pmatrix}
1 & 0 \\
e^{-W(z)}\om(z)^{-1}e^{-2n\xi(z)} & 1
\end{pmatrix}, & & \mbox{if } z \in \gamma_{+}\cup \gamma_{-}. \label{J_S 3}
\end{align}
\item[(c)] As $z \to \infty$, $S(z) = I + \bigO(z^{-1})$.
\item[(d)] As $z \to t_j$, for $j=0,1,\ldots,m+1$, we have
\begin{equation}\label{S-asymptotics-t_j}
S(z)=\begin{cases}\begin{pmatrix}
\mathcal{O}(1) & \mathcal{O}(1) \\
\mathcal{O}(1) & \mathcal{O}(1) \end{pmatrix} & \mbox{if } \Re \al_j > 0 , \ \ z \in \C \setminus \overline{(\Om_{+} \cup \Om_{-})},  \\[0.35cm]
\begin{pmatrix}
\mathcal{O}((z-t_j)^{-\al_j}) & \mathcal{O}(1) \\
\mathcal{O}((z-t_j)^{-\al_j}) & \mathcal{O}(1) \end{pmatrix} & \mbox{if } \Re \al_j > 0 , \ \ z \in \Om_{+} \cup \Om_{-},  \\[0.35cm]
\begin{pmatrix}
\mathcal{O}(1) &  \mathcal{O}((z-t_j)^{\al_j})\\
\mathcal{O}(1) & \mathcal{O}((z-t_j)^{\al_j})\end{pmatrix} & \mbox{if } \Re \al_j < 0 , \ \ \ z \notin \Ga_S, \\[0.35cm]
\begin{pmatrix}
\mathcal{O}(1) & \mathcal{O}(\log(z-t_j))\\
\mathcal{O}(1) & \mathcal{O}(\log(z-t_j))\end{pmatrix}  & \mbox{if } \Re \al_j = 0, \ \ z \in \C \setminus \overline{(\Om_{+} \cup \Om_{-})},  \\[0.35cm]
\begin{pmatrix}
\mathcal{O}(\log(z-t_j)) & \mathcal{O}(\log(z-t_j))\\
\mathcal{O}(\log(z-t_j)) & \mathcal{O}(\log(z-t_j))\end{pmatrix}  & \mbox{if } \Re \al_j = 0, \ \ z \in  \Om_{+} \cup \Om_{-}. \\
\end{cases}
\end{equation}
\end{itemize}
Now, the rest of the steepest descent analysis consists of finding good approximations to $S$ in different regions of the complex plane. If $z$ is away from neighbourhoods of $-1$, $t_{1}$, ..., $t_{m}$, $1$, then the jumps for $S$ are uniformly exponentially close to the identity matrix, except those on $(-1,1)$ (see the discussion at the end of Section \ref{subsection: eq measure}). By ignoring the jumps that tend to the identity matrix, we are left with an RH problem that does not depend on $n$, and whose solution will be a good approximation of $S$ away from $-1$, $t_{1}$, ..., $t_{m}$, $1$. This approximation is called the global parametrix, denoted by $P^{(\infty)}$, and will be given in Section \ref{subsection:global parametrix} below. Near the points $-1$, $t_{1}$, ..., $t_{m}$, $1$ we need to construct local approximations to $S$ (also called local parametrices and denoted in the present paper by $P^{(-1)}$, $P^{(t_{1})}$, $\ldots$, $P^{(1)}$). Let $\de > 0$, independent of $n$, be such that 
\begin{equation}
\de \leq \min_{0 \leq k \neq j \leq m+1} |t_j-t_k|.
\end{equation}
The local parametrix $P^{(t_{k})}$ (for $k \in \{0,1,\ldots,m,m+1\}$) solves an RH problem with the same jumps as $S$, but on a domain which is a disk $\mathcal{D}_{t_k}$ centered at $t_{k}$ of radius $\leq \de/3$. Furthermore, we require the following matching condition with $P^{(\infty)}$ on the boundary $\partial\mathcal{D}_{t_k}$. As $n \to \infty$, uniformly for $z \in \partial \mathcal{D}_{t_{k}}$, we have
\begin{equation}\label{lol1}
P^{(t_{k})}(z) = (I + o(1))P^{(\infty)}(z).
\end{equation} Again, these constructions are standard and well-known: near a FH singularity in the bulk, the local parametrix is given in terms of hypergeometric functions, near a soft edge in terms of Airy functions, and near a hard edge in terms of Bessel functions. The local parametrices are presented in Section \ref{subsection: local param t_k}, Section \ref{subsection: local param near 1} and Section \ref{subsection: local param near -1}.
\subsection{Global parametrix}\label{subsection:global parametrix}
By disregarding the jump conditions on the lenses $\ga_+ \cup \ga_-$ and on $\mathcal{I}\setminus [-1,1]$, we are left with the following RH problem for $P^{(\infty)}$ (condition (d) below ensures uniqueness of the RH problem and can not be seen from the RH problem for $S$).

\subsubsection*{RH problem for $P^{(\infty)}$}
\begin{itemize}
\item[(a)] $P^{(\infty)}: \C \setminus [-1,1] \to \C^{2 \times 2}$ is analytic.
\item[(b)] The jumps for $P^{(\infty)}$ are given by
\begin{align*}
& P^{(\infty)}_+(z)=P^{(\infty)}_-(z)\begin{pmatrix}
0 & e^{W(z)}\om(z) \\
-e^{-W(z)}\om(z)^{-1} & 0
\end{pmatrix}, & & \mbox{if }z\in (-1,1) \setminus \{ t_1,\ldots,t_m \}.
\end{align*}
\item[(c)] As $z \to \infty$, $P^{(\infty)}(z) = I + P_1^{(\infty)} z^{-1} + \bigO(z^{-2})$.
\item[(d)] As $z \to t_j$, for $j=1,\ldots,m$, we require
\begin{equation}\label{Global_Asymptotics_tj}
P^{(\infty)}(z)= \begin{pmatrix}
\bigO(1) & \bigO(1) \\
\bigO(1) & \bigO(1)
\end{pmatrix}(z-t_j)^{-(\frac{\al_j}{2}+\be_j)\sigma_3}.
\end{equation}
As $z \to t_{j}$ with $j \in \{0,m+1\}$ (we recall that $t_{0} = -1$ and $t_{m+1} = 1$, and that $\alpha_{m+1} = 0$ for Laguerre-type weight), we have
\begin{equation}\label{Global_Asymptotics_tj extreme}
P^{(\infty)}(z)= \begin{pmatrix}
\bigO((z-t_{j})^{-\frac{1}{4}}) & \bigO((z-t_{j})^{-\frac{1}{4}}) \\
\bigO((z-t_{j})^{-\frac{1}{4}}) & \bigO((z-t_{j})^{-\frac{1}{4}})
\end{pmatrix}(z-t_{j})^{-\frac{\al_j}{2}\sigma_3}.
\end{equation}
\end{itemize}
\begin{remark}
Note that this RH problem is the same regardless of the weight, the only exception being that $\alpha_{m+1} = 0$ for Laguerre-type weight (and not necessarily for Jacobi-type weight).
\end{remark}
This RH problem was solved first in \cite{Deiftetal} with $W \equiv 0$ and $ \omega \equiv 0$. In \cite{KMcLVAV}, the authors explain how to construct the solution to the above RH problem for general $W$ and $\omega$ by using Szeg\"{o} functions. Our RH problem for $P^{(\infty)}$ is close to the one obtained in \cite{Charlier} for Gaussian-type weights. The solution is given by	
\begin{equation}\label{Global_Parametrix}P^{(\infty)}(z) = D^{\sigma_3}_{\infty} \begin{pmatrix}
\frac{1}{2}(a(z)+a(z)^{-1}) & \frac{1}{2i}(a(z)-a(z)^{-1}) \\
-\frac{1}{2i}(a(z)-a(z)^{-1}) & \frac{1}{2}(a(z)+a(z)^{-1})
\end{pmatrix} D(z)^{-\sigma_3},
\end{equation}
where $a(z) = \sqrt[4]{\frac{z-1}{z+1}}$ is analytic on $\C \setminus [-1,1]$ and $a(z) \sim 1$ as $z \to \infty$. The Szeg\"{o} function $D$ is given by $D(z)=D_{\al}(z)D_{\be}(z)D_{W}(z)$, where
\begin{align}
& \hspace{-0.2cm} D_{W}(z) = \exp \left(  \frac{\sqrt{z^2-1}}{2\pi} \int_{-1}^{1} \frac{W(x)}{\sqrt{1-x^2}} \frac{dx}{z-x} \right), \label{D_W_L} \\
& \hspace{-0.2cm} D_{\al}(z)= \prod_{j=0}^{m+1}  \exp \left(  \frac{\sqrt{z^2-1}}{2\pi} \int_{-1}^{1} \frac{\log \om_{\al_j}(x)}{\sqrt{1-x^2}} \frac{dx}{z-x} \right) = \left(z+\sqrt{z^2-1} \right)^{-\frac{\mathcal{A}}{2}} \prod^{m+1}_{j=0} (z-t_j)^{\frac{\al_j}{2}} \label{D_al}, \\
& \hspace{-0.2cm} D_{\be}(z)= \prod_{j=1}^m  \exp \left(  \frac{\sqrt{z^2-1}}{2\pi} \int_{-1}^{1} \frac{\log \om_{\be_j}(x)}{\sqrt{1-x^2}} \frac{dx}{z-x} \right)=e^{\frac{i \pi \mathcal{B}}{2}} \prod_{j=1}^m \left( \frac{z t_j -1 -i \sqrt{(z^2-1)(1-t^2_j)}}{z-t_j} \right)^{\be_j}, \label{D_be}
\end{align}
where $\mathcal{A} = \sum_{j=0}^{m+1} \alpha_{j}$ and $\mathcal{B} = \sum_{j=1}^{m} \beta_{j}$. The simplified forms of \eqref{D_al} and \eqref{D_be} were found in \cite{KMcLVAV} and \cite{ItsKrasovsky}, respectively. Also, $D_{\infty} = \lim_{z \to \infty} D(z)$ appearing in \eqref{Global_Parametrix} is given by 
\begin{equation}\label{D_infty}
D_{\infty} = 2^{-\frac{\mathcal{A}}{2}}\exp\Big(i \sum_{j=1}^{m} \be_j \arcsin{t_j}\Big)\exp \left(  \frac{1}{2\pi} \int_{-1}^{1} \frac{W(x)}{\sqrt{1-x^2}} dx \right).
\end{equation}
The following asymptotic expressions were obtained in \cite[Section 4.4]{Charlier} with $\alpha_{0} = \alpha_{m+1} = 0$. It is straightforward to adapt them for general $\alpha_{0}$ and $\alpha_{m+1}$. As $z \to t_{k}$, with $k \in \{1,\ldots,m\}$ and $\Im z>0$, we have
\begin{align}
& D_{\alpha}(z) = e^{-i \frac{\mathcal{A}}{2}\arccos t_{k}}\bigg(  \prod_{0 \leq j \neq k \leq m+1} |t_{k}-t_{j}|^{\frac{\alpha_{j}}{2}} \prod_{j=k+1}^{m}e^{\frac{i\pi\alpha_{j}}{2}} \bigg) (z-t_{k})^{\frac{\alpha_{k}}{2}}(1+\bigO(z-t_{k})), \label{asymptotics for D_alpha in D_t} \\
& D_{\beta}(z) = e^{-\frac{i\pi}{2}(\mathcal{B}_{k}+\beta_{k})}  \bigg( \prod_{1 \leq j\neq k \leq m} T_{kj}^{\beta_{j}}\bigg) (1-t_{k}^{2})^{-\beta_{k}}2^{-\beta_{k}}(z-t_{k})^{\beta_{k}}(1+\bigO(z-t_{k})), \label{asymptotics for D_beta in D_t} 
\end{align}
where
\begin{equation}
\mathcal{B}_{k} = \sum_{j=1}^{k-1}\beta_{j} - \sum_{j=k+1}^{m} \beta_{j}, \qquad T_{kj} = \frac{1-t_{k}t_{j}-\sqrt{(1-t_{k}^{2})(1-t_{j}^{2})}}{|t_{k}-t_{j}|}.
\end{equation}
Let us also define the following quantities:
\begin{equation}
\widetilde{\mathcal{B}}_{1} = 2 i \sum_{j=1}^{m} \sqrt{\frac{1+t_{j}}{1-t_{j}}}\beta_{j},\qquad \widetilde{\mathcal{B}}_{-1} = 2 i \sum_{j=1}^{m} \sqrt{\frac{1-t_{j}}{1+t_{j}}}\beta_{j}.
\end{equation}
As $z \to 1$, we have
\begin{align}\label{SzAsym+1}
& D_{\alpha}^{2}(z)\prod_{j=0}^{m+1}(z-t_{j})^{-\alpha_{j}} =  1-\sqrt{2} \mathcal{A} \sqrt{z-1} + \mathcal{A}^{2}(z-1)+ \bigO((z-1)^{3/2}), \\
& D_{\beta}^{2}(z)e^{i\pi\mathcal{B}} = 1+   \sqrt{2} \widetilde{\mathcal{B}}_{1} \sqrt{z-1} + \widetilde{\mathcal{B}}_{1}^{2}(z-1) + \bigO((z-1)^{3/2}). \label{SzAsym+1be}
\end{align}
As $z \to -1$, $\Im z > 0$, we have
\begin{align}\label{SzAsym-1}
& D_{\alpha}^{2}(z)\prod_{j=0}^{m+1}(t_{j}-z)^{-\alpha_{j}} =  1 + i\sqrt{2} \mathcal{A} \sqrt{z+1} - \mathcal{A}^{2}(z+1)+ \bigO((z+1)^{3/2}), \\
& D_{\beta}^{2}(z)e^{-i\pi\mathcal{B}} = 1 +  i\sqrt{2}  \widetilde{\mathcal{B}}_{-1} \sqrt{z+1} - \widetilde{\mathcal{B}}_{-1}^{2} (z+1) + \bigO((z+1)^{3/2}). \label{SzAsym-1be}
\end{align}
As $z \to \infty$, with $W \equiv 0$, we have
\begin{equation}\label{P_1^inf}
P_{1}^{(\infty)} = \begin{pmatrix}
\displaystyle \sum_{j=0}^{m+1} \left( \frac{\alpha_{j}t_{j}}{2}+i\sqrt{1-t_{j}^{2}}\beta_{j}\right) & \displaystyle \frac{i}{2} D_{\infty}^{2} \\
\displaystyle -\frac{i}{2}D_{\infty}^{-2} & \displaystyle - \sum_{j=0}^{m+1} \left( \frac{\alpha_{j}t_{j}}{2}+i\sqrt{1-t_{j}^{2}}\beta_{j}\right)
\end{pmatrix},
\end{equation}
where we recall that $\beta_{0} = \beta_{m+1} = 0$.

\subsection{Local parametrix near $t_k$, $1 \leq k \leq m$}\label{subsection: local param t_k}

It is well-known \cite{FouMarSou,ItsKrasovsky,DIK} that $P^{(t_{k})}$ can be written in terms of hypergeometric functions. In \cite{Charlier}, the local parametrix was obtained for Gaussian-type weights, and it is straightforward to adapt the construction for Laguerre-type and Jacobi-type weights, the only difference being in the definition of $\xi$. Let us define the function $f_{t_{k}}$ by
\begin{equation}\label{local_conformal_map}
f_{t_k}(z) = -2 \begin{cases}
\xi(z)-\xi_{+}(t_k), & \Im z > 0, \\
-(\xi(z)-\xi_{-}(t_k)), & \Im z < 0,
\end{cases} \ = 
2\pi i  \int_{t_k}^{z} \rho(s)ds,
\end{equation}
where in the above expression $\rho$ is the analytic continuation on $U_{V} \setminus ((-\infty,-1)\cup(1,+\infty))$ of the density of the equilibrium measure ($\rho$ was previously only defined on $[-1,1]$). This is a conformal map from $\mathcal{D}_{t_{k}}$ to a neighbourhood of $0$, and its expansion as $z \to t_{k}$ is given by
\begin{equation}\label{f_t_kAsymptotics}
f_{t_{k}}(z) = 2\pi i \rho(t_{k}) (z-t_{k})(1+\bigO(z-t_{k})), \qquad \mbox{ as } z \to t_{k}.
\end{equation}
\begin{figure}
	\centering
	
	\begin{tikzpicture}[scale=1.5]

	\draw[->-=0.15,->-=0.9,black,thick] (10.78,0) to (13.22,0);		
	
	\draw[->-=0.3,black,thick] (11,0.7) to [out=0, in=135] (12,0);
	\draw[->-=0.3,black,thick] (11,-0.7) to [out=0, in=-135] (12,0);

	\draw[->-=0.75,black,thick] (12,0) to [out=-45, in=-180] (13,-0.7);
	\draw[->-=0.75,black,thick] (12,0) to [out=45, in=180] (13,0.7);	
	
	\fill (12,0) circle (0.035cm);
	
	\node at (11.4,0.72) {$\ga_{+}$}	;	\node at (11.4,-0.7) {$\ga_{+}$};
	\node at (12.7,-0.72) {$\ga_{-}$}	;
	\node at (12.7,+0.72) {$\ga_{+}$}	;
	\node at (12.1,-0.2) {\textbf{$t_k$}};	
	
	\draw[dashed] (12,0) circle (1.22cm);		
	
	\node at (12,+1.35) {$\partial \mathcal{D}_{t_k}$};
	
	
	\draw[->-=0.99,thick] (13.5,0.7) to [out=45, in=135] (14.7,0.7);
	
	\node at (14.1,1.2) {$f_{t_{k}}$}	;
	
	
	\draw[dashed,black,thick] (15.25,0) to (17.75,0);
	
	\draw[->-=0.2,->-=0.8,black,thick] (16.5,-1.25) to (16.5,1.25);	
	\draw[->-=0.2,->-=0.8,black,thick] (15.25,-1.25) to (17.75,1.25);	
	\draw[->-=0.2,->-=0.8,black,thick] (17.75,-1.25) to (15.25,1.25);	
	
	\fill (16.5,0) circle (0.035cm);
	
	\node at (16.9,0.8) {$VIII$}	;
	\node at (17.5,0.4) {$VII$}	;
	\node at (16.9,-0.8) {$V$}	;
	\node at (17.5,-0.4) {$VI$}	;
	
	\node at (16.1,0.8) {$I$}	;
	\node at (15.5,0.4) {$II$}	;
	\node at (16.1,-0.8) {$IV$}	;
	\node at (15.5,-0.4) {$III$}	;
	
	\node at (16.5,1.4) {$\Ga_1$}	;
	\node at (15.25,1.4) {$\Ga_2$}	;
	\node at (16.5,-1.4) {$\Ga_5$}	;
	\node at (15.25,-1.4) {$\Ga_4$}	;		
	
	\node at (18,1.25) {$\Ga_8$}	;
	\node at (15,0) {$\Ga_3$}	;
	\node at (17.9,-1.25) {$\Ga_6$}	;
	\node at (17.9,0) {$\Ga_7$}	;				
	\node at (16.55,-0.13) {$0$}	;
	
	
	\coordinate (origin) at (16.5,0);
	
	\coordinate (west) at (15.25,0);
	
	\coordinate (northwest) at (15.25,1.25);

	\pic [draw, <->, "$\frac{\pi}{4}$", angle eccentricity=1.5] {angle = northwest--origin--west};
	\end{tikzpicture}
	
	\caption{The neighborhood $\mathcal{D}_{t_k}$ and its image under the mapping $f_{t_k}$.}
	\label{local-t_k}
\end{figure}
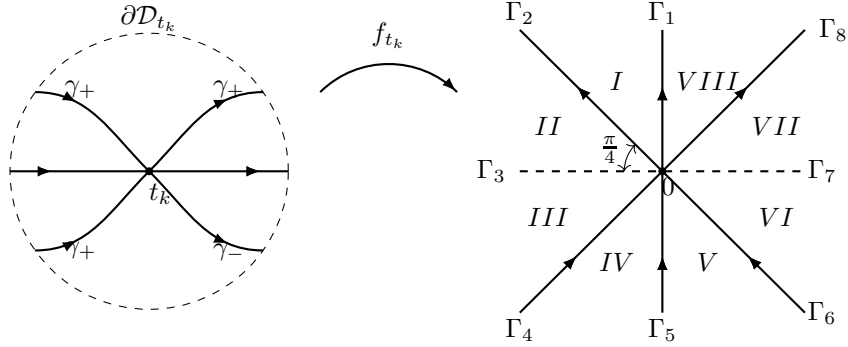
The lenses in a neighbourhood of $t_{k}$ are chosen such that $f_{t_{k}}(\gamma_{+} \cap \mathcal{D}_{t_{k}}) \subset \Gamma_{4} \cup \Gamma_{2}$ and $f_{t_{k}}(\gamma_{-} \cap \mathcal{D}_{t_{k}}) \subset \Gamma_{6} \cup \Gamma_{8}$, see Figure \ref{local-t_k}. Let us define $Q_{+,k}^{R} = f_{t_{k}}^{-1}(II)\cap \mathcal{D}_{t_{k}}$, that is, it is the subset of $\mathcal{D}_{t_{k}}$ that lies outside the lenses in the upper half plane and which is mapped by $f_{t_{k}}$ into a subset of $II$. All we need is to find the expression of $P^{(t_{k})}$ in the region $Q_{+,k}^{R}$. This was done in \cite[equation (4.48) and below (5.2)]{Charlier} for Gaussian-type weights. It is straightforward to adapt the construction in our situations, and we omit the details. For $z \in Q_{+,k}^{R}$, $P^{(t_k)}(z)$ is given by
\begin{multline}\label{P^{(t_k)}II}
P^{(t_k)}(z)=E_{t_k}(z) \times \\ \begin{pmatrix}
 \frac{\Gamma(1 + \frac{\alpha_{k}}{2}-\beta_{k})}{\Gamma(1+\alpha_{k})}G(\frac{\alpha_{k}}{2}+\beta_{k}, \alpha_{k}; nf_{t_{k}}(z))e^{-\frac{i\pi\alpha_{k}}{2}} & -\frac{\Gamma(1 + \frac{\alpha_{k}}{2}-\beta_{k})}{\Gamma(\frac{\alpha_{k}}{2}+\beta_{k})}H(1+\frac{\alpha_{k}}{2}-\beta_{k},\alpha_{k};nf_{t_{k}}(z)e^{-\pi i}) \\ 
 \frac{\Gamma(1 + \frac{\alpha_{k}}{2}+\beta_{k})}{\Gamma(1+\alpha_{k})}G(1+\frac{\alpha_{k}}{2}+\beta_{k},\alpha_{k};nf_{t_{k}}(z))e^{- \frac{i\pi\alpha_{k}}{2}}  & H(\frac{\alpha_{k}}{2}-\beta_{k},\alpha_{k};nf_{t_{k}}(z)e^{-\pi i})
\end{pmatrix} \\  \times (z-t_k)^{-\frac{\al_k}{2}\sigma_3}e^{\frac{\pi i \al_k}{4}\sigma_3} e^{-n\xi(z)\sigma_3}e^{-\frac{W(z)}{2}\sigma_3}\om_{t_k}(z)^{-\frac{\sigma_3}{2}},
\end{multline}
where $G$ and $H$ are given in terms of the Whittaker functions (see \cite[Chapter 13]{NIST}):
\begin{equation}\label{relation between G and H and Whittaker}
G(a,\alpha;z) = \frac{M_{\kappa,\mu}(z)}{\sqrt{z}}, \quad H(a,\alpha;z) = \frac{W_{\kappa,\mu}(z)}{\sqrt{z}}, \quad \mu = \frac{\alpha}{2}, \quad \kappa = \frac{1}{2}+\frac{\alpha}{2}-a.
\end{equation}
The function $E_{t_{k}}$ is analytic in $\mathcal{D}_{t_{k}}$  (see \cite[(4.49)-(4.51)]{Charlier}) and its value at $t_k$ is given by
\begin{multline}\label{E_{t_k}(t_k)}
\small E_{t_{k}}(t_{k}) = \frac{D_{\infty}^{\sigma_{3}}}{2 \sqrt[4]{1-t_{k}^{2}}}\begin{pmatrix}
e^{-\frac{\pi i}{4}}\sqrt{1+t_{k}} + e^{\frac{\pi i}{4}}\sqrt{1-t_{k}} & i \left( e^{-\frac{\pi i}{4}}\sqrt{1+t_{k}} - e^{\frac{\pi i}{4}}\sqrt{1-t_{k}} \right) \\
-i \left( e^{-\frac{\pi i}{4}}\sqrt{1+t_{k}} - e^{\frac{\pi i}{4}}\sqrt{1-t_{k}} \right) & e^{-\frac{\pi i}{4}}\sqrt{1+t_{k}} + e^{\frac{\pi i}{4}}\sqrt{1-t_{k}}
\end{pmatrix} \Lambda_{k}^{\sigma_{3}},
\end{multline}
where
\begin{equation}\label{def of Lambdak}
\Lambda_{k} = e^{\frac{W(t_{k})}{2}}D_{W,+}(t_{k})^{-1}e^{i \frac{\lambda_{k}}{2} } (4\pi \rho(t_{k})n(1-t_{k}^{2}))^{\beta_{k}} \prod_{1 \leq j\neq k \leq m} T_{kj}^{-\beta_{j}}, 
\end{equation}
and
\begin{equation}\label{def of lambdak}
\lambda_{k} = \mathcal{A} \arccos t_{k} - \frac{\pi}{2}\alpha_{k} - \sum_{j=k+1}^{m+1} \pi \alpha_{j} + 2\pi n \int_{t_{k}}^{1} \rho(s)ds.
\end{equation}
Also, we need a more detailed knowledge of the asymptotics \eqref{lol1}. By \cite[equation (4.52)]{Charlier}, we have
\begin{equation}\label{asymptotics on the disk D_t}
P^{(t_{k})}(z)P^{(\infty)}(z)^{-1} = I + \frac{v_{k}}{n f_{t_{k}}(z)} E_{t_{k}}(z) \begin{pmatrix}
-1 & \tau(\alpha_{k},\beta_{k}) \\ - \tau(\alpha_{k},-\beta_{k}) & 1
\end{pmatrix}E_{t_{k}}(z)^{-1} + \bigO (n^{-2+2|\Re\beta_{k}|}),
\end{equation}
uniformly for $z \in \partial \mathcal{D}_{t_k}$ as $n \to \infty$, where $v_{k} = \beta_{k}^{2}-\frac{\alpha_{k}^{2}}{4}$ and $\tau(\alpha_{k},\beta_{k}) = \frac{ - \Gamma(\frac{\alpha_{k}}{2}-\beta_{k})}{\Gamma(\frac{\alpha_{k}}{2}+\beta_{k}+1)}$.

\subsection{Local parametrix near $1$}\label{subsection: local param near 1}
The local parametrix near $1$ cannot be treated for both Laguerre-type and Jacobi-type weights simultaneously, since $1$ is a soft edge for Laguerre-type weights, and a hard edge for Jacobi-type weights. At a soft edge, the construction relies on the Airy model RH problem (whose solution is denoted $\Phi_{\mathrm{Ai}}$), and at a hard edge on the Bessel model RH problem (whose solution is denoted $\Phi_{\mathrm{Be}}$). For the reader's convenience, we recall these model RH problems in the appendix. 
\subsection*{Laguerre-type weights}
Let us define $f_{1}(z) = ( -\frac{3}{2}\xi(z) )^{2/3}$. This is a conformal map in $\mathcal{D}_{1}$ whose expansion as $z \to 1$ is given by
\begin{equation}\label{f_1 asymp L}
f_{1}(z) = \left(\frac{\pi \psi(1)}{\sqrt{2}}\right)^{2/3}(z-1)\left( 1-\frac{1}{10}\left( 1-4\frac{\psi^{\prime}(1)}{\psi(1)} \right)(z-1) + \bigO((z-1)^{2}) \right).
\end{equation}
The lenses $\gamma_{+}$ and $\gamma_{-}$ in a neighborhood of $1$ are chosen such that $f_{1}(\gamma_{+} \cap \mathcal{D}_{1})\subset e^{\frac{2\pi i}{3}}\mathbb{R}^{+}$ and $f_{1}(\gamma_{-} \cap \mathcal{D}_{1})\subset e^{-\frac{2\pi i}{3}}\mathbb{R}^{+}$. The local parametrix is given by
\begin{equation}
P^{(1)}(z) = E_{1}(z)\Phi_{\mathrm{Ai}}(n^{2/3}f_{1}(z))\omega(z)^{-\frac{\sigma_{3}}{2}}e^{-n\xi(z)\sigma_{3}}e^{-\frac{W(z)}{2}\sigma_{3}},
\end{equation}
where $E_{1}$ is analytic in $\mathcal{D}_{1}$ and given by
\begin{equation}
E_{1}(z) = P^{(\infty)}(z)e^{\frac{W(z)}{2}\sigma_{3}}\omega(z)^{\frac{\sigma_{3}}{2}}N^{-1}f_{1}(z)^{\frac{\sigma_{3}}{4}}n^{\frac{\sigma_{3}}{6}}, \qquad N = \frac{1}{\sqrt{2}}\begin{pmatrix}
1 & i \\ i & 1
\end{pmatrix},
\end{equation}
and $\Phi_{\mathrm{Ai}}(z)$ is the solution to the Airy model RH problem presented in the appendix (see Subsection \ref{subsection: model Airy}). Using \eqref{Asymptotics Airy}, we obtain a more detailed description of the matching condition (\ref{lol1}):
\begin{equation}\label{asymptotics on the disk D_1}
P^{(1)}(z)P^{(\infty)}(z)^{-1} = I + \frac{P^{(\infty)}(z)e^{\frac{W(z)}{2}\sigma_{3}}\omega(z)^{\frac{\sigma_{3}}{2}}}{8n f_{1}(z)^{3/2}}\begin{pmatrix}
\frac{1}{6} & i \\ i & -\frac{1}{6}
\end{pmatrix}\omega(z)^{-\frac{\sigma_{3}}{2}}e^{-\frac{W(z)}{2}\sigma_{3}}P^{(\infty)}(z)^{-1} + \bigO(n^{-2})
\end{equation}
uniformly for $z \in \partial \mathcal{D}_{1}$ as $n \to \infty$.

\subsection*{Jacobi-type weights}
In this case we define $f_{1}(z) = \xi(z)^{2}/4$. This is a conformal map in $\mathcal{D}_{1}$ whose expansion as $z \to 1$ is given by
\begin{equation}\label{f_1 asymp J}
f_{1}(z) = \left( \frac{\pi}{\sqrt{2}}\psi(1) \right)^{2}(z-1)\Big(1+\Big( \frac{2}{3}\frac{\psi^{\prime}(1)}{\psi(1)}-\frac{1}{6} \Big)(z-1)+ \bigO((z-1)^{2})\Big).
\end{equation}
The lenses $\gamma_{+}$ and $\gamma_{-}$ in a neighborhood of $1$  are again chosen such that $f_{1}(\gamma_{+} \cap \mathcal{D}_{1})\subset e^{\frac{2\pi i}{3}}\mathbb{R}^{+}$ and $f_{1}(\gamma_{-} \cap \mathcal{D}_{1})\subset e^{-\frac{2\pi i}{3}}\mathbb{R}^{+}$. The local parametrix is given by
\begin{equation}
P^{(1)}(z) = E_{1}(z) \Phi_{\mathrm{Be}}(n^{2}f_{1}(z);\alpha_{m+1})\omega_{1}(z)^{-\frac{\sigma_{3}}{2}}(z-1)^{-\frac{\alpha_{m+1}}{2}\sigma_{3}}e^{-n\xi(z)\sigma_{3}}e^{-\frac{W(z)}{2}\sigma_{3}},
\end{equation}
where the principal branch is taken for $(z-1)^{\frac{\alpha_{m+1}}{2}}$, $\Phi_{\mathrm{Be}}(z)$ is the solution to the Bessel model RH problem presented in Subsection \ref{ApB}, and $E_{1}$ is analytic in $\mathcal{D}_{1}$ and given by
\begin{equation}
E_{1}(z) = P^{(\infty)}(z)e^{\frac{W(z)}{2}\sigma_{3}}(z-1)^{\frac{\alpha_{m+1}}{2}\sigma_{3}}\omega_{1}(z)^{\frac{\sigma_{3}}{2}}N^{-1}(2\pi n f(z)^{1/2})^{\frac{\sigma_{3}}{2}}.
\end{equation}
In this case, using \eqref{large z asymptotics Bessel}, the matching condition (\ref{lol1}) can be written as 
\begin{multline}\label{asymptotics on the disk D_1J}
P^{(1)}(z)P^{(\infty)}(z)^{-1} = I + \frac{P^{(\infty)}(z)e^{\frac{W(z)}{2}\sigma_{3}}\omega_{1}(z)^{\frac{\sigma_{3}}{2}}(z-1)^{\frac{\alpha_{m+1}}{2}\sigma_{3}}}{16n f_{1}(z)^{1/2}} \\ \times \begin{pmatrix}
-(1+4\alpha_{m+1}^{2}) & -2i \\ -2i & 1+4\alpha_{m+1}^{2}
\end{pmatrix}(z-1)^{-\frac{\alpha_{m+1}}{2}\sigma_{3}}\omega_{1}(z)^{-\frac{\sigma_{3}}{2}}e^{-\frac{W(z)}{2}\sigma_{3}}P^{(\infty)}(z)^{-1} + \bigO(n^{-2}),
\end{multline}
uniformly for $z \in \partial \mathcal{D}_{1}$ as $n \to \infty$.
\subsection{Local parametrix near $-1$}\label{subsection: local param near -1}
Since Laguerre-type and Jacobi-type weights both have a hard edge at $-1$, the construction of this local parametrix can be treated simultaneously for both cases, the only difference being in the conformal map. This map is defined by $f_{-1}(z) = -(\xi(z)-\pi i)^{2}/4$, and its expansion as $z \to -1$ is given by
\begin{equation}\label{f_{-1}-asymp}
f_{-1}(z) = \small \left\{ \begin{array}{l l}
\ds \big(\sqrt{2}\pi \psi(-1)\big)^{2}(z+1)\Big(1+\Big( \frac{2}{3}\frac{\psi^{\prime}(-1)}{\psi(-1)}-\frac{1}{6} \Big)(z+1) + \bigO((z+1)^{2})\Big), & \mbox{for Laguerre-type weights}, \\[0.3cm]
\ds \Big(\di \frac{\pi}{\sqrt{2}} \psi(-1)\Big)^{2}(z+1)\Big(1+\Big( \frac{2}{3}\frac{\psi^{\prime}(-1)}{\psi(-1)}+\frac{1}{6} \Big)(z+1) + \bigO((z+1)^{2})\Big), & \mbox{for Jacobi-type weights}.
\end{array} \right.
\end{equation}
The local parametrix is given by
\begin{equation}\label{P^{(-1)}}
P^{(-1)}(z) = E_{-1}(z)\sigma_{3}\Phi_{\mathrm{Be}}(-n^{2}f_{-1}(z);\alpha_{0})\sigma_{3}\omega_{-1}(z)^{-\frac{\sigma_{3}}{2}}(-z-1)^{-\frac{\alpha_{0}}{2}\sigma_{3}}e^{-n\xi(z)\sigma_{3}}e^{-\frac{W(z)}{2}\sigma_{3}},
\end{equation}
where the principal branch is chosen for $(-z-1)^{-\frac{\alpha_{0}}{2}\sigma_{3}}$, and $E_{-1}$ is analytic in $\mathcal{D}_{-1}$ and given by
\begin{equation}\label{E-1 general}
E_{-1}(z) = (-1)^{n}P^{(\infty)}(z)e^{\frac{W(z)}{2}\sigma_{3}}\omega_{-1}(z)^{\frac{\sigma_{3}}{2}}(-z-1)^{\frac{\alpha_{0}}{2}\sigma_{3}}N(2\pi n (-f_{-1}(z))^{1/2})^{\frac{\sigma_{3}}{2}}.
\end{equation}
For Laguerre-type weights with $W \equiv 0$, by taking the limit $z \to -1$ in \eqref{E-1 general} (from e.g. the upper half plane) and using the asymptotics \eqref{SzAsym-1}--\eqref{SzAsym-1be} we have 
\begin{equation}\label{E_{-1}(-1)}
E_{-1}(-1)=(-1)^nD_{\infty}^{\sigma_3}\Big(N + \begin{pmatrix}
0 & \frac{i}{\sqrt{2}}(\mathcal{A}+\widetilde{\mathcal{B}}_{-1}) \\
0 & \frac{-1}{\sqrt{2}}(\mathcal{A}+\widetilde{\mathcal{B}}_{-1})
\end{pmatrix}\Big)(4\pi^2\psi(-1)n)^{\frac{\sigma_3}{2}}.    
\end{equation}
Furthermore, as $n \to \infty$, we have
\begin{multline}\label{asymptotics on the disk D_-1}
P^{(-1)}(z)P^{(\infty)}(z)^{-1} = I + \frac{P^{(\infty)}(z)e^{\frac{W(z)}{2}\sigma_{3}}\omega_{-1}(z)^{\frac{\sigma_{3}}{2}}(-z-1)^{\frac{\alpha_{0}}{2}\sigma_{3}}}{16 n (-f_{-1}(z))^{1/2}} \\ \times \begin{pmatrix}
-(1+4\alpha_{0}^{2}) & 2i \\ 2i & 1+4\alpha_{0}^{2}
\end{pmatrix}(-z-1)^{-\frac{\alpha_{0}}{2}\sigma_{3}}\omega_{-1}(z)^{-\frac{\sigma_{3}}{2}}e^{-\frac{W(z)}{2}\sigma_{3}}P^{(\infty)}(z)^{-1} + \bigO(n^{-2}),
\end{multline}
uniformly for $z \in \partial \mathcal{D}_{-1}$.
	
\subsection{Small norm RH problem}\label{subsection: small norm}
We are now in a position to do the last transformation. We recall that the disks are nonoverlapping. Using the parametrices, we define the matrix valued function $R$ as
\begin{equation}
R(z) = \left\{ \begin{array}{l l}
S(z)P^{(\infty)}(z)^{-1}, & \mbox{if }z \in \mathbb{C}\setminus \cup_{j=0}^{m+1}\overline{  \mathcal{D}_{t_j}}, \\
S(z)P^{(t_{j})}(z)^{-1}, & \mbox{if }z \in \mathcal{D}_{t_j}, \quad j  = 0,\ldots,m+1.
\end{array} \right.
\end{equation}
We recall that the local parametrices have the same jumps as $S$ inside the disks and also that the global parametrix has the same jumps as $S$ on $(-1,1)$, hence $R$ has jumps only on the contour $\Sigma_R$ depicted in Figure \ref{R-contour}, where the orientation of the jump contour on $\partial \mathcal{D}_{t_j}$ is chosen to be clockwise.  Since $P^{(t_j)}$ and $S$ have the same asymptotic behavior near $t_j$, $j=0,\ldots,m+1$, $R$ is bounded at these points. Therefore, it satisfies the following RH problem.
\subsubsection*{RH problem for $R$}
\begin{itemize}
	\item[(a)] $R: \C \setminus \Sigma_R \to \C^{2 \times 2}$ is analytic.
	\item[(b)] $R$ satisfies $R_+(z)=R_-(z)J_R(z)$ for $z$ on $\Sigma_R\setminus \{\mbox{intersection points of }\Sigma_R\}$ with
	\begin{align}\label{J_R}
	 J_R(z) = \begin{cases}
	 P^{(t_j)}(z)P^{(\infty)}(z)^{-1} & z \in \partial \mathcal{D}_{t_j}, \\
	 P^{(\infty)}(z)J_S(z)P^{(\infty)}(z)^{-1} & z \in \Sigma_R \setminus \cup^{m+1}_{j=0} \partial \mathcal{D}_{t_j},
	 \end{cases}
	\end{align}
where $J_{S}(z) := S_{-}^{-1}(z)S_{+}(z)$ is given in \eqref{J_S 1}--\eqref{J_S 3}.
	\item[(c)] As $z \to \infty$, $R(z) = I + R_1 z^{-1} + \bigO(z^{-2})$ for a certain matrix $R_{1}$ independent of $z$.
	
As $z\to z_{\star} \in \{\mbox{intersections points of } \Sigma_{R} \}$, $R(z)$ is bounded.
\end{itemize}	
We recall that outside fixed neighbourhoods of $t_{j}$, $j = 0,\ldots,m+1$, the jumps for $S$ on $\gamma_{+}\cup \gamma_{-}$ and on $\mathcal{I}\setminus [-1,1]$ are exponentially and uniformly close to the identity matrix (see the discussion at the end of Subsection \ref{Subsection: T to S}). Therefore, from (\ref{asymptotics on the disk D_t}), (\ref{asymptotics on the disk D_1}),(\ref{asymptotics on the disk D_1J}), (\ref{asymptotics on the disk D_-1}) and (\ref{J_R}), as $n \to \infty$ we have
 
\begin{equation}\label{J_R asymp}
 J_R(z) = \begin{cases}
 I+\mathcal{O}(e^{-cn}), & \mbox{uniformly for } z \in \Sigma_R \cap (\ga^+ \cup \ga^- \cup \R), \\
 I+\mathcal{O}(n^{-1}), & \mbox{uniformly for }  z \in \partial \mathcal{D}_{1} \cup \partial \mathcal{D}_{-1}, \\
 I+\mathcal{O}(n^{-1+2|\Re \be_k|}), & \mbox{uniformly for }  z \in \partial \mathcal{D}_{t_k} , k=1,\ldots,m,
 \end{cases}
 \end{equation}
for a positive constant $c$. By standard theory of small-norm RH problems (see e.g.  \cite{Deiftetal, Deiftetal2}), $R$ exists for sufficiently large $n$ (we also refer to \cite{Krasovsky, ItsKrasovsky, DeiftItsKrasovsky, Charlier} for very similar situations with more details provided). Furthermore, for any $r \in \mathbb{N}$, as $n \to \infty$, $R$ has an expansion given by
\begin{align}
& R(z) = I + \sum_{j=1}^{r} \frac{R^{(j)}(z)}{n^j} + R_R^{(r+1)}(z)n^{-r-1},  \label{R_large_n_asymp} \\
& R^{(j)}(z) = \bigO(n^{2\beta_{\max}}), \quad R^{(j)}(z)^{\prime} = \bigO(n^{2\beta_{\max}}) \quad R^{(r+1)}_{R}(z) = \bigO(n^{2\beta_{\max}}), \quad R^{(r+1)}_{R}(z)^{\prime} = \bigO(n^{2\beta_{\max}}), \nonumber 
\end{align}
uniformly for $z \in \mathbb{C}\setminus \Sigma_{R}$, uniformly for $(\vec{\al},\vec{\be})$ in any fixed compact set, and uniformly in $\vec{t}$ if there exists $\delta > 0$, independent of $n$, such that 
\begin{equation}\label{delta def in small norm section}
\min_{j \neq k} \{ |t_j-t_k|, |t_j-1|, |t_j+1| \} \geq \de.
\end{equation} 
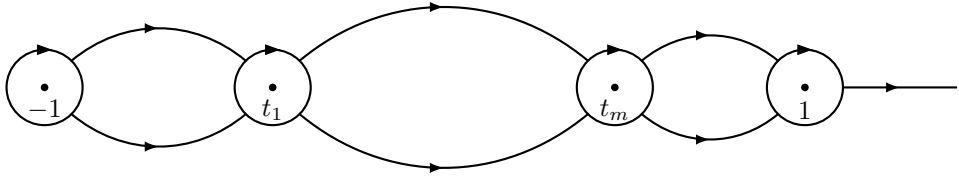
\begin{figure}
\centering
\begin{tikzpicture}
\fill (0,0) circle (0.05cm);
\fill (3,0) circle (0.05cm);
\fill (7.5,0) circle (0.05cm);
\fill (10,0) circle (0.05cm);
\draw [thick] (0,0) circle (0.5cm);
\draw [thick] (3,0) circle (0.5cm);
\draw [thick] (7.5,0) circle (0.5cm);
\draw [thick] (10,0) circle (0.5cm);
\draw [->-=0.5,thick] (10.5,0)--(12,0);
\draw[black,arrows={-Triangle[length=0.2cm,width=0.14cm]}]
($(0.1,0)+(90:0.5)$) --  ++(0.0001,0);
\draw[black,arrows={-Triangle[length=0.2cm,width=0.14cm]}]
($(3.1,0)+(90:0.5)$) --  ++(0.0001,0);
\draw[black,arrows={-Triangle[length=0.2cm,width=0.14cm]}]
($(7.6,0)+(90:0.5)$) --  ++(0.0001,0);
\draw[black,arrows={-Triangle[length=0.2cm,width=0.14cm]}]
($(10.1,0)+(90:0.5)$) --  ++(0.0001,0);

\node at (0,-0.3) {$-1$};
\node at (3,-0.3) {$t_{1}$};
\node at (7.5,-0.3) {$t_{m}$};
\node at (10,-0.3) {$1$};

\draw[->-=0.5,black,thick] ($(0,0)+(45:0.5)$) to [out=40, in=140] ($(3,0)+(135:0.5)$);
\draw[->-=0.5,black,thick] ($(3,0)+(45:0.5)$) to [out=40, in=140] ($(7.5,0)+(135:0.5)$);
\draw[->-=0.5,black,thick] ($(7.5,0)+(45:0.5)$) to [out=40, in=140] ($(10,0)+(135:0.5)$);

\draw[->-=0.5,black,thick] ($(0,0)+(-45:0.5)$) to [out=-40, in=-140] ($(3,0)+(-135:0.5)$);
\draw[->-=0.5,black,thick] ($(3,0)+(-45:0.5)$) to [out=-40, in=-140] ($(7.5,0)+(-135:0.5)$);
\draw[->-=0.5,black,thick] ($(7.5,0)+(-45:0.5)$) to [out=-40, in=-140] ($(10,0)+(-135:0.5)$);
\end{tikzpicture}
\caption{Jump contour $\Sigma_R$ for the RH problem for $R$ for Laguerre-type weights with $m=2$. For Jacobi-type weights, $\Sigma_R$ is of the same shape except that there are no jumps on $(1,\infty)\setminus \overline{\mathcal{D}_{1}}$.}
\label{R-contour}
\end{figure}
Furthermore, in the way as done in \cite{Charlier}, we show that
\begin{equation}\label{R^{(1)} derivative order}
\partial_{\nu} R^{(j)}(z) = \bigO (n^{2 \be_{max}}\log n), \qquad \partial_{\nu} R_{R}^{(r+1)}(z) = \bigO (n^{2 \be_{max}}\log n)
\end{equation}
for $\nu \in \{\alpha_{0},\alpha_{1},\ldots,\alpha_{m+1},\beta_{1},\ldots,\beta_{m}\}$.
From (\ref{asymptotics on the disk D_t}), (\ref{asymptotics on the disk D_1}),(\ref{asymptotics on the disk D_1J}), (\ref{asymptotics on the disk D_-1}), we show that $J_{R}$ admits an expansion as $n \to + \infty$ of the form
\begin{equation}
J_{R}(z) = I +  \sum_{j=1}^{r} \frac{J_R^{(j)}(z)}{n^j} + \bigO
(n^{-r-1+2\be_{max}}) , \qquad J_R^{(j)}(z) = \bigO(n^{2\beta_{\max}}),
\end{equation}
uniformly for $z \in \cup_{j=0}^{m+1}\partial \mathcal{D}_{t_{j}}$. The matrices $R^{(j)}$ are obtained in a recursive way via the Plemelj-Sokhotski formula (for instance see \cite{KMcLVAV}), in particular one has
\begin{equation}\label{R^(1)}
R^{(1)}(z) = \sum_{j=0}^{m+1}\frac{1}{2\pi i} \int_{\partial \mathcal{D}_{t_{j}}} \frac{J_{R}^{(1)}(s)}{s-z}ds,
\end{equation}
where we recall that the orientation on $\partial \mathcal{D}_{t_{j}}$ is clockwise.
The goal for the rest of this section is to explicitly compute $R^{(1)}$ in the case $W \equiv 0$ for Laguerre-type and Jacobi-type weights.
\subsection*{Laguerre-type weights}
From \eqref{asymptotics on the disk D_t}, \eqref{asymptotics on the disk D_1},  and \eqref{asymptotics on the disk D_-1} we easily show that $J_{R}^{(1)}$ has a double pole at $1$ and a simple pole at $t_j$, $j=0,\ldots,m$. Therefore $R^{(1)}(z)$ can be explicitly computed from (\ref{R^(1)}) via a residue calculation. For $z \in \mathbb{C}\setminus \cup_{j=0}^{m+1}\overline{\mathcal{D}_{t_j}}$, we have
\begin{equation}\label{R^{(1)}}
\begin{array}{r c l}
\displaystyle R^{(1)}(z) & = & \displaystyle \sum_{j=1}^{m} \frac{1}{z-t_{j}} \mbox{Res}\big(J_{R}^{(1)}(s),s=t_{j}\big) +\frac{1}{z+1} \mbox{Res}\big(J_{R}^{(1)}(s),s=-1\big) \\
 &  & \displaystyle +\frac{1}{z-1} \mbox{Res}\big(J_{R}^{(1)}(s),s=1\big)+ \frac{1}{(z-1)^{2}} \mbox{Res}\big((s-1)J_{R}^{(1)}(s),s=1\big). \\
\end{array}
\end{equation}
The residue at $t_{k}$ can be computed from \eqref{asymptotics on the disk D_t} (in the same way as in \cite[eq (4.82)]{Charlier})
\begin{equation}\label{res t_{k}}
\mbox{Res}\big(J_{R}^{(1)}(z),z=t_{k}\big) = \frac{v_{k} D_{\infty}^{\sigma_{3}}}{2\pi \rho(t_{k})\sqrt{1-t_{k}^{2}}}\begin{pmatrix}
t_{k} + \widetilde{\Lambda}_{I,k} & -i - i\widetilde{\Lambda}_{R,2,k} \\
-i + i \widetilde{\Lambda}_{R,1,k} & - t_{k} - \widetilde{\Lambda}_{I,k}
\end{pmatrix}D_{\infty}^{-\sigma_{3}},
\end{equation}
where
\begin{align}
& \widetilde{\Lambda}_{I,k} = \frac{\tau(\alpha_{k},\beta_{k})\Lambda_{k}^{2}-\tau(\alpha_{k},-\beta_{k})\Lambda_{k}^{-2}}{2i}, \\[0.3cm]
& \widetilde{\Lambda}_{R,1,k} = \frac{\tau(\alpha_{k},\beta_{k})\Lambda_{k}^{2}e^{i\arcsin t_{k}}+\tau(\alpha_{k},-\beta_{k})\Lambda_{k}^{-2}e^{-i\arcsin t_{k}}}{2}, \\[0.3cm] & \widetilde{\Lambda}_{R,2,k} = \frac{\tau(\alpha_{k},\beta_{k})\Lambda_{k}^{2}e^{-i\arcsin t_{k}}+\tau(\alpha_{k},-\beta_{k})\Lambda_{k}^{-2}e^{i\arcsin t_{k}}}{2}.
\end{align}
Furthermore, we note the following relation:
\begin{align}
& \widetilde{\Lambda}_{R,1,k} - \widetilde{\Lambda}_{R,2,k} = -2t_{k} \widetilde{\Lambda}_{I,k}. \label{connection Lambdas}
\end{align}
Now let us compute the other terms in (\ref{R^{(1)}}). We compute the residue at $-1$ from (\ref{Global_Parametrix}), (\ref{SzAsym-1}), (\ref{SzAsym-1be}), (\ref{f_{-1}-asymp}) and (\ref{asymptotics on the disk D_-1}), and we find
\begin{equation}\label{res at -1L}
\mbox{Res}\big(J_{R}^{(1)}(z),z=-1\big) = \frac{1-4\alpha_{0}^{2} }{2^{5} \pi \psi(-1)}D_{\infty}^{\sigma_{3}}\begin{pmatrix}
-1 & -i \\
-i & 1
\end{pmatrix}D_{\infty}^{-\sigma_{3}}.
\end{equation}
Similarly, from (\ref{Global_Parametrix}), (\ref{SzAsym+1}), (\ref{SzAsym+1be}), (\ref{f_1 asymp L}) and (\ref{asymptotics on the disk D_1}) we obtain 
\begin{equation}\label{res at 1'L}
\mbox{Res}\big((z-1)J_{R}^{(1)}(z),z=1\big) = \frac{5}{2^{4} 3 \pi \psi(1)} D_{\infty}^{\sigma_{3}} \begin{pmatrix}
-1 & i \\ i & 1
\end{pmatrix} D_{\infty}^{-\sigma_{3}},
\end{equation}
and
\begin{multline}\label{res at 1L}
\mbox{Res}\big(J_{R}^{(1)}(z),z=1\big) = \frac{D_{\infty}^{\sigma_{3}}}{2^{5}\pi \psi(1)} \times \\[0.3cm] \hspace{0.5cm} \begin{pmatrix}
-4(\mathcal{A}-\widetilde{\mathcal{B}_{1}})^{2} + 1 + 2\frac{\psi^{\prime}(1)}{\psi(1)}  & 4i \left( (\mathcal{A}-\widetilde{\mathcal{B}_{1}})^{2} + 2 (\mathcal{A}-\widetilde{\mathcal{B}_{1}}) + \frac{11}{12} - \frac{1}{2}\frac{\psi^{\prime}(1)}{\psi(1)} \right) \\
4i \left( (\mathcal{A}-\widetilde{\mathcal{B}_{1}})^{2} - 2 (\mathcal{A}-\widetilde{\mathcal{B}_{1}}) + \frac{11}{12} - \frac{1}{2}\frac{\psi^{\prime}(1)}{\psi(1)} \right) & 4(\mathcal{A}-\widetilde{\mathcal{B}_{1}})^{2} - 1 -2\frac{\psi^{\prime}(1)}{\psi(1)}
\end{pmatrix}D_{\infty}^{-\sigma_{3}}.
\end{multline}
The quantity $R^{(1)}(-1)$ will also play an important role in Section \ref{Section: FH integration}. From another residue calculation, we obtain
\begin{equation}\label{R^(1) Laguerre at -1}
\begin{array}{r c l}
\displaystyle R^{(1)}(-1) & = & \displaystyle \sum_{j=1}^{m} \frac{-1}{1+t_{j}} \mbox{Res}\big(J_{R}^{(1)}(s),s=t_{j}\big) - \mbox{Res}\big(\frac{J_{R}^{(1)}(s)}{s+1},s=-1\big) \\
 &  & \displaystyle -\frac{1}{2} \mbox{Res}\big(J_{R}^{(1)}(s),s=1\big)+ \frac{1}{4} \mbox{Res}\big((s-1)J_{R}^{(1)}(s),s=1\big). \\
\end{array}
\end{equation}
In (\ref{res t_{k}}), (\ref{res at 1'L}) and (\ref{res at 1L}) we have already computed the above residues at $t_{1},\ldots,t_{m}$ and at $1$, the other residue at $-1$ can be computed from (\ref{Global_Parametrix}), (\ref{SzAsym-1})--\eqref{SzAsym-1be}, (\ref{f_{-1}-asymp}) and (\ref{asymptotics on the disk D_-1})  from which we obtain:
\begin{multline}\label{new res -1 L}
\mbox{Res}\big(\frac{J_{R}^{(1)}(s)}{s+1},s=-1\big) =  \frac{D_{\infty}^{\sigma_{3}}}{2^{3}3 \pi \psi(-1)} \left( \begin{array}{l}
\frac{3}{2}(\mathcal{A}+\mathcal{B}_{-1})^{2}-2\alpha_{0}^{2}-1+\frac{1-4\alpha_{0}^{2}}{4}\frac{\psi^{\prime}(-1)}{\psi(-1)} \\
i \Big( \frac{3}{2}(\mathcal{A}+ \mathcal{B}_{-1})^{2}-3(\mathcal{A}+ \mathcal{B}_{-1}) + \alpha_{0}^{2} + \frac{5}{4}+\frac{1-4\alpha_{0}^{2}}{4}\frac{\psi^{\prime}(-1)}{\psi(-1)} \Big)
\end{array} \right.  \\
\cdots \left. \begin{array}{l}
 i \Big( \frac{3}{2}(\mathcal{A}+ \mathcal{B}_{-1})^{2}+3(\mathcal{A}+ \mathcal{B}_{-1}) + \alpha_{0}^{2} + \frac{5}{4}+\frac{1-4\alpha_{0}^{2}}{4}\frac{\psi^{\prime}(-1)}{\psi(-1)} \Big) \\ -\frac{3}{2}(\mathcal{A}+\mathcal{B}_{-1})^{2}+2\alpha_{0}^{2}+1-\frac{1-4\alpha_{0}^{2}}{4}\frac{\psi^{\prime}(-1)}{\psi(-1)}
\end{array} \right)D_{\infty}^{-\sigma_{3}}.
\end{multline}

\subsection*{Jacobi-type weights} 
In this case $J^{(1)}_R(z)$ has simple poles at all $t_j$, $j=0,1,\ldots, m+1$ as can be seen from  (\ref{asymptotics on the disk D_t}), (\ref{asymptotics on the disk D_1J}),  and (\ref{asymptotics on the disk D_-1}). For $z$ outside all of the disks $\mathcal{D}_{t_j}$, $j=0,1,\ldots, m+1$, we have
\begin{equation}\label{R^{(1)}Jacobi}
\begin{array}{r c l}
\displaystyle R^{(1)}(z) & = & \displaystyle \sum_{j=1}^{m} \frac{1}{z-t_{j}} \mbox{Res}\big(J_{R}^{(1)}(s),s=t_{j}\big) +\frac{1}{z+1} \mbox{Res}\big(J_{R}^{(1)}(s),s=-1\big) \\
 &  & \displaystyle +\frac{1}{z-1} \mbox{Res}\big(J_{R}^{(1)}(s),s=1\big). \\
\end{array}
\end{equation}
Here the residue at $t_{k}$ is again given by \eqref{res t_{k}} (with $\rho$ given by \eqref{def of rho}). The residues at $-1$ can be computed from (\ref{Global_Parametrix}), (\ref{SzAsym-1}), (\ref{SzAsym-1be}), (\ref{f_{-1}-asymp}) and (\ref{asymptotics on the disk D_-1}) and is given by
\begin{equation}
\mbox{Res}\big(J_{R}^{(1)}(z),z=-1\big) = \frac{1-4\alpha_{0}^{2}}{2^{4} \pi \psi(-1)}D_{\infty}^{\sigma_{3}}\begin{pmatrix}
-1 & -i \\
-i & 1
\end{pmatrix}D_{\infty}^{-\sigma_{3}}.
\end{equation}
Similarly, from (\ref{Global_Parametrix}), (\ref{SzAsym+1}), (\ref{SzAsym+1be}), (\ref{f_1 asymp J}) and (\ref{asymptotics on the disk D_1J}) we obtain the residue at $1$:

\begin{equation}
\mbox{Res}\big(J_{R}^{(1)}(z),z=1\big) = \frac{1-4\alpha_{m+1}^{2}}{2^{4}  \pi \psi(1)}D_{\infty}^{\sigma_{3}}\begin{pmatrix}
1 & -i \\
-i & -1
\end{pmatrix}D_{\infty}^{-\sigma_{3}}.
\end{equation}

\section{Starting points of integration}\label{Section: starting points of integrations}
Since we will find large $n$ asymptotics only for the logarithmic derivative of Hankel determinant, we still face the classical problem of finding a good starting point for the integration. It turns out that in our case, it can be obtained by a direct computation, using some known results in the literature concerning standard Laguerre and Jacobi polynomials, and using the formula \eqref{det as product}.
\begin{lemma}
As $n \to \infty$, we have
\begin{multline}\label{Laguerre starting point}
\log L_{n}((\alpha_{0},0,\ldots,0),\vec{0},2(x+1),0) = \left( -\frac{3}{2}-\log 2 \right)n^{2} + \left( \log(2\pi)-\alpha_{0}(1+\log 2) \right)n \\ + \left( \frac{\alpha_{0}^{2}}{2}-\frac{1}{6} \right) \log n + \frac{\alpha_{0}}{2} \log(2\pi) + 2\zeta^{\prime}(-1) - \log G(1+\alpha_{0}) + \bigO(n^{-1}).
\end{multline}
As $n \to \infty$, we have
\begin{multline}\label{Jacobi starting point}
\log J_{n}((\alpha_{0},0,\ldots,0,\alpha_{m+1}),\vec{0},0,0) = - n^{2} \log 2 + [(1-\alpha_{0}-\alpha_{m+1})\log 2 + \log \pi]n + \frac{2\alpha_{0}^{2}+2\alpha_{m+1}^{2}-1}{4}\log n \\ - \log (G(1+\alpha_{0})G(1+\alpha_{m+1})) + 3\zeta^{\prime}(-1)+\left( \frac{1}{12}-\frac{(\alpha_{0}+\alpha_{m+1})^{2}}{2}\right)\log 2  + \frac{\alpha_{0}+\alpha_{m+1}}{2} \log(2\pi) + \bigO(n^{-1}).
\end{multline}
\end{lemma}
\begin{proof}
From \cite[equations (5.1.1) and (5.1.8)]{Szego OP}, the orthonormal polynomials of degree $k$ with respect to the weight $e^{-x}x^{\alpha_{0}}$ (supported on $(0,\infty)$) has a leading coefficient given by 
\begin{equation*}
\frac{(-1)^{k}}{\sqrt{k! \, \Gamma(k + \alpha_{0}+1)}}.
\end{equation*}
Therefore, by a simple change of variables, the degree $k$ orthonormal polynomials with respect to the weight $(x+1)^{\alpha_{0}}e^{-2n(x+1)}$ (supported on $(-1,\infty)$) has a leading coefficient given by
\begin{equation*}
\frac{(-1)^{k}(2n)^{k + \frac{1+\alpha_{0}}{2}}}{\sqrt{k! \, \Gamma(k+\alpha_{0}+1)}}.
\end{equation*}
By applying formula \eqref{det as product} for this weight, one obtains that
\begin{multline}\label{explciit Laguerre}
L_{n}((\alpha_{0},0,\ldots,0),\vec{0},2(x+1),0) = (2n)^{-n(n+\alpha_{0})} \prod_{k=1}^{n} \Gamma(k+\alpha_{0})\Gamma(k) \\ = (2n)^{-n(n+\alpha_{0})}\frac{G(n+1)G(n+\alpha_{0}+1)}{G(1+\alpha_{0})},
\end{multline}
where we have used $G(z+1) = \Gamma(z)G(z)$. The Barne's $G$-function has a known asymptotics for large argument (see \cite[eq (5.17.5)]{NIST}). As $z \to \infty$ with $|\arg z|< \pi$, we have
\begin{equation}\label{logG(z+1) large z}
\log G(z+1) = \frac{z^{2}}{4}+z \log \Gamma(z+1)-\left(\frac{z(z+1)}{2}+\frac{1}{12}\right)\log z - \frac{1}{12} + \zeta^{\prime}(-1) + \bigO(z^{-2}).
\end{equation}
The asymptotics of $\log\Ga(z)$ are given by
\begin{equation}\label{logGamma(z) large z}
\log \Gamma(z) = (z-\tfrac{1}{2})\log z - z + \tfrac{1}{2}\log(2\pi) + \frac{1}{12z} + \bigO(z^{-3}), \qquad \mbox{as } z \to \infty, \quad |\arg z| < \pi,
\end{equation}
(see \cite[eq (5.11.1)]{NIST}). We obtain \eqref{Laguerre starting point} by using the above asymptotic formulas in \eqref{explciit Laguerre}. Similarly, from \cite[equations (4.3.3) and (4.21.6)]{Szego OP}, the degree $k$ orthonormal polynomial with respect to the weight $(1-x)^{\alpha_{m+1}}(1+x)^{\alpha_{0}}$ has a leading coefficient given by
\begin{equation*}
\frac{2^{-k} \sqrt{2k + \alpha_{0} + \alpha_{m+1} + 1} \, \, \Gamma(2k + \alpha_{0} + \alpha_{m+1} + 1)}{\sqrt{2^{\alpha_{0} + \alpha_{m+1}+1}\Gamma(k+1)\Gamma(k + \alpha_{0} + 1)\Gamma(k + \alpha_{m+1} + 1)\Gamma(k + \alpha_{0} + \alpha_{m+1} + 1)}}
\end{equation*}
By applying formula \eqref{det as product} to this weight, one obtains
\begin{multline*}
J_{n}((\alpha_{0},0,\ldots,0,\alpha_{m+1}),\vec{0},0,0) = 2^{n^{2}+n(\alpha_{0}+\alpha_{m+1})} \\ \times \prod_{k=0}^{n-1} \frac{\Gamma(k+1)\Gamma(k + \alpha_{0}+1)\Gamma(k+\alpha_{m+1}+1)\Gamma(k + \alpha_{0}+\alpha_{m+1}+1)}{\Ga(2k + \alpha_{0} + \alpha_{m+1}+1)\Gamma(2k + \alpha_{0} + \alpha_{m+1}+2)}.
\end{multline*}
Using the functional equation $G(z+1) = \Gamma(z)G(z)$ we can simplify the above product. We obtain
\begin{multline}\label{Jn Jacobi inside proof}
J_{n}((\alpha_{0},0,\ldots,0,\alpha_{m+1}),\vec{0},0,0) =  2^{n^{2}+n(\alpha_{0}+\alpha_{m+1})} \\ \times \frac{G(n+1)G(n+\alpha_{0}+1)G(n+\alpha_{m+1}+1)G(n+\alpha_{0}+\alpha_{m+1}+1)}{G(1+\alpha_{0})G(1+\alpha_{m+1})G(2n+\alpha_{0}+\alpha_{m+1}+1)}.
\end{multline}
We obtain \eqref{Jacobi starting point} by expanding \eqref{Jn Jacobi inside proof} as $n \to + \infty$, using the asymptotic formulas \eqref{logG(z+1) large z} and \eqref{logGamma(z) large z}.
\end{proof}
As mentioned in the outline, large $n$ asymptotics for $J_{n}(\vec{\alpha},\vec{\beta},0,0)$ are known in the literature, and we reproduce the precise statement here.
\begin{theorem}\label{DIK-Starting Point Jacobi}(Deift-Its-Krasovsky \cite{DIK}).
As $n \to \infty$, we have
\begin{align}\label{Jacobi asymptotics of DIK}
&\log \frac{J_{n}(\vec{\alpha},\vec{\beta},0,0)}{ J_{n}(\vec{0},\vec{0},0,0)}  =   \bigg[ 2i \sum_{j=1}^{m}\beta_{j} \arcsin t_{j} - \mathcal{A} \log 2 \bigg] n +  \bigg[ \frac{\alpha_{0}^{2}+\alpha_{m+1}^{2}}{2} + \sum_{j=1}^{m} \Big( \frac{\alpha_{j}^{2}}{4}-\beta_{j}^{2} \Big) \bigg]\log n \nonumber \\ &  + i \mathcal{A} \sum_{j=1}^{m}\beta_{j} \arcsin t_{j} + \frac{i \pi}{2} \sum_{0 \leq j < k \leq m+1} (\alpha_{k}\beta_{j}-\alpha_{j}\beta_{k}) + \frac{\alpha_{0}+\alpha_{m+1}}{2} \log(2\pi) - \frac{\alpha_{0}^{2}+\alpha_{m+1}^{2}}{2}\log 2 \nonumber \\ &  +\sum_{0\leq j < k \leq m+1} \log  \Bigg(  \frac{\big(1-t_{j}t_{k}-\sqrt{(1-t_{j}^{2})(1-t_{k}^{2})}\big)^{2\beta_{j}\beta{k}}}{2^{\frac{\alpha_{j}\alpha_{k}}{2}}|t_{j}-t_{k}|^{\frac{\alpha_{j}\alpha_{k}}{2} + 2\beta_{j}\beta_{k}}} \Bigg) + \sum_{j=1}^{m} \log \frac{G(1+\frac{\alpha_{j}}{2}+\beta_{j})G(1+\frac{\alpha_{j}}{2}-\beta_{j})}{G(1+\alpha_{j})} \nonumber \\ & - \sum_{j=1}^{m} \Big( \frac{\alpha_{j}^{2}}{4}+\beta_{j}^{2} \Big) \log(\sqrt{1-t_{j}^{2}})   - \log (G(1+\alpha_{0})G(1+\alpha_{m+1})) - \sum_{j=1}^{m} 2 \beta_{j}^{2} \log 2 + \bigO \left( \frac{\log n}{n^{1-2\beta_{\max}}} \right),
\end{align}
where $\mathcal{A} = \sum_{j=0}^{m+1} \alpha_{j}$.
\end{theorem}
\begin{remark}
The asymptotics \eqref{Jacobi asymptotics of DIK} with $\vec{\beta} = \vec{0}$ and $\alpha_{1}=\ldots=\alpha_{m} = 0$ is consistent with \eqref{Jacobi starting point}.
\end{remark}

\begin{remark}
Our notation differs slightly from the one used in \cite{DIK}: $\alpha_{j}$ and $\beta_{j}$ in our paper corresponds to $2\alpha_{m+1-j}$ and $\beta_{m+1-j}$ in the paper \cite{DIK}.
\end{remark}
The goal of the next section is to obtain a similar formula as \eqref{Jacobi asymptotics of DIK} for $L_{n}(\vec{\alpha},\vec{\beta},2(x+1),0)$.

\section{Integration in $\vec{\alpha}$ and $\vec{\beta}$ for the Laguerre weight}\label{Section: FH integration}
In this section, we specialize to the classical Laguerre weight with FH singularities
\begin{equation}
w(x) = e^{-2n(x+1)}\omega(x),
\end{equation}
supported on $\mathcal{I} = [-1,+\infty)$. In this case, we recall that $\ell = 2+2\log 2$ and $\psi(x) = \frac{1}{\pi}$. We will find large $n$ asymptotics for the differential identity \eqref{Diff_identity_alphas_betas}, and then integrate in the parameters $\alpha_{0},\ldots,\alpha_{m},\beta_{1},\ldots,\beta_{m}$. We first focus on finding large $n$ asymptotics for $\widetilde{Y}(t_k)$, $k=0,\ldots,m$.
\begin{proposition}
For $k \in \{1,\ldots,m\}$, as $n \to + \infty$, we have
\begin{equation}\label{Ytilde-t_k}
\widetilde{Y}(t_k)=e^{-\frac{n\ell}{2}\sigma_3}(I + \mathcal{O}(n^{-1+2\be_{max}}))E_{t_k}(t_k) \begin{pmatrix}
\Phi_{k,11} &  \Phi_{k,12} \\
\Phi_{k,21} &  \Phi_{k,22} 
\end{pmatrix} e^{n(t_k+1)\sigma_3},
\end{equation}
where 
\begin{align}
& \Phi_{k,11}=\frac{\Ga(1+\frac{\al_k}{2}-\be_k)}{\Ga(1+\al_k)} \left(2n\frac{\sqrt{1-t_k}}{\sqrt{1+t_k}} \right)^{\frac{\al_k}{2}}\om^{-\frac{1}{2}}_k(t_k), \quad  \Phi_{k,12}= \frac{-\al_k\Ga(\al_k)}{\Ga(\frac{\al_k}{2}+\be_k)} \left(2n\frac{\sqrt{1-t_k}}{\sqrt{1+t_k}} \right)^{-\frac{\al_k}{2}}\om^{\frac{1}{2}}_k(t_k), \nonumber \\
& \Phi_{k,21}=\frac{\Ga(1+\frac{\al_k}{2}+\be_k)}{\Ga(1+\al_k)} \left(2n\frac{\sqrt{1-t_k}}{\sqrt{1+t_k}} \right)^{\frac{\al_k}{2}}\om^{-\frac{1}{2}}_k(t_k), \quad    \Phi_{k,22}= \frac{\al_k\Ga(\al_k)}{\Ga(\frac{\al_k}{2}-\be_k)} \left(2n\frac{\sqrt{1-t_k}}{\sqrt{1+t_k}} \right)^{-\frac{\al_k}{2}}\om^{\frac{1}{2}}_k(t_k). \label{Phik Gamma}
\end{align}
As $n \to + \infty$, we have
\begin{equation}\label{Ytilde_-1}
\widetilde{Y}(-1)=e^{-\frac{n\ell}{2}\sigma_3}\Big(I+\frac{R^{(1)}(-1)}{n}+\bigO(n^{-2+2\beta_{\max}})\Big)E_{-1}(-1) \begin{pmatrix}
\Phi_{0,11} &  \Phi_{0,12}  \\
\Phi_{0,21} &  \Phi_{0,22}
\end{pmatrix},
\end{equation}
where $R^{(1)}(-1)$ is given explicitly in \eqref{R^(1) Laguerre at -1} and
\begin{equation}\label{Phi0 Gamma}
\begin{split}
& \Phi_{0,11}=\frac{1}{\Ga(1+\al_0)} \left(\sqrt{2}n \right)^{\al_0}\om^{-\frac{1}{2}}_{-1}(-1), \ \ \ \ \  \Phi_{0,12}= -\frac{i\al_0\Ga(\al_0)}{2\pi} \left(\sqrt{2}n \right)^{-\al_0}\om^{\frac{1}{2}}_{-1}(-1), \\ & \Phi_{0,21}=-\frac{\pi i \al_0}{\Ga(1+\al_0)} \left(\sqrt{2}n \right)^{\al_0}\om^{-\frac{1}{2}}_{-1}(-1), \ \ \ \Phi_{0,22}= \frac{\al^2_0\Ga(\al_0)}{2} \left(\sqrt{2}n \right)^{-\al_0}\om^{\frac{1}{2}}_{-1}(-1).    
\end{split}    
\end{equation}
\end{proposition}
\begin{proof}
For fixed $1\leq k\leq m$, let $z \in \mathcal{D}_{t_k} \cap Q^R_{+,k}$ be outside the lenses. By inverting the RH transformations $Y \mapsto T \mapsto S \mapsto R$, we obtain 
\begin{equation}\label{Y-near-t_k}
Y(z)=e^{-\frac{n\ell}{2} \sigma_3}R(z)P^{(t_k)}(z)e^{ng(z)\sigma_3}e^{\frac{n\ell}{2}\sigma_3}
\end{equation}	
where $P^{(t_k)}(z)$ is given by (\ref{P^{(t_k)}II}). From \cite[Section 13.14(iii)]{NIST}, we have	
\begin{equation}
G(a,\al_k;z)=z^{\frac{\al_k}{2}}(1+\mathcal{O}(z)), \qquad z \to 0,
\end{equation}
and, if $\al_k \neq 0$, and $a-\frac{\al_k}{2} \pm \frac{\al_k}{2} \neq 0, -1, -2,\ldots$, as $z \to 0$ we have
\begin{equation}
H(a,\al_k;z) = \begin{cases}
\di \frac{\Ga(\al_k)}{\Ga(a)}z^{-\frac{\al_k}{2}}+\mathcal{O}(z^{1-\frac{\Re \al_k}{2}}) + \mathcal{O}(z^{\frac{\Re \al_k}{2}}) & \Re \al_k > 0, \\
\di   \frac{\Ga(-\al_k)}{\Ga(a-\al_k)}z^{\frac{\al_k}{2}}+ \frac{\Ga(\al_k)}{\Ga(a)}z^{-\frac{\al_k}{2}}+\mathcal{O}(z^{1+\frac{\Re \al_k}{2}})  & -1 < \Re \al_k \leq 0.
\end{cases}
\end{equation}
Conditions $a-\frac{\al_k}{2} \pm \frac{\al_k}{2} \neq 0, -1, -2,\ldots$ for $a = \frac{\alpha_{k}}{2}-\beta_{k}$ and $a = 1+\frac{\alpha_{k}}{2}-\beta_{k}$ reduce to $-\beta_{k} \pm \frac{\alpha_{k}}{2} \neq 0,-1,-2,\ldots$. Recalling that $V(x)=2(x+1)$ and $\psi(x)=\di \frac{1}{\pi}$, and using (\ref{f_t_kAsymptotics}), we find that the leading terms of $E^{-1}_{t_k}(z)P^{(t_k)}(z)e^{n\xi(z)\sigma_3}$ as $z\to t_k$ for $\alpha_{k} \neq 0$, $-\beta_{k} \pm \frac{\alpha_{k}}{2} \neq 0,-1,-2,\ldots$ are given by
\begin{equation}\label{scriptE-k}
\begin{pmatrix}
\Phi_{k,11} & \al^{-1}_k \left( \Phi_{k,12} +\Tilde{c_k}\Phi_{k,11}(z-t_k)^{\al_k} \right) \\
\Phi_{k,21} & \al^{-1}_k \left( \Phi_{k,22} +\Tilde{c_k}\Phi_{k,21}(z-t_k)^{\al_k} \right)
\end{pmatrix},
\end{equation}
where
\begin{equation}
\Tilde{c_k} = \al_k \frac{\Ga(-\al_k)\Ga(1+\al_k)e^{-\frac{\pi i \al_k}{2}}\om_k(t_k)}{\Ga(-\frac{\al_k}{2}-\be_k)\Ga(1+\frac{\al_k}{2}+\be_k)} = \frac{\pi \al_k}{\sin(\pi \al_k)} \frac{e^{i\pi \be_k}-e^{-i\pi \al_k}e^{-i\pi \be_k}}{2\pi i}\om_k(t_k) = e^{2n(t_k+1)}c_k
\end{equation}
and $c_k$ is given by (\ref{c_j}). The claim \eqref{Ytilde-t_k} for $\alpha_{k} \neq 0$, $-\beta_{k} \pm \frac{\alpha_{k}}{2} \neq 0,-1,-2,\ldots$ follows from (\ref{RHPsolution_versus_Regsolution_k}), (\ref{c_j}), (\ref{analytic-continuation-of-xi}), (\ref{R_large_n_asymp}), (\ref{Y-near-t_k}) and (\ref{scriptE-k}). We extend it for general parameters $\alpha_{k}$ and $\beta_{k}$ (still subject to the constraint $\Re \alpha_{k} >-1$ and $\Re \beta_{k} \in (-\frac{1}{2},\frac{1}{2})$) by continuity of $\widetilde{Y}(t_{k})$ in $\alpha_{k}$ and $\beta_{k}$ (this can be shown by a simple contour deformation, see e.g. \cite[eq (29) and below]{Krasovsky}). Now we turn to the proof of \eqref{Ytilde_-1}. For $z \in \mathcal{D}_{-1} \setminus \ovl{\left(\Om_+\cup\Om_-\right)}$, from Section \ref{Section: steepest descent}, we have 
\begin{equation}
Y(z)=    e^{-\frac{n\ell}{2} \sigma_3}R(z)P^{(-1)}(z)e^{ng(z)\sigma_3}e^{\frac{n\ell}{2}\sigma_3}.
\end{equation}
In this region, by (\ref{P^{(-1)}}) and (\ref{Phi explicit}), $P^{(-1)}(z)$ is given by 
\begin{equation*}
\begin{split}
P^{(-1)}(z) & = E_{-1}(z)\sigma_{3} \begin{pmatrix}
I_{\alpha_0}(2n(-f_{-1}(z))^{\frac{1}{2}}) & \frac{ i}{\pi} K_{\alpha_0}(2n(-f_{-1}(z))^{\frac{1}{2}}) \\
2\pi i n (-f_{-1}(z))^{\frac{1}{2}} I_{\alpha_0}^{\prime}(2n(-f_{-1}(z))^{\frac{1}{2}}) & -2n (-f_{-1}(z))^{\frac{1}{2}} K_{\alpha_0}^{\prime}(2n(-f_{-1}(z))^{\frac{1}{2}})
\end{pmatrix} \\ & \times \sigma_{3}\omega_{-1}(z)^{-\frac{\sigma_{3}}{2}}(-z-1)^{-\frac{\alpha_{0}}{2}\sigma_{3}}e^{-n\xi(z)\sigma_{3}}.
\end{split}    
\end{equation*}
From \cite[Section 10.30(i)]{NIST}, we have the following asymptotic behaviors as $z \to 0$ for the modified Bessel functions
\begin{align*}
& I_{\al_0}(z) = \frac{1}{\Ga(\al_0+1)}\Big(\frac{z}{2}\Big)^{\al_0}(1+\mathcal{O}(z^2)), \\
& K_{\al_0}(z) = \begin{cases}
\frac{\Ga(\al_0)}{2}(\frac{z}{2})^{-\al_0}+\mathcal{O}(z^{1-\Re\al_0})+\mathcal{O}(z^{\Re\al_0}), & \mbox{if } \Re \al_0 \geq 0, \alpha_{0} \neq 0, \\
\frac{\Ga(-\al_0)}{2}(\frac{z}{2})^{\al_0}+\frac{\Ga(\al_0)}{2}(\frac{z}{2})^{-\al_0}+\mathcal{O}(z^{2+\Re\al_0}), & \mbox{if } -1<\Re \al_0 < 0.
\end{cases}
\end{align*}
Using (\ref{f_{-1}-asymp}), for $\alpha_{k} \neq 0$, we find that the leading terms of $E^{-1}_{-1}(z)P^{(-1)}(z)e^{n\xi(z)\sigma_3}$ as $z\to -1$  are given by
\begin{equation}\label{scriptE,0}
\begin{pmatrix}
\Phi_{0,11} & \al^{-1}_0 \left( \Phi_{0,12} +\Tilde{c_0}\Phi_{0,11}(-z-1)^{\al_0} \right) \\
\Phi_{0,21} & \al^{-1}_0 \left( \Phi_{0,22} +\Tilde{c_0}\Phi_{0,21}(-z-1)^{\al_0} \right)
\end{pmatrix},
\end{equation}
where
\begin{equation}
\Tilde{c_0} = \frac{i\al_0}{2\sin(\pi \al_0)}\om_{-1}(-1)=e^{\pi i \al_0}c_0
\end{equation}
and where we recall that $c_{0}$ is defined in \eqref{c_0}. This proves \eqref{Ytilde_-1} for $\alpha_{0} \neq 0$. The case $\alpha_{0} = 0$ follows by continuity of $\widetilde{Y}(-1)$.
\end{proof}

\subsection{Asymptotics for $\partial_{\nu} \log L_n(\vec{\al},\vec{\be}, 2(x+1),0)$, $\nu \in \{\alpha_{0},\ldots,\alpha_{m},\beta_{1},\ldots,\beta_{m}\}$}
From (\ref{OPsolution}) and \eqref{Orthogonal_Polynomials}, we have 
\begin{equation}\label{polynomial coefficients and Y}
\ka^2_{n-1}=\lim_{z \to \infty} \frac{iY_{21}(z)}{ 2\pi z^{n-1}}, \ \ \ \ \ \ka^{-2}_n= -2\pi i \lim_{z\to \infty} z^{n+1} Y_{12}(z) , \ \ \ \ \ \eta_n=\lim_{z\to \infty} \frac{Y_{11}(z)-z^n}{z^{n-1}}.    
\end{equation}
Inverting the transformations $Y \mapsto T \mapsto S \mapsto R$ for $z \in \C \setminus \overline{\left( \Om_+ \cup \Om_- \cup(\mathcal{I}\setminus \mathcal{S}) \cup^{m+1}_{j=0} \mathcal{D}_j  \right)}$ (i.e. outside the lenses and outside the disks) gives
\begin{equation}\label{traceback}
Y(z)=e^{-\frac{n \ell}{2}\sigma_3}R(z)P^{(\infty)}(z)e^{ng(z)\sigma_3}e^{\frac{n \ell}{2}\sigma_3}.
\end{equation}
From (\ref{e^ng_asym_Laguerre}), (\ref{P_1^inf}), (\ref{R_large_n_asymp}), (\ref{R^{(1)}}) and \eqref{polynomial coefficients and Y}, we find large $n$ asymptotic for $\ka^2_{n-1}, \ka^2_{n}$ and $\eta_n$. As $n \to + \infty$, we have
\begin{equation}\label{ka n-1 asymp}
\ka^2_{n-1} = e^{2n}2^{2(n-1)+\mathcal{A}} \pi^{-1} \exp \bigg(\hspace{-0.15cm}-2i\sum^{m}_{j=1}\be_j \arcsin{t_j}\bigg)  \Bigg(1+ \frac{R_{1,21}^{(1)}}{nP_{1,21}^{(\infty)}} + \mathcal{O}(n^{-2+2\be_{max}}) \Bigg),
\end{equation}
where $\mathcal{A} = \alpha_{0}+\alpha_{1}+\ldots+\alpha_{m}$ and
\begin{equation}
\frac{R_{1,21}^{(1)}}{P_{1,21}^{(\infty)}} = \sum^m_{j=1} \frac{v_j(1-\widetilde{\La}_{R,1,j})}{1-t_j}  + \frac{1-4\al^2_0}{16} -\frac{1}{4}\bigg( (\mathcal{A}-\widetilde{\mathcal{B}_{1}})^{2} - 2 (\mathcal{A}-\widetilde{\mathcal{B}_{1}}) + \frac{11}{12} \bigg). 
\end{equation}
Similarly, for $\ka^2_{n}$ we find
\begin{equation}\label{ka n asymp}
\ka^2_{n} = e^{2n}2^{2n + \mathcal{A}} \pi^{-1} \exp{\bigg(\hspace{-0.15cm}-2i\sum^{m}_{j=1}\be_j \arcsin{t_j}\bigg)} \Bigg(1- \frac{R_{1,12}^{(1)}}{nP_{1,12}^{(\infty)}}  + \mathcal{O}(n^{-2+2\be_{max}}) \Bigg),
\end{equation}
as $n \to + \infty$, where by \eqref{R^{(1)}} we have
\begin{equation}
- \frac{R_{1,12}^{(1)}}{P_{1,12}^{(\infty)}} = \sum^m_{j=1} \frac{v_j(1+\widetilde{\La}_{R,2,j})}{1-t_j}  + \frac{1-4\al^2_0}{16} -\frac{1}{4}\bigg( (\mathcal{A}-\widetilde{\mathcal{B}_{1}})^{2} + 2 (\mathcal{A}-\widetilde{\mathcal{B}_{1}}) + \frac{11}{12} \bigg).
\end{equation}
Finally, for $\eta_{n}$ we obtain
\begin{equation}\label{eta n asymp}
\eta_n = \frac{n}{2}+ P^{(\infty)}_{1,11}+\frac{R_{1,11}^{(1)}}{n} + \mathcal{O}(n^{-2+2\be_{max}}), \qquad \mbox{ as } n \to + \infty,
\end{equation}
where $P_{1,11}^{(\infty)}$ is given by \eqref{P_1^inf} and $R^{(1)}_{1,11}$ can be computed from \eqref{R^{(1)}} and is given by
\begin{align}
& R_{1,11}^{(1)} = \sum^{m}_{j=1} \frac{v_j(t_j+\widetilde{\Lambda}_{I,j})}{2(1-t_j)} - \frac{1-4\al^2_{0}}{32} + \frac{1-4(\mathcal{A}-\widetilde{\mathcal{B}_{1}})^{2}}{32}.
\end{align}
Let $\nu \in \{\alpha_{0},\alpha_{1},\ldots,\alpha_{m},\beta_{1},\ldots,\beta_{m}\}$. Then, from \eqref{R^{(1)} derivative order}, \eqref{ka n-1 asymp}, \eqref{ka n asymp} and \eqref{eta n asymp}, we find that the large $n$ asymptotics of the first part of the differential identity \eqref{Diff_identity_alphas_betas} are given by
\begin{multline}\label{first part FH}
 -(n + \mathcal{A})\partial_{\nu}\log(\ka_n\ka_{n-1})+2n \partial_{\nu} \eta_n =  \partial_{\nu} \Big( 2 \log D_{\infty} - \alpha_{0} + \sum_{j=1}^{m}t_{j}\alpha_{j} + 2i \sum_{j=1}^{m} \sqrt{1-t_{j}^{2}}\beta_{j} \Big) n + \\  2\mathcal{A} \partial_{\nu} \log D_{\infty} +  \partial_{\nu} \Big( \frac{\alpha_{0}^{2}}{2} \Big) + \partial_{\nu} \sum_{j=1}^{m} v_{j}(\widetilde{\Lambda}_{I,j}-1) + \bigO\left( \frac{\log n}{n^{1-4\beta_{\max}}} \right).
\end{multline}
Now we compute the second part of the differential identity \eqref{Diff_identity_alphas_betas}. First, we compute the contributions from $t_j$, $j=1,\ldots,m$ using \eqref{E_{t_k}(t_k)}, \eqref{R^{(1)} derivative order}, \eqref{Ytilde-t_k} and \eqref{Phik Gamma}. We obtain
\begin{align}\label{2nd 1->m}
& \sum^{m}_{j=1} \Big( \widetilde{Y}_{22}(t_j)\partial_{\nu}Y_{11}(t_j) - \widetilde{Y}_{12}(t_j)\partial_{\nu}Y_{21}(t_j) + Y_{11}(t_j)\widetilde{Y}_{22}(t_j)\partial_{\nu}\log(\ka_n\ka_{n-1}) \Big) = \nonumber \\ &
-(\mathcal{A}-\alpha_{0})\partial_{\nu} \log D_{\infty} + \sum_{j=1}^{m} \Big( \Phi_{j,22}\partial_{\nu}\Phi_{j,11}-\Phi_{j,12}\partial_{\nu} \Phi_{j,21}-2\beta_{j} \partial_{\nu} \log \Lambda_{j} \Big) + \bigO\Big( \frac{\log n}{n^{1-4\beta_{\max}}} \Big).
\end{align}
Note that $E_{-1}(-1) = \bigO(n^{\frac{\sigma_{3}}{2}})$ as $n \to + \infty$, while $E_{t_k}(t_k) = \bigO(n^{\be_k \sigma_{3}})$, $k=1,\ldots,m$. This makes the computations for the contribution from $-1$ more involved. From \eqref{E_{-1}(-1)}, \eqref{R^{(1)} derivative order} \eqref{Ytilde_-1} and \eqref{Phi0 Gamma}, we obtain
\begin{multline}\label{2nd 0}
\widetilde{Y}_{22}(-1)\partial_{\nu}Y_{11}(-1) - \widetilde{Y}_{12}(-1)\partial_{\nu}Y_{21}(-1) + Y_{11}(-1)\widetilde{Y}_{22}(-1)\partial_{\nu}\log(\ka_n\ka_{n-1}) =  \\ + \partial_{\nu}\Big( R_{11}^{(1)}(-1)-R_{22}^{(1)}(-1)+iD_{\infty}^{-2}R_{12}^{(1)}(-1)+iD_{\infty}^{2}R_{21}^{(1)}(-1) + iD_{\infty}^{-2}R_{1,12}^{(1)}+iD_{\infty}^{2}R_{1,21}^{(1)} \Big)    \\  - \alpha_{0} \partial_{\nu} \log D_{\infty}+\Phi_{0,22} \partial_{\nu}\Phi_{0,11}-\Phi_{0,12}\partial_{\nu} \Phi_{0,21} + \bigO\Big( \frac{\log n}{n^{1-4\beta_{\max}}} \Big).
\end{multline}
We observe significant simplifications using \eqref{res t_{k}}, \eqref{connection Lambdas}, \eqref{res at -1L}, \eqref{res at 1L}, \eqref{R^(1) Laguerre at -1}, and  \eqref{new res -1 L}:
\begin{multline}\label{huge cancellations}
R_{11}^{(1)}(-1)-R_{22}^{(1)}(-1)+iD_{\infty}^{-2}R_{12}^{(1)}(-1)+iD_{\infty}^{2}R_{21}^{(1)}(-1) +iD_{\infty}^{-2}R_{1,12}^{(1)} + iD_{\infty}^{2}R_{1,21}^{(1)} = -\sum_{j=1}^{m}v_{j}\widetilde{\Lambda}_{I,j}.
\end{multline}
Adding \eqref{first part FH}, \eqref{2nd 1->m} and \eqref{2nd 0} yields
\begin{multline}\label{explicit diff identity FH 1}
\partial_{\nu} \log L_{n}(\vec{\alpha},\vec{\beta},2(x+1),0) = \partial_{\nu} \Big( 2 \log D_{\infty} - \alpha_{0} + \sum_{j=1}^{m}t_{j}\alpha_{j} + 2i \sum_{j=1}^{m} \sqrt{1-t_{j}^{2}}\beta_{j} \Big) n + \mathcal{A} \partial_{\nu} \log D_{\infty}  \\  + \partial_{\nu} \Big( \frac{\alpha_{0}^{2}}{2} \Big) + \sum_{j=0}^{m} \Big( \Phi_{j,22}\partial_{\nu}\Phi_{j,11}-\Phi_{j,12}\partial_{\nu} \Phi_{j,21}\Big) - \sum_{j=1}^{m} (\partial_{\nu}v_{j}+2\beta_{j} \partial_{\nu} \log \Lambda_{j})  + \bigO\Big( \frac{\log n}{n^{1-4\beta_{\max}}} \Big),
\end{multline}
as $n \to + \infty$. Now, we perform some computations to make the above asymptotic formula more explicit. From \eqref{Phi0 Gamma} and using the identity $z\Ga(z)=\Ga(z+1)$ we have
\begin{align}\label{explicit Phi0 in identity}
& \Phi_{0,22}\partial_{\nu}\Phi_{0,11}-\Phi_{0,12}\partial_{\nu} \Phi_{0,21} =  \nonumber \\ & \frac{\alpha_{0}}{2} \partial_{\nu} \log \Big( \frac{\alpha_{0}}{\Gamma(1+\alpha_{0})^{2}}\Big) + \alpha_{0} \log(\sqrt{2}n) \partial_{\nu} \alpha_{0}
-\frac{\alpha_{0}}{2}\partial_{\nu}\bigg(\sum_{\ell=1}^{m} \alpha_{\ell}\log(1+t_{\ell})+ i\pi\sum_{\ell=1}^{m} \beta_{\ell} \bigg).
\end{align}
And from \eqref{Phik Gamma}, after a long computation, for $1 \leq j \leq m$ we obtain
\begin{align}\label{explicit Phij in identity}
	& \Phi_{j,22}\partial_{\nu}\Phi_{j,11}-\Phi_{j,12}\partial_{\nu} \Phi_{j,21} = \frac{\alpha_{j}}{2} \partial_{\nu} \log \frac{\Gamma(1+\frac{\alpha_{j}}{2}-\beta_{j})\Gamma(1+\frac{\alpha_{j}}{2}+\beta_{j})}{\Gamma(1+\alpha_{j})^{2}} + \frac{\alpha_{j}}{2}\log \Big( 2n \frac{\sqrt{1-t_{j}}}{\sqrt{1+t_{j}}} \Big) \partial_{\nu} \alpha_{j} \nonumber \\ & + \beta_{j} \partial_{\nu} \log \frac{\Gamma(1+\frac{\alpha_{j}}{2}+\beta_{j})}{\Gamma(1+\frac{\alpha_{j}}{2}-\beta_{j})} - \frac{\alpha_{j}}{2}\partial_{\nu} \bigg( \sum_{\substack{\ell=0\\ \ell \neq j}}^{m} \alpha_{\ell} \log |t_{\ell}-t_{j}| - i \pi \sum_{\ell=1}^{j-1} \beta_{\ell} + i \pi \sum_{\ell = j+1}^{m} \beta_{\ell} \bigg).
\end{align}
Also, from \eqref{def of Lambdak} and \eqref{def of lambdak}, we have
\begin{multline}\label{explicit Lambdaj in identity}
\partial_{\nu} \log \Lambda_{j} =  \partial_{\nu} \bigg( \frac{i \mathcal{A}}{2} \arccos t_{j} - \frac{\pi i}{4}\alpha_{j} - \frac{\pi i}{2} \sum_{\ell = j+1}^{m} \alpha_{\ell} + \beta_{j} \log (4\pi \rho(t_{j})n(1-t_{j}^{2})) - \sum_{\substack{\ell = 1 \\ \ell \neq j}}^{m} \beta_{\ell} \log T_{j\ell} \bigg).
\end{multline}
Substituting \eqref{explicit Phi0 in identity}--\eqref{explicit Lambdaj in identity} into \eqref{explicit diff identity FH 1}, and using the expression for $D_{\infty}$ and $v_{j}$ given by \eqref{D_infty} and below \eqref{asymptotics on the disk D_t}, we obtain
\begin{align}\label{final expression diff identity FH}
& \partial_{\nu} \log L_{n}(\vec{\alpha},\vec{\beta},2(x+1),0) = \partial_{\nu} \bigg( \sum_{j=0}^{m}(t_{j}-\log 2)\alpha_{j} + 2i \sum_{j=1}^{m}\beta_{j}\big(\arcsin t_{j} + \sqrt{1-t_{\smash{j}}^{2}}\, \big) \bigg)n \nonumber \\
& +\mathcal{A} \partial_{\nu} \Big( i \sum_{j=1}^{m} \beta_{j} \arcsin t_{j} - \frac{\mathcal{A}}{2} \log 2 \Big) + \partial_{\nu} \Big( \frac{\alpha_{0}^{2}}{2} \Big) + \frac{\alpha_{0}}{2}\partial_{\nu} \log \frac{\alpha_{0}}{\Gamma(1+\alpha_{0})^{2}} +\alpha_{0} \log(\sqrt{2}n)\partial_{\nu} \alpha_{0} \nonumber \\ & - \sum_{j=0}^{m} \frac{\alpha_{j}}{2} \partial_{\nu} \bigg( \sum_{\substack{\ell = 0 \\\ell \neq j}}^{m}  \alpha_{\ell} \log|t_{\ell}-t_{j}| - i \pi \sum_{\ell =1}^{j-1}\beta_{\ell} + i \pi \sum_{\ell = j+1}^{m}\beta_{\ell} \bigg) + \sum_{j=1}^{m} \frac{\alpha_{j}}{2} \log \Big( 2n \frac{\sqrt{1-t_{j}}}{\sqrt{1+t_{j}}} \Big) \partial_{\nu}  \alpha_{j} \nonumber \\ &
+ \sum_{j=1}^{m} \frac{\alpha_{j}}{2} \partial_{\nu} \log \frac{\Gamma(1+\frac{\alpha_{j}}{2}-\beta_{j})\Gamma(1+\frac{\alpha_{j}}{2}+\beta_{j})}{\Gamma(1+\alpha_{j})^{2}} + \sum_{j=1}^{m} \beta_{j} \partial_{\nu} \log \frac{\Gamma(1+\frac{\alpha_{j}}{2}+\beta_{j})}{\Gamma(1+\frac{\alpha_{j}}{2}-\beta_{j})} + \sum_{j=1}^{m} \partial_{\nu} \Big( \frac{\alpha_{j}^{2}}{4}-\beta_{j}^{2} \Big) \nonumber \\ &
- \sum_{j=1}^{m} 2\beta_{j} \partial_{\nu} \bigg( \frac{i\mathcal{A}}{2}\arccos t_{j} - \frac{\pi i}{4}\alpha_{j} - \frac{\pi i}{2} \sum_{\ell = j+1}^{m} \alpha_{\ell} + \beta_{j} \log (4\pi \rho(t_{j})n(1-t_{j}^{2})) - \sum_{\substack{\ell = 1 \\ \ell \neq j}}^{m} \beta_{\ell}\log T_{j\ell} \bigg) \nonumber \\
&  + \bigO\Big( \frac{\log n}{n^{1-4\beta_{\max}}} \Big), \qquad \mbox{ as } n \to + \infty,
\end{align}
where we recall that $t_{0} = -1$. From the discussion in Subsection \ref{subsection: small norm}, the above error term is uniform \textit{for all} $(\vec{\alpha},\vec{\beta})$ in a given compact set $\Omega$, and uniform in $\vec{t}$ such that \eqref{delta def in small norm section} holds. However, as stated in Proposition \ref{prop: diff identity Laguerre}, the identity \eqref{final expression diff identity FH} itself is valid on the subset $\Omega \setminus \widetilde{\Omega}$ for which $p_{0},\ldots,p_{n}$ exist. From the determinantal representation of orthogonal polynomials, $\widetilde{\Omega}$ is locally finite and we can extend \eqref{final expression diff identity FH} \textit{for all} $(\vec{\alpha},\vec{\beta}) \in \Omega$ by continuity (for $n$ large enough such that the r.h.s. exists). We refer to \cite{Krasovsky, ItsKrasovsky, DIK, Charlier} for very similar situations, with more details provided. Our goal for the rest of this section is to prove Proposition \ref{61} below.
\begin{proposition}\label{61}
As $n \to \infty$, we have 
\begin{multline}\label{61r}
\hspace{-0.3cm}\log \frac{L_{n}(\vec{\alpha},\vec{\beta},2(x+1),0)}{L_{n}(\vec{\alpha_0},\vec{0},2(x+1),0)} = 2i n \sum_{j=1}^{m} \beta_{j} \Big(\arcsin t_{j} + \sqrt{1-t_{\smash{j}}^{2}}\Big) + \sum_{j=1}^{m} (t_{j}-\log 2)\alpha_{j} n + \sum_{j=1}^{m} \frac{\alpha_{j}^{2}}{4}\log \Big( n \frac{\sqrt{1-t_{j}}}{\sqrt{1+t_{j}}} \Big) \\
- \sum_{j=1}^{m} \beta_{j}^{2} \log(4\pi \rho(t_{j})n(1-t_{j}^{2})) + \sum_{j=1}^{m} \log \frac{G(1+\frac{\alpha_{j}}{2}+\beta_{j})G(1+\frac{\alpha_{j}}{2}-\beta_{j})}{G(1+\alpha_{j})} + \frac{i\pi}{2} \sum_{0 \leq j < k \leq m}(\alpha_{k}\beta_{j}-\alpha_{k}\beta_{j}) \\
 + i \mathcal{A} \sum_{j=1}^{m}\beta_{j}\arcsin t_{j} + 2 \hspace{-0.2cm} \sum_{1 \leq j < k \leq m} \hspace{-0.3cm} \beta_{j}\beta_{k} \log T_{jk}  - \frac{\log 2}{2} \hspace{-0.3cm} \sum_{0 \leq j < k \leq m} \alpha_{j} \alpha_{k} 
- \hspace{-0.2cm} \sum_{0 \leq j < k \leq m} \hspace{-0.3cm} \frac{\alpha_{j}\alpha_{k}}{2} \log |t_{k}-t_{j}|  + \bigO \Big( \frac{\log n}{n^{1-4\beta_{\max}}} \Big),
\end{multline}
where $\log L_{n}(\vec{\alpha_0},\vec{0},2(x+1),0)$ is given by \eqref{Laguerre starting point}.
\end{proposition}

\subsection{Integration in $\alpha_{0}$}\label{Int al_0}
In this short subsection, we make a consistency check with \eqref{Laguerre starting point}. Let us set $\alpha_{1}=\ldots=\alpha_{m} = 0 = \beta_{1}=\ldots=\beta_{m}$ and $\nu = \alpha_{0}$ in \eqref{final expression diff identity FH}. With the notations $\vec{\alpha}_{0} = (\alpha_{0},0,\ldots,0) \in \C^{m+1}$ and $\vec{0} = (0,\ldots,0) \in \C^m$, this gives
\begin{multline}\label{diff alpha0}
\partial_{\alpha_{0}} \log L_{n}(\vec{\alpha}_{0},\vec{0},2(x+1),0) = -(1+\log 2)n - \frac{\log 2}{2}\alpha_{0} + \alpha_{0} + \frac{\alpha_{0}}{2}\partial_{\alpha_{0}} \log \frac{\alpha_{0}}{\Gamma(1+\alpha_{0})^{2}} \\ + \alpha_{0} \log(\sqrt{2}n) + \bigO\Big( \frac{\log n}{n} \Big)
\end{multline}
as $n \to + \infty$. Integrating \eqref{diff alpha0} from $\alpha_{0} = 0$ to an arbitrary $\alpha_{0}$, we obtain
\begin{multline}\label{lol3}
\log \frac{L_{n}(\vec{\alpha}_{0},\vec{0},2(x+1),0)}{L_{n}(\vec{0},\vec{0},2(x+1),0)} = -(1+\log 2)\alpha_{0}n + \frac{\alpha_{0}^{2}}{2}\Big( 1-\frac{\log 2}{2}\Big) + \int_{0}^{\alpha_{0}} \frac{x}{2}\partial_{x} \log \frac{x}{\Gamma(1+x)^{2}}dx  \\ + \frac{\alpha_{0}^{2}}{2}\log(\sqrt{2}n) + \bigO \Big( \frac{\log n}{n} \Big).
\end{multline}
From \cite[formula 5.17.4]{NIST}, we have
\begin{equation}\label{integral of Gamma function}
\int_{0}^{z}\log \Gamma (1+x) dx = \frac{z}{2} \log 2\pi - \frac{z(z+1)}{2} + z \log \Gamma(z+1) - \log G(z+1),
\end{equation}
where $G$ is Barnes' $G$-function. Therefore, after an integration by parts, \eqref{lol3} can be rewritten as
\begin{multline*}
\log \frac{L_{n}(\vec{\alpha}_{0},\vec{0},2(x+1),0)}{L_{n}(\vec{0},\vec{0},2(x+1),0)} = -(1+\log 2)\alpha_{0}n + \frac{\alpha_{0}^{2}}{2} \log n + \frac{\alpha_{0}}{2}\log 2\pi - \log G(1+\alpha_{0}) + \bigO \Big( \frac{\log n}{n} \Big),
\end{multline*}
which is consistent with \eqref{Laguerre starting point}.
\subsection{Integration in $\alpha_{1},\ldots,\alpha_{m}$}\label{Int al_k}
We set $\alpha_{2} = \ldots = \alpha_{m} = 0 = \beta_{1} = \ldots = \beta_{m}$ and $\nu = \alpha_{1}$ in \eqref{final expression diff identity FH}. With the notation $\vec{\alpha}_{1} = (\alpha_{0},\alpha_{1},0,\ldots,0) \in \C^{m+1}$, we obtain
\begin{multline}\label{diff alpha1}
\partial_{\alpha_{1}} \log L_{n}(\vec{\alpha}_{1},\vec{0},2(x+1),0) = (t_{1}-\log 2)n - \frac{\log 2}{2} \alpha_{0}  - \frac{\alpha_{0}}{2} \log |t_{1}-t_{0}| \\ + \frac{\alpha_{1}}{2} \log \Big( n \frac{\sqrt{1-t_{1}}}{\sqrt{1+t_{1}}} \Big) + \alpha_{1} \partial_{\alpha_{1}} \log \frac{\Gamma(1+\frac{\alpha_{1}}{2})}{\Gamma(1+\alpha_{1})} + \frac{\alpha_{1}}{2} + \bigO \Big( \frac{\log n}{n} \Big),
\end{multline}
as $n \to + \infty$. Using integration by parts and \eqref{integral of Gamma function} we obtain, we obtain the following relation
\begin{equation}\label{integral Gamma for alpha integration}
\int_{0}^{z} x \partial_{x} \log \frac{\Gamma(1+\frac{x}{2})}{\Gamma(1+x)}dx = - \frac{z^{2}}{4} + \log \frac{G(1+\tfrac{z}{2})^{2}}{G(1+z)}.
\end{equation}
Using \eqref{integral Gamma for alpha integration}, we integrate \eqref{diff alpha1} from $\al_1=0$ to an arbitrary $\al_1$. We get
\begin{multline}
\log \frac{L_{n}(\vec{\alpha}_{1},\vec{0},2(x+1),0)}{L_{n}(\vec{\alpha}_{0},\vec{0},2(x+1),0)} = (t_{1}-\log 2)\alpha_{1}n - \frac{\log 2}{2} \alpha_{0} \alpha_{1} - \frac{\alpha_{0}\alpha_{1}}{2}\log|t_{1}-t_{0}| \\ + \frac{\alpha_{1}^{2}}{4} \log \Big( n \frac{\sqrt{1-t_{1}}}{\sqrt{1+t_{1}}} \Big) + \log \frac{G(1+\frac{\alpha_{1}}{2})^{2}}{G(1+\alpha_{1})}+ \bigO \Big( \frac{\log n}{n} \Big), \qquad \mbox{as } n \to + \infty.
\end{multline}
We proceed in a similar way for the other variables, by integrating successively in $\al_2, \alpha_{3},\ldots,\alpha_{m}$. At the last step, setting $\beta_{1} = \ldots = \beta_{m} = 0$ and $\nu = \alpha_{m}$ in \eqref{final expression diff identity FH}, we obtain
\begin{multline}\label{diff alpham}
\partial_{\alpha_{m}}  \log L_{n}(\vec{\alpha},\vec{0},2(x+1),0) = (t_{m}-\log 2)n - \frac{\log 2}{2}(\mathcal{A}-\alpha_{m}) - \sum_{j=0}^{m-1} \frac{\alpha_{j}}{2} \log |t_{m}-t_{j}| \\
+ \frac{\alpha_{m}}{2} \log \Big( n \frac{\sqrt{1-t_{m}}}{\sqrt{1+t_{m}}} \Big) + \alpha_{m} \partial_{\alpha_{m}} \log \frac{\Gamma(1+\frac{\alpha_{m}}{2})}{\Gamma(1+\alpha_{m})} + \frac{\alpha_{m}}{2} + \bigO \Big( \frac{\log n}{n} \Big).
\end{multline} Integrating \eqref{diff alpham} from $\alpha_{m} = 0$ to an arbitrary $\alpha_{m}$ using again \eqref{integral Gamma for alpha integration}, and with the notation $\vec{\alpha}_{m-1} = (\alpha_{0},\ldots,\alpha_{m-1},0)$, we obtain
\begin{multline}
\log \frac{L_{n}(\vec{\alpha},\vec{0},2(x+1),0)}{L_{n}(\vec{\alpha}_{m-1},\vec{0},2(x+1),0)} = (t_{m}-\log 2)\alpha_{m}n - \frac{\log 2}{2} \sum_{j=0}^{m-1} \alpha_{j}\alpha_{m} - \sum_{j=0}^{m-1} \frac{\alpha_{j}\alpha_{m}}{2}\log|t_{m}-t_{j}|  \\    + \frac{\alpha_{m}^{2}}{4} \log \Big( n \frac{\sqrt{1-t_{m}}}{\sqrt{1+t_{m}}} \Big) + \log \frac{G(1+\frac{\alpha_{m}}{2})^{2}}{G(1+\alpha_{m})}+ \bigO \Big( \frac{\log n}{n} \Big),
\end{multline}
as $n \to + \infty$. Summing the contributions of each step, we arrive at
\begin{multline}\label{al0 to al integration}
\log \frac{L_{n}(\vec{\alpha},\vec{0},2(x+1),0)}{L_{n}(\vec{\alpha}_{0},\vec{0},2(x+1),0)} = \sum_{j=1}^{m} (t_{j}-\log 2)\alpha_{j} n - \frac{\log 2}{2} \sum_{0 \leq j < k \leq m} \alpha_{j} \alpha_{k} \\
- \sum_{0 \leq j < k \leq m} \frac{\alpha_{j}\alpha_{k}}{2} \log |t_{k}-t_{j}| + \sum_{j=1}^{m} \frac{\alpha_{j}^{2}}{4}\log \Big( n \frac{\sqrt{1-t_{j}}}{\sqrt{1+t_{j}}} \Big) + \sum_{j=1}^{m} \log \frac{G(1+\frac{\alpha_{j}}{2})^{2}}{G(1+\alpha_{j})} + \bigO\Big( \frac{\log n}{n} \Big),
\end{multline}
as $n \to + \infty$.

\subsection{Integration in $\beta_{1},\ldots,\beta_{m}$}\label{Int be_k}
For convenience, we introduce the notation
\begin{equation}\label{def of Aj}
\mathcal{A}_{k} = \sum_{j=0}^{k-1} \alpha_{j} - \sum_{j=k+1}^{m} \alpha_{j}, \qquad k = 0,1,\ldots,m.
\end{equation}
We set $\beta_{2} = \ldots = \beta_{m} = 0$ and $\nu = \beta_{1}$ in \eqref{final expression diff identity FH}. With the notation $\vec{\beta}_{1} = (\beta_{1},0,\ldots,0)$, we have
\begin{align}\label{diff beta1}
& \partial_{\beta_{1}} \log L_{n}(\vec{\alpha},\vec{\beta}_{1},2(x+1),0) = 2i \big( \arcsin t_{1} + \sqrt{1-t_{\smash{1}}^{2}} \, \big) n + i \mathcal{A} \arcsin t_{1} - \frac{i \pi}{2} \mathcal{A}_{1} \nonumber \\ 
& \hspace{1cm} + \frac{\alpha_{1}}{2} \partial_{\beta_{1}} \log \Gamma(1+\tfrac{\alpha_{1}}{2}-\beta_{1})\Gamma(1+\tfrac{\alpha_{1}}{2}+\beta_{1}) + \beta_{1} \partial_{\beta_{1}} \log \frac{\Gamma(1+\frac{\alpha_{1}}{2}+\beta_{1})}{\Gamma(1+\frac{\alpha_{1}}{2}-\beta_{1})} - 2\beta_{1} \nonumber \\ 
& \hspace{1cm} -2\beta_{1} \log \big( 4\pi \rho(t_{1})n(1-t_{1}^{2}) \big) + \bigO \Big( \frac{\log n}{n^{1-4\beta_{\max}}} \Big).
\end{align}
After some computations using \eqref{integral of Gamma function}, we obtain
\begin{multline}\label{integral Gamma for beta integration}
\int_{0}^{\beta_{1}} \bigg( \frac{\alpha_{1}}{2} \partial_{x} \log \Gamma(1+\tfrac{\alpha_{1}}{2}-x)\Gamma(1+\tfrac{\alpha_{1}}{2}+x) + x \partial_{x} \log \frac{\Gamma(1+\frac{\alpha_{1}}{2}+x)}{\Gamma(1+\frac{\alpha_{1}}{2}-x)} - 2x \bigg) dx \\ = \log \frac{G(1+\frac{\alpha_{1}}{2}+\beta_{1})G(1+\frac{\alpha_{1}}{2}-\beta_{1})}{G(1+\frac{\alpha_{1}}{2})^{2}}.
\end{multline}
Integrating \eqref{diff beta1} from $\beta_{1} = 0$ to an arbitrary $\beta_{1}$ and using \eqref{integral Gamma for beta integration}, we obtain
\begin{multline}
\log \frac{L_{n}(\vec{\alpha},\vec{\beta}_{1},2(x+1),0)}{L_{n}(\vec{\alpha},\vec{0},2(x+1),0)} = 2i \beta_{1} \big( \arcsin t_{1} + \sqrt{1-t_{\smash{1}}^{2}} \, \big) n + i \mathcal{A} \beta_{1} \arcsin t_{1} - \frac{i\pi}{2}\mathcal{A}_{1}\beta_{1} \\
 + \log \frac{G(1+\frac{\alpha_{1}}{2}+\beta_{1})G(1+\frac{\alpha_{1}}{2}-\beta_{1})}{G(1+\frac{\alpha_{1}}{2})^{2}} - \beta_{1}^{2} \log(4\pi \rho(t_{1})n(1-t_{1}^{2})) + \bigO \Big( \frac{\log n}{n^{1-4\beta_{\max}}} \Big).
\end{multline}
We integrate successively in $\beta_{2},\ldots,\beta_{m}$. At the last step, we set $\nu = \beta_{m}$ in \eqref{final expression diff identity FH}, which gives
\begin{align}\label{diff betam}
& \partial_{\beta_{m}} \log L_{n}(\vec{\alpha},\vec{\beta},2(x+1),0) = 2i \big( \arcsin t_{m} + \sqrt{1-t_{m}^{2}} \, \big) n + i \mathcal{A} \arcsin t_{m} - \frac{i \pi}{2} \mathcal{A}_{m} \nonumber \\ & + \frac{\alpha_{m}}{2} \partial_{\beta_{m}} \log \Gamma(1+\tfrac{\alpha_{m}}{2}-\beta_{m})\Gamma(1+\tfrac{\alpha_{m}}{2}+\beta_{m}) + \beta_{m} \partial_{\beta_{m}} \log \frac{\Gamma(1+\frac{\alpha_{m}}{2}+\beta_{m})}{\Gamma(1+\frac{\alpha_{m}}{2}-\beta_{m})} - 2\beta_{m} \nonumber \\ &
+\sum_{j=1}^{m-1}2\beta_{j} \log T_{jm}-2\beta_{m} \log \big( 4\pi \rho(t_{m})n(1-t_{m}^{2}) \big) + \bigO \Big( \frac{\log n}{n^{1-4\beta_{\max}}} \Big),
\end{align}
as $n \to + \infty$. Integrating \eqref{diff betam} from $\beta_{m} = 0$ to an arbitrary $\beta_{m}$, using the notation $\vec{\beta}_{m-1} = (\beta_{1},\ldots,\beta_{m-1},0)$, we obtain
\begin{equation}
\begin{array}{r c l}
\ds \log \frac{L_{n}(\vec{\alpha},\vec{\beta}_{1},2(x+1),0)}{L_{n}(\vec{\alpha},\vec{\beta}_{m-1},2(x+1),0)} & = & \ds 2i \beta_{m} \big( \arcsin t_{m} + \sqrt{1-t_{m}^{2}} \, \big) n + i \mathcal{A} \beta_{m} \arcsin t_{m} - \frac{i\pi}{2}\mathcal{A}_{m}\beta_{m} \\[0.4cm]
& & \ds \hspace{-0.65cm} + \log \frac{G(1+\frac{\alpha_{m}}{2}+\beta_{m})G(1+\frac{\alpha_{m}}{2}-\beta_{m})}{G(1+\frac{\alpha_{m}}{2})^{2}} - \beta_{m}^{2} \log(4\pi \rho(t_{m})n(1-t_{m}^{2})) \\[0.4cm]
& & \ds \hspace{-0.65cm}  + \sum_{j=1}^{m-1} 2\beta_{j}\beta_{m} \log T_{jm} + \bigO \Big( \frac{\log n}{n^{1-4\beta_{\max}}} \Big).
\end{array}
\end{equation}
Summing all the contributions, as $n \to + \infty$ we obtain

	\begin{align}\label{al to be integration}
& \log \frac{L_{n}(\vec{\alpha},\vec{\beta},2(x+1),0)}{L_{n}(\vec{\alpha},\vec{0},2(x+1),0)} = 2i n \sum_{j=1}^{m} \beta_{j} \Big(\arcsin t_{j} + \sqrt{1-t_{\smash{j}}^{2}}\Big) + i \mathcal{A} \sum_{j=1}^{m}\beta_{j}\arcsin t_{j} \nonumber \\ &
- \frac{i\pi}{2} \sum_{j=1}^{m} \mathcal{A}_{j}\beta_{j} + \sum_{j=1}^{m} \log \frac{G(1+\frac{\alpha_{j}}{2}+\beta_{j})G(1+\frac{\alpha_{j}}{2}-\beta_{j})}{G(1+\frac{\alpha_{j}}{2})^{2}} - \sum_{j=1}^{m} \beta_{j}^{2} \log(4\pi \rho(t_{j})n(1-t_{j}^{2})) \nonumber \\ &
+ 2 \sum_{1 \leq j < k \leq m} \beta_{j}\beta_{k} \log T_{jk} + \bigO \Big( \frac{\log n}{n^{1-4\beta_{\max}}} \Big).
\end{align}
The claim of Proposition \ref{61} follows now by summing \eqref{al0 to al integration} and \eqref{al to be integration} using the definition of $\mathcal{A}_{j}$ given in \eqref{def of Aj}.
\newpage
\section{Integration in $V$}\label{Section: integration in V}
In this section, we obtain asymptotics for general Laguerre-type and Jacobi-type weights by means of a deformation parameter $s$ and by using the analysis of Section \ref{Section: steepest descent} for the weight 
\begin{equation}\label{def weight s}
w_s(x)=e^{-nV_s(x)} \om(x),
\end{equation}
where we emphasize in the notation the dependence in $s$. We specify in Subsection \ref{subsection: def param s} the exact deformations we consider. In Subsection \ref{BWW JLG}, we adapt several identities from \cite{BerWebbWong} (that are valid for Gaussian-type weights) for our situations. Finally, we proceed with the integration in $s$ for Laguerre-type and Jacobi-type weights in Subsection \ref{subsection: s for Laguerre} and Subsection \ref{subsection: s for Jacobi}, respectively.
\subsection{Deformation parameters $s$}\label{subsection: def param s}
Inspired by \cite{BerWebbWong, Charlier, DeiftOpe, DeiftItsKrasovsky}, for each $s \in [0,1]$, we define
\begin{align}
& V_{s}(x) = (1-s)2(x+1) + s V(x), & & \mbox{for Laguerre-type weights,} \label{def Vs for Laguerre} \\
& V_{s}(x) = s V(x), & & \mbox{for Jacobi-type weights.} \label{def Vs for Jacobi}
\end{align}
If $s = 0$, we already know large $n$ asymptotics for the associated Hankel determinants (from Section \ref{Section: FH integration} and the result of \cite{DIK}, see Proposition \ref{61} and Theorem \ref{DIK-Starting Point Jacobi}). It follows easily from \eqref{var equality}-\eqref{var inequality} that $V_{s}$ is one-cut regular for each $s \in [0,1]$, and the associated density $\psi_{s}$ and Euler-Lagrange constant $\ell_{s}$ are given by
\begin{align}
& \psi_{s}(x) = (1-s) \frac{1}{\pi} + s \psi(x), & & \ell_{s} = (1-s)(2+2\log 2) + s \ell, \\
& \psi_{s}(x) = (1-s) \frac{1}{\pi} + s \psi(x), & & \ell_{s} = (1-s)2\log 2 + s \ell,
\end{align}
where the first and second lines read for Laguerre-type and Jacobi-type weights respectively. We will use the differential identities 
\begin{align}
& \partial_{s} \log L_{n}(\vec{\alpha},\vec{\beta},V_{s},0) =  \frac{1}{2\pi i}\int_{-1}^{+\infty}[Y^{-1}(x)Y^{\prime}(x)]_{21}\partial_{s}w_{s}(x)dx, \label{diff identity L s} \\
& \partial_{s} \log J_{n}(\vec{\alpha},\vec{\beta},V_{s},0) =  \frac{1}{2\pi i}\int_{-1}^{1}[Y^{-1}(x)Y^{\prime}(x)]_{21}\partial_{s}w_{s}(x)dx, \label{diff identity J s}
\end{align}
which were obtained in Proposition \ref{prop: general diff identity}. Our objective in this section is to compute asymptotics of these differential identities, and finally integrate them in the parameter $s$ from $0$ to $1$. 
\subsection{Some identities}\label{BWW JLG}
We generalize here several formulas of \cite{BerWebbWong} (valid only for Gaussian-type potentials) for all three-types of canonical one-cut regular potentials. Most of the proofs are minor modifications of those done in \cite{BerWebbWong}.
\begin{lemma}\label{lem1}
For $t \in [-1,1]$, we have
\begin{align}
& \dashint_{-1}^{1} \frac{V^{\prime}(x)\sqrt{1-x^{2}}}{x-t}dx = -2\pi + 2\pi^{2} \sqrt{1-t^{2}}\rho(t),  \label{principal value 1} \\
& \int_{t}^{1}\rho(x) dx = \frac{\sqrt{1-t^{2}}}{2\pi^{2}}\dashint_{-1}^{1} \frac{V(x)}{t-x} \frac{dx}{\sqrt{1-x^{2}}} + \frac{1}{\pi}\arccos t.  \label{principal value 2}
	\end{align}
\end{lemma}
\begin{proof}
The proof goes as in \cite[Lemma 5.8]{BerWebbWong}. Let $H:\C \setminus [-1,1] \to \C$ be defined by
\begin{equation}
H(z)=2\pi \sqrt{z-1} \sqrt{z+1} \int_{-1}^{1} \frac{\rho(x)}{x-z}dx + \int_{-1}^{1} \frac{V'(x)\sqrt{1-x^2}}{x-z}dx
\end{equation} 
where the principal branches are chosen for $\sqrt{z-1}$ and $\sqrt{z+1}$. For $t \in (-1,1) $, one can check that  $H_+(t)=H_-(t)$. Also $H$ is bounded at $\pm 1$ and $H(\infty)=-2\pi$; so Liouville's theorem implies that $H(z)=-2\pi$. Considering $H_+(t)+H_-(t)$ for $t \in (-1,1)$ yields	(\ref{principal value 1}). Now, (\ref{principal value 2}) follows from (\ref{principal value 1}) and the following identity which is proved in \cite[eq (5.18) and below]{BerWebbWong} \begin{equation}\label{p=q}
 \sqrt{1-t^2} \dashint_{-1}^{1} \frac{V(x)}{t-x} \frac{dx}{\sqrt{1-x^{2}}}	=	\int_{t}^{1} \frac{1}{\sqrt{1-x^2}} \left(\dashint^{1}_{-1} \frac{V^{\prime}(y)}{y-x}\sqrt{1-y^2} dy \right) dx.
\end{equation}
\end{proof}
\begin{lemma}\label{lemma: integral tk a}
Let $\mathcal{C}$ be a closed curve surrounding $[-1,1]$ in the clockwise direction, let $a(z) = \sqrt[4]{\frac{z-1}{z+1}}$ be analytic on $\mathbb{C}\setminus [-1,1]$ such that $a(z) \sim 1$ as $z \to \infty$, and let $f$ be analytic in a neighbourhood of $[-1,1]$. We have
\begin{align}
	&\frac{1}{2\pi i} \int_{\mathcal{C}} \left[ \frac{a^2(z)}{a_+^2(t_j)} +\frac{a_+^2(t_j)}{a^2(z)}\right] \frac{f(z)}{(z-t_j)^2}dz = \frac{2}{\pi i \sqrt{1-t_{\smash{j}}^2}} \dashint^{1}_{-1} f^{\prime}(x) \frac{\sqrt{1-x^2}}{x-t_j}dx, \label{a+a} \\
& \frac{1}{2\pi i} \int_{\mathcal{C}} \left[   \frac{a^2(z)}{a_+^2(t_j)} - \frac{a_+^2(t_j)}{a^2(z)}  \right] \frac{f(z)}{(z-t_j)^2}dz = \frac{2}{\pi i \sqrt{1-t_{\smash{j}}^2}} \dashint^{1}_{-1} \frac{f(x)}{(t_j-x)\sqrt{1-x^2}}dx. \label{a-a}
\end{align}
\end{lemma}
\begin{proof}
The proof is the same as in \cite[equations (5.22)-(5.23) and above]{BerWebbWong}.
\end{proof}
Applying Lemma \ref{lemma: integral tk a} to $f = \partial_{s}V_{s}$ (with $V_{s}$ given by \eqref{def Vs for Laguerre}--\eqref{def Vs for Jacobi}), and then simplifying using Lemma \ref{lem1}, we obtain
\begin{equation}\label{plus}
\int_{\mathcal{C}} \left[ \frac{a^2(z)}{a_+^2(t_j)} +\frac{a_+^2(t_j)}{a^2(z)}\right] \frac{\partial_s V_s(z)}{(z-t_j)^2}dz = \left\{ \begin{array}{l l}
\ds 8 \pi^{2} \big(  \psi(t_{j})-\tfrac{1}{\pi}\big) \frac{\sqrt{1-t_{j}}}{\sqrt{1+t_{j}}}, & \mbox{for Laguerre-type potentials} \\[0.3cm]
\ds 8 \pi^{2} \big(  \psi(t_{j})-\tfrac{1}{\pi}\big) \frac{1}{\sqrt{1-t_{\smash{j}}^{2}}} , & \mbox{for Jacobi-type potentials}
\end{array} \right. 
\end{equation}
and
\begin{equation}\label{minus}
\int_{\mathcal{C}} \left[   \frac{a^2(z)}{a_+^2(t_j)} - \frac{a_+^2(t_j)}{a^2(z)}  \right] \frac{\partial_s V_s(z)}{(z-t_j)^2}dz = \left\{ \begin{array}{l l}
\ds \frac{8\pi^2}{1-t^2_j}\int_{t_j}^{1} \big(\psi(x)-\tfrac{1}{\pi}\big) \frac{\sqrt{1-x}}{\sqrt{1+x}} dx, & \mbox{for Laguerre-type potentials} \\[0.3cm]
\ds  \frac{8\pi^2}{1-t^2_j}\int_{t_j}^{1} \big(\psi(x)-\tfrac{1}{\pi}\big) \frac{1}{\sqrt{1-x^2}} dx , & \mbox{for Jacobi-type potentials}
\end{array} \right.
\end{equation}
\newpage
\begin{lemma}\label{lemma: integral 1 -1}
Let $\mathcal{C}$ be a closed curve surrounding $[-1,1]$ in the clockwise direction, let $a(z) = \sqrt[4]{\frac{z-1}{z+1}}$ be analytic on $\mathbb{C}\setminus [-1,1]$ such that $a(z) \sim 1$ as $z \to \infty$, and let $f$ be analytic in a neighbourhood of $[-1,1]$. We have
\begin{align}
& \int_{\mathcal{C}} \frac{a(z)^{2}}{(z-1)^{2}}f(z)dz = 2i \int_{-1}^{1} \frac{f^{\prime}(x)\sqrt{1-x^{2}}}{x-1}dx, \label{lol5} \\
& \int_{\mathcal{C}} \frac{a(z)^{3}}{(z-1)^{3}}f(z)dz = -\frac{2i}{3}\int_{-1}^{1} \frac{f^{\prime}(x)\sqrt{1-x^{2}}}{x-1}dx + \frac{2i}{3} \left.\frac{d}{dt} \right|_{t=1} \dashint_{-1}^{1} \frac{f^{\prime}(x)\sqrt{1-x^{2}}}{x-t}dx, \label{lol6} \\
& \int_{\mathcal{C}} \frac{a(z)^{-2}}{(z-1)^{3}}f(z)dz = \frac{2i}{3}\int_{-1}^{1} \frac{f^{\prime}(x)\sqrt{1-x^{2}}}{x-1}dx + \frac{4i}{3} \left.\frac{d}{dt} \right|_{t=1} \dashint_{-1}^{1} \frac{f^{\prime}(x)\sqrt{1-x^{2}}}{x-t}dx, \label{lol7} \\
& \int_{\mathcal{C}} \frac{a(z)^{-2}}{(z+1)^{2}}f(z)dz = -2i \int_{-1}^{1} \frac{f^{\prime}(x)\sqrt{1-x^{2}}}{x+1}dx. \label{lol8}
\end{align}
\end{lemma}
\begin{proof}
The proof of \eqref{lol5}--\eqref{lol7} is done in \cite[Lemma 5.10]{BerWebbWong}, and the proof for \eqref{lol8} is similar.
\end{proof}
Applying Lemma \ref{lemma: integral 1 -1} to $f(x) = \partial_{s}V_{s} = V(x)-2(x+1)$ with $V_{s}$ given by \eqref{def Vs for Laguerre} for Laguerre-type potentials, and then simplifying using Lemma \ref{lem1}, we obtain
\begin{align}
& \int_{\mathcal{C}} \frac{a(z)^{2}}{(z-1)^{2}}\partial_{s}V_{s}(z)dz = 0, \label{lol10} \\
& \int_{\mathcal{C}} \frac{a(z)^{2}}{(z-1)^{3}}\partial_{s}V_{s}(z)dz = - \frac{4\pi^{2}i}{3}\Big( \psi(1)-\frac{1}{\pi} \Big), \label{lol11} \\
& \int_{\mathcal{C}} \frac{a(z)^{-2}}{(z-1)^{2}}\partial_{s}V_{s}(z)dz = -\frac{8\pi^{2}i}{3}\Big( \psi(1)-\frac{1}{\pi} \Big), \label{lol12} \\
& \int_{\mathcal{C}} \frac{a(z)^{-2}}{(z-1)^{2}}\partial_{s}V_{s}(z)dz = -8\pi^{2}i \Big( \psi(-1)-\frac{1}{\pi} \Big). \label{lol13}
\end{align}
Similarly, for Jacobi-type weights with $f(x) = \partial_{s}V_{s} = V(x)$ with $V_{s}$ given by \eqref{def Vs for Jacobi} for Jacobi-type potentials, we obtain
\begin{align}
& \int_{\mathcal{C}} \frac{a(z)^{2}}{(z-1)^{2}}\partial_{s}V_{s}(z)dz = 4\pi^{2}i \Big( \psi(1)-\frac{1}{\pi} \Big), \label{lol14} \\
& \int_{\mathcal{C}} \frac{a(z)^{-2}}{(z+1)^{2}}\partial_{s}V_{s}(z)dz = -4\pi^{2}i \Big( \psi(-1)-\frac{1}{\pi} \Big). \label{lol15}
\end{align}

\subsection{Integration in $s$ for Laguerre-type weights}\label{subsection: s for Laguerre}
In this subsection we prove Proposition \ref{prop 71} below. 
\begin{proposition}\label{prop 71}
As $n \to +\infty$, we have
\begin{align}\label{prop 71 r}
& \log \frac{L_{n}(\vec{\alpha},\vec{\beta},V,0)}{L_{n}(\vec{\alpha},\vec{\beta},2(x+1),0)} = - \frac{n^{2}}{2} \int_{-1}^{1} \big(V(x)-2(x+1)\big)\Big( \frac{1}{\pi} + \psi(x)\Big)\sqrt{\frac{1-x}{1+x}}dx  \nonumber \\
	& + n \sum_{j=0}^{m} \frac{\alpha_{j}}{2}\big(V(t_{j})-2(1+t_{j})\big) - \frac{n\mathcal{A}}{2\pi} \int_{-1}^{1} \frac{V(x)-2(1+x)}{\sqrt{1-x^{2}}}dx -2\pi n \sum_{j=1}^{m} i \beta_{j} \int_{t_{j}}^{1} \Big(\psi(x)-\frac{1}{\pi}\Big)\sqrt{\frac{1-x}{1+x}}dx \nonumber \\
	& \hspace{0cm} + \sum_{j=1}^{m} \bigg( \frac{\alpha_{j}^{2}}{4} - \beta_{j}^{2} \bigg) \log \left( \pi \psi(t_{j}) \right) - \frac{1}{24}\log\left( \pi\psi(1) \right) - \frac{1-4\alpha_{0}^{2}}{8}\log(\pi \psi(-1)) + \bigO\big(n^{-1+4\beta_{\max}}\big).
\end{align}
\end{proposition}
Let $\mathcal{C}$ be a closed contour surrounding $[-1,1]$ and the lenses $\ga_+ \cup \ga_-$, which is oriented clockwise and passes through $-1-\ep$ and $1+\ep$ for a certain $\ep>0$. Using the jumps for $Y$ given by \eqref{Y_Jump}, we rewrite the differential identity  \eqref{diff identity L s} as follows
\begin{equation}\label{diff s L Ln}
\partial_{s} \log L_{n}(\vec{\alpha},\vec{\beta},V_{s},0) =  \int^{+\infty}_{1+\ep} [Y^{-1}(x)Y^{\prime}(x)]_{21} \partial_s w_s(x)\frac{dx}{2\pi i} - \frac{1}{2\pi i} \int_{\mathcal{C}} [Y^{-1}(z)Y^{\prime}(z)]_{11} \partial_s \log w_s(z) \frac{dz}{2\pi i}. 
\end{equation}
From \eqref{sum}, \eqref{EL <} and by inverting the transformations $Y  \mapsto T \mapsto S \mapsto R$ outside the lenses and outside the disks, we conclude that the first integral in the r.h.s. of \eqref{diff s L Ln} is of order $\bigO(e^{-cn})$ as $n \to + \infty$, for a positive constant $c$, and that the integral over $\mathcal{C}$ can be decomposed into three integrals:
\begin{equation}\label{lol 4}
\begin{array}{r c l}
\displaystyle \partial_{s} \log L_{n}(\vec{\alpha},\vec{\beta},V_{s},0)  & = & \displaystyle I_{1,s} + I_{2,s} + I_{3,s} + \bigO(e^{-cn}), \qquad \mbox{as } n \to  \infty, \\[0.3cm]
\displaystyle I_{1,s} & = & \displaystyle \frac{-n}{2\pi i} \int_{\mathcal{C}} g^{\prime}(z) \partial_{s} \log w_{s}(z) dz, \\[0.3cm]
\displaystyle I_{2,s} & = & \displaystyle \frac{-1}{2\pi i} \int_{\mathcal{C}} [P^{(\infty)}(z)^{-1}P^{(\infty)}(z)^{\prime}]_{11} \partial_{s} \log w_{s}(z) dz, \\[0.3cm]
\displaystyle I_{3,s} & = & \displaystyle \frac{-1}{2\pi i} \int_{\mathcal{C}} [P^{(\infty)}(z)^{-1}R^{-1}(z)R^{\prime}(z)P^{(\infty)}(z)]_{11} \partial_{s} \log w_{s}(z) dz.
\end{array}
\end{equation}
In exactly the same way as in \cite{BerWebbWong, Charlier}, we show from a detailed analysis of the Cauchy operator associated to $R$ that the estimates in \eqref{R_large_n_asymp} hold uniformly for $(\vec{\alpha},\vec{\beta})$ in any fixed compact set $\Omega$, and uniformly in $s \in [0,1]$. However, from Proposition \ref{prop: general diff identity}, the identity \eqref{lol 4} itself is not valid for the values of $(\vec{\alpha},\vec{\beta},s)$ for which at least one of the polynomials $p_{0},\ldots,p_{n}$ does not exist. From \cite[beginning of Section 3]{Charlier}, this set is locally finite except possible some accumulation points at $s = 0$ and $s = 1$. As in \cite{Charlier}, we extend \eqref{lol 4} \textit{for all} $(\vec{\alpha},\vec{\beta},s)\in \Omega \times [0,1]$ (for sufficiently large $n$) using the continuity of the l.h.s. of \eqref{lol 4}. A similar reasoning holds also for \eqref{diff s L} below.

\vspace{0.3cm}\hspace{-0.55cm}Note from \eqref{def weight s} and \eqref{def Vs for Laguerre} that $\partial_{s} \log w_{s}(z) = -n \partial_{s} V_{s}(z) = -n(V(x)-2(x+1))$. Using the definition of $g$ given by \eqref{g function} and switching the order of integration, we get
\begin{equation}\label{I1s L}
I_{1,s} = -n^{2} \int_{-1}^{1} \rho_{s}(x) \partial_{s}V_{s}(x)dx = - n^{2} \int_{-1}^{1} (V(x)-2(x+1))\Big( (1-s)\frac{1}{\pi}+s\psi(x) \Big)\frac{\sqrt{1-x}}{\sqrt{1+x}}dx.
\end{equation}
Therefore, we have
\begin{equation}
\int_{0}^{1}I_{1,s}ds = -\frac{n^{2}}{2} \int_{-1}^{1} (V(x)-2(x+1)) \Big( \frac{1}{\pi}+\psi(x) \Big)\frac{\sqrt{1-x}}{\sqrt{1+x}}dx.
\end{equation}
From \eqref{Global_Parametrix}, \eqref{D_al}, \eqref{D_be} and a contour deformation, we obtain the following expression for $I_{2,s}$:
\begin{multline}\label{lol9}
I_{2,s} = n \sum_{j=0}^{m} \frac{\alpha_{j}}{2}\Big( V(t_{j})-2(1+t_{j}) \Big) - \frac{n\mathcal{A}}{2\pi} \int_{-1}^{1} \frac{V(x)-2(1+x)}{\sqrt{1-x^{2}}}dx \\
+ n \sum_{j=1}^{m} \frac{i \beta_{j}}{\pi}\sqrt{1-t_{j}^{2}} \dashint_{-1}^{1} \frac{V(x)-2(1+x)}{\sqrt{1-x^{2}}(x-t_{j})}dx.
\end{multline}
We simplify the last integral of \eqref{lol9} using \eqref{principal value 2}:
\begin{equation}
\sqrt{1-t_{j}^{2}}\dashint_{-1}^{1} \frac{V(x)-2(1+x)}{\sqrt{1-x^{2}}(x-t_{j})}dx = -2\pi^{2} \int_{t_{j}}^{1} \Big( \psi(x)-\frac{1}{\pi} \Big) \sqrt{\frac{1-x}{1+x}}dx.
\end{equation}
Then, integrating in $s$ (note that $I_{2,s}$ is in fact independent of $s$), we obtain
\begin{multline}
\int_{0}^{1}I_{2,s}ds = n \sum_{j=0}^{m} \frac{\alpha_{j}}{2}\Big( V(t_{j})-2(1+t_{j}) \Big) - \frac{n\mathcal{A}}{2\pi} \int_{-1}^{1} \frac{V(x)-2(1+x)}{\sqrt{1-x^{2}}}dx \\
- 2\pi n \sum_{j=1}^{m} i \beta_{j} \int_{t_{j}}^{1}\Big( \psi(x)-\frac{1}{\pi} \Big) \sqrt{\frac{1-x}{1+x}}dx.
\end{multline}
Using the expansion of $R$ given by \eqref{R_large_n_asymp}, we have
\begin{equation}
I_{3,s} = \frac{1}{2\pi i} \int_{\mathcal{C}} [P^{(\infty)}(z)^{-1}R^{(1)}(z)^{\prime}P^{(\infty)}(z)]_{11} \partial_{s}V_{s}(z)dz + \bigO\big(n^{-1+4\beta_{\max}}\big), \qquad \mbox{ as } n \to \infty,
\end{equation}
The leading term of $I_{3,s}$ can be written down more explicitly using the definition of $P^{(\infty)}$ given by \eqref{Global_Parametrix}, and we obtain
\begin{multline}\label{lol 12n}
I_{3,s} = \frac{1}{2\pi i} \int_{\mathcal{C}} \left( \frac{a(z)^{2}+a(z)^{-2}}{4}[R_{11}^{(1)}(z)^{\prime} - R_{22}^{(1)}(z)^{\prime}] + \frac{1}{2}[R_{11}^{(1)}(z)^{\prime}+R_{22}^{(1)}(z)^{\prime}] \right. \\ \left. + i \frac{a(z)^{2}-a(z)^{-2}}{4} [R_{12}^{(1)}(z)^{\prime}D_{\infty}^{-2} + R_{21}^{(1)}(z)^{\prime}D_{\infty}^{2}] \right) \big(V(z)-2(z+1)\big)dz + \bigO\big(n^{-1+4\beta_{\max}}\big).
\end{multline}
From \eqref{R^{(1)}}, \eqref{res t_{k}}, \eqref{res at -1L}, \eqref{res at 1'L} and \eqref{res at 1L} we have
\begin{align}
& R_{11}^{(1)\prime}(z) - R_{22}^{(1)\prime}(z) = \sum_{j=1}^{m} \frac{1}{(z-t_{j})^{2}} \frac{-2v_{j}(t_{j}+\widetilde{\Lambda}_{I,j})}{2\pi \rho_{s}(t_{j}) \sqrt{1-t_{\smash{j}}^{2}}}+ \frac{1}{(z-1)^{3}} \frac{5}{2^{2}3\pi \psi_{s}(1)} \nonumber \\
& \hspace{2.8cm} +\frac{1}{(z-1)^{2}} \frac{(\mathcal{A}-\widetilde{\mathcal{B}}_{1})^{2} - \frac{1}{4} - \frac{1}{2}\frac{\psi_{s}^{\prime}(1)}{\psi_{s}(1)}}{2^{2}\pi \psi_{s}(1)}+\frac{1}{(z+1)^{2}} \frac{1-4\alpha_{0}^{2}}{2^{4}\pi \psi_{s}(-1)}, \label{I3 1} \\[0.3cm]
& R_{11}^{(1)\prime}(z) + R_{22}^{(1)\prime}(z) = 0, \label{I3 2} \\
& i[R_{12}^{(1)\prime}(z)D_{\infty}^{-2}+R_{21}^{(1)\prime}(z)D_{\infty}^{2}] = \sum_{j=1}^{m} \frac{1}{(z-t_{j})^{2}} \frac{v_{j}(-2+\widetilde{\Lambda}_{R,1,j}-\widetilde{\Lambda}_{R,2,j})}{2\pi \rho_{s}(t_{j}) \sqrt{1-t_{\smash{j}}^{2}}}+ \frac{1}{(z-1)^{3}} \frac{5}{2^{2}3\pi \psi_{s}(1)}  \nonumber \\
& \hspace{4.5cm} +\frac{1}{(z-1)^{2}} \frac{(\mathcal{A}-\widetilde{\mathcal{B}}_{1})^{2} + \frac{11}{12} - \frac{1}{2}\frac{\psi_{s}^{\prime}(1)}{\psi_{s}(1)}}{2^{2}\pi \psi_{s}(1)} +\frac{1}{(z+1)^{2}} \frac{-(1-4\alpha_{0}^{2})}{2^{4}\pi \psi_{s}(-1)}. \label{I3 3}
\end{align}
Therefore, from \eqref{lol 12n}--\eqref{I3 3} and using the connection formula \eqref{connection Lambdas}, we obtain
\begin{equation}\label{lol 3 n}
I_{3,s} = \sum_{j=1}^{m} I_{3,s,t_{j}} + I_{3,s,1} + I_{3,s,-1} + \bigO\big(n^{-1+4\beta_{\max}}\big), \qquad \mbox{ as } n \to \infty,
\end{equation}
where
\begin{align}
& I_{3,s,t_{k}} = \frac{-v_{k}}{8\pi^{2} \rho_{s}(t_{k})} \int_{\mathcal{C}} \left[ \frac{a^{2}(z)}{a_{+}^{2}(t_{k})} + \frac{a_{+}^{2}(t_{k})}{a^{2}(z)} + \widetilde{\Lambda}_{I,k} \left( \frac{a^{2}(z)}{a_{+}^{2}(t_{k})} - \frac{a_{+}^{2}(t_{k})}{a^{2}(z)} \right) \right] \frac{\partial_{s}V_{s}(z)}{(z-t_{k})^{2}}dz, \nonumber \\
& I_{3,s,1} = \int_{\mathcal{C}} \Bigg[ \frac{a^{2}(z)}{4\pi \psi_{s}(1)} \Bigg( \frac{2(\mathcal{A}-\widetilde{\mathcal{B}}_{1})^{2}+\frac{2}{3}-\frac{\psi_{s}^{\prime}(1)}{\psi_{s}(1)}}{2^{2}(z-1)^{2}}+\frac{5}{6(z-1)^{3}} \Bigg) + \frac{a^{-2}(z)}{4(z-1)^{2}}\frac{-\frac{7}{6}}{2^{2}\pi \psi_{s}(1)}  \Bigg] \partial_{s}V_{s}(z)\frac{dz}{2\pi i}, \nonumber \\
& I_{3,s,-1}  =  \int_{\mathcal{C}}  \Bigg[  \frac{a^{-2}(z)}{4(z+1)^{2}}\frac{1-4\alpha_{0}^{2}}{2^{3}\pi \psi_{s}(-1)}  \Bigg]  \partial_{s}V_{s}(z)\frac{dz}{2\pi i}. \nonumber
\end{align}
Formulas \eqref{plus} and \eqref{minus} allow us to simplify $I_{3,s,t_{k}}$ as follows:
\begin{equation}
I_{3,s,t_{k}} = - \frac{v_{k}}{\psi_{s}(t_{k})}\Big( \psi(t_{k})-\frac{1}{\pi} \Big) - \frac{v_{k} \widetilde{\Lambda}_{I,k}}{\rho_{s}(t_{k})(1-t_{k}^{2})}\int_{t_{k}}^{1} \Big( \psi(x)-\frac{1}{\pi}\Big) \sqrt{\frac{1-x}{1+x}}dx.
\end{equation}
Integrating the above from $s = 0$ to $s = 1$, we have
\begin{equation}\label{int I3t}
\int_{0}^{1} I_{3,s,t_{k}} ds = -v_{k} \log (\pi \psi(t_{k})) - \frac{v_{k}}{1-t_{k}^{2}}\int_{t_{k}}^{1} \Big( \psi(x)-\frac{1}{\pi} \Big) \sqrt{\frac{1-x}{1+x}}dx \int_{0}^{1} \frac{\widetilde{\Lambda}_{I,k}}{\rho_{s}(t_{k})}ds.
\end{equation}
By the same argument as the one given in \cite[equations (6.23) and (6.24)]{Charlier}, the second term in the r.h.s of \eqref{int I3t} is of order $\bigO(n^{-1+2|\Re \be_k|})$ as $n \to +\infty$, that is,
\begin{equation}\label{int I3t L}
\int_{0}^{1} I_{3,s,t_{ k}} ds = -v_{k} \log (\pi \psi(t_{k})) +\bigO(n^{-1+2|\Re \beta_{k}|}).
\end{equation}
We can also simplify the expression for $I_{3,s,1}$. Using the formulas \eqref{lol10}--\eqref{lol12}, we obtain 
\begin{equation}
I_{3,s,1} = -\frac{1}{24}\frac{\psi(1)-\frac{1}{\pi}}{\psi_{s}(1)}, \qquad \mbox{and then} \qquad \int_{0}^{1}I_{3,s,1}ds = - \frac{1}{24} \log(\pi \psi(1)).
\end{equation}
Similarly, using \eqref{lol13} we get 
\begin{equation}
I_{3,s,-1} = -\frac{1-4\alpha_{0}^{2}}{8} \frac{\psi(-1)-\frac{1}{\pi}}{\psi_{s}(-1)}, \quad \mbox{and then} \quad \int_{0}^{1} I_{3,s,-1} ds = -\frac{1-4\alpha_{0}^{2}}{8} \log \big( \pi \psi(-1) \big).
\end{equation}
This finishes the proof of Proposition \ref{prop 71}.

\subsection{Jacobi-type weights}\label{subsection: s for Jacobi}
We prove here the analogue of Proposition \ref{prop 71} for Jacobi-type weights.

\begin{proposition}\label{prop 72}
As $n \to \infty$, we have
\begin{align}\label{prop 72 r}
& \log \frac{J_{n}(\vec{\alpha},\vec{\beta},V,0)}{J_{n}(\vec{\alpha},\vec{\beta},0,0)} = - \frac{n^{2}}{2} \int_{-1}^{1}\frac{V(x)}{\sqrt{1-x^{2}}}\Big( \frac{1}{\pi} + \psi(x)\Big)dx + n \sum_{j=0}^{m+1} \frac{\alpha_{j}}{2} V(t_{j}) \nonumber \\
&  - \frac{n\mathcal{A}}{2\pi} \int_{-1}^{1} \frac{V(x)}{\sqrt{1-x^{2}}}dx -2\pi n \sum_{j=1}^{m} i \beta_{j} \int_{t_{j}}^{1} \Big(\psi(x)-\frac{1}{\pi}\Big) \frac{dx}{\sqrt{1-x^{2}}} + \sum_{j=1}^{m} \bigg( \frac{\alpha_{j}^{2}}{4} - \beta_{j}^{2} \bigg) \log \left( \pi \psi(t_{j}) \right) \nonumber \\
& \hspace{0cm}   - \frac{1-4\alpha_{m+1}^{2}}{8}\log\left( \pi\psi(1) \right) - \frac{1-4\alpha_{0}^{2}}{8}\log(\pi \psi(-1)) + \bigO\big(n^{-1+4\beta_{\max}}\big).
\end{align}
\end{proposition}
The computations of this subsection are organised similarly to those done in Subsection \ref{subsection: s for Laguerre}, and we provide less details. Let $\mathcal{C}$ be a closed contour surrounding $[-1,1]$ and the lenses $\ga_+ \cup \ga_-$, which is oriented clockwise and passes through $-1-\ep$ and $1+\ep$ for a certain $\ep>0$. Using the jumps for $Y$ \eqref{Y_Jump}, we rewrite the differential identity  \eqref{diff identity J s} as follows
\begin{equation}\label{diff s L}
\partial_{s} \log J_{n}(\vec{\alpha},\vec{\beta},V_{s},0) =  - \frac{1}{2\pi i} \int_{\mathcal{C}} [Y^{-1}(z)Y^{\prime}(z)]_{11} \partial_s \log w_s(z) \frac{dz}{2\pi i},
\end{equation}
where from \eqref{def weight s} and \eqref{def Vs for Jacobi}, we have $\partial_{s} \log w_{s}(z) = -n\partial_{s}V_{s}(z) = -nV(z)$. In the same way as done in \eqref{lol 4}, by inverting the transformations $Y \mapsto T \mapsto S \mapsto R$ in the region outside the lenses and outside the disks, we have
\begin{equation}
\partial_{s} \log J_{n}(\vec{\alpha},\vec{\beta},V_{s},0)  = \displaystyle I_{1,s} + I_{2,s} + I_{3,s},
\end{equation}
where $I_{1,s}$, $I_{2,s}$ and $I_{3,s}$ are given as in \eqref{lol 4}. For $I_{1,s}$, a simple calculation implies
\begin{equation}
I_{1,s} = -n^{2} \int_{-1}^{1} \rho_{s}(x) \partial_{s}V_{s}(x)dx = - n^{2} \int_{-1}^{1} \frac{1}{\sqrt{1-x^{2}}}V(x)\Big( (1-s)\frac{1}{\pi}+s\psi(x) \Big)dx,
\end{equation}
which gives
\begin{equation}
\int_{0}^{1}I_{1,s}ds = -\frac{n^{2}}{2} \int_{-1}^{1} \frac{V(x)}{\sqrt{1-x^{2}}}\Big( \frac{1}{\pi}+\psi(x) \Big)dx.
\end{equation}
The computations of $I_{2,s}$ are similar to those done for \cite[equations (6.10)--(6.15)]{Charlier} and for \eqref{lol9}. We obtain
\begin{multline}
\int_{0}^{1}I_{2,s}ds = n \sum_{j=0}^{m+1} \frac{\alpha_{j}}{2} V(t_{j}) - \frac{n\mathcal{A}}{2\pi} \int_{-1}^{1} \frac{V(x)}{\sqrt{1-x^{2}}}dx
- 2\pi n \sum_{j=1}^{m} i \beta_{j} \int_{t_{j}}^{1}\Big( \psi(x)-\frac{1}{\pi} \Big) \frac{1}{\sqrt{1-x^{2}}}dx.
\end{multline}
For $I_{3,s}$, similar to \eqref{lol 12n} we get
\begin{multline}\label{lol 12}
I_{3,s} = \frac{1}{2\pi i} \int_{\mathcal{C}} \left( \frac{a(z)^{2}+a(z)^{-2}}{4}[R_{11}^{(1)}(z)^{\prime} - R_{22}^{(1)}(z)^{\prime}] + \frac{1}{2}[R_{11}^{(1)}(z)^{\prime}+R_{22}^{(1)}(z)^{\prime}] \right. \\ \left. + i \frac{a(z)^{2}-a(z)^{-2}}{4} [R_{12}^{(1)}(z)^{\prime}D_{\infty}^{-2} + R_{21}^{(1)}(z)^{\prime}D_{\infty}^{2}] \right) V(z)dz + \bigO\big(n^{-1+4\beta_{\max}}\big).
\end{multline}
The quantities involving $R^{(1)}$ are made explicit using \eqref{R^{(1)}Jacobi}, we obtain
\begin{align*}
& R_{11}^{(1)\prime}(z) - R_{22}^{(1)\prime}(z) = \sum_{j=1}^{m} \frac{1}{(z-t_{j})^{2}} \frac{-2v_{j}(t_{j}+\widetilde{\Lambda}_{I,j})}{2\pi \rho_{s}(t_{j}) \sqrt{1-t_{\smash{j}}^{2}}} +\frac{1}{(z-1)^{2}} \frac{4\alpha_{m+1}^{2}-1}{2^{3}\pi \psi_{s}(1)}+\frac{1}{(z+1)^{2}} \frac{1-4\alpha_{0}^{2}}{2^{3}\pi \psi_{s}(-1)}, \nonumber \\[0.3cm]
& R_{11}^{(1)\prime}(z) + R_{22}^{(1)\prime}(z) = 0, \\
& i[R_{12}^{(1)\prime}(z)D_{\infty}^{-2}+R_{21}^{(1)\prime}(z)D_{\infty}^{2}] = \sum_{j=1}^{m} \frac{1}{(z-t_{j})^{2}} \frac{v_{j}(-2+\widetilde{\Lambda}_{R,1,j}-\widetilde{\Lambda}_{R,2,j})}{2\pi \rho_{s}(t_{j}) \sqrt{1-t_{\smash{j}}^{2}}}+ \frac{1}{(z-1)^{2}} \frac{-(1-4\alpha_{m+1}^{2})}{2^{3}\pi \psi_{s}(1)} \\
& \hspace{4.5cm} +\frac{1}{(z+1)^{2}} \frac{-(1-4\alpha_{0}^{2})}{2^{3}\pi \psi_{s}(-1)}.
\end{align*}
As in Subsection \ref{subsection: s for Laguerre}, we rewrite $I_{3,s}$ in the form
\begin{equation}\label{lol 3}
I_{3,s} = \sum_{j=1}^{m} I_{3,s,t_{j}} + I_{3,s,1} + I_{3,s,-1} + \bigO\big(n^{-1+4\beta_{\max}}\big), \qquad \mbox{ as } n \to \infty,
\end{equation}
where
\begin{align}
& I_{3,s,t_{k}} = \frac{-v_{k}}{8\pi^{2} \rho_{s}(t_{k})} \int_{\mathcal{C}} \left[ \frac{a^{2}(z)}{a_{+}^{2}(t_{k})} + \frac{a_{+}^{2}(t_{k})}{a^{2}(z)} + \widetilde{\Lambda}_{I,k} \left( \frac{a^{2}(z)}{a_{+}^{2}(t_{k})} - \frac{a_{+}^{2}(t_{k})}{a^{2}(z)} \right) \right] \frac{\partial_{s}V_{s}(z)}{(z-t_{k})^{2}}dz, \nonumber \\
& I_{3,s,1} = \frac{4\alpha_{m+1}^{2}-1}{2^{5}\pi^{2}i\psi_{s}(1)}\int_{\mathcal{C}} \frac{a^{2}(z)}{(z-1)^{2}} \partial_{s}V_{s}(z)dz, \nonumber \\
& I_{3,s,-1}  = \frac{1-4\alpha_{0}^{2}}{2^{5}\pi^{2}i\psi_{s}(-1)} \int_{\mathcal{C}}  \frac{a^{-2}(z)}{(z+1)^{2}} \partial_{s}V_{s}(z) dz. \nonumber
\end{align}
From \eqref{plus} and \eqref{minus}, $I_{3,s,t_{k}}$ simplifies to
\begin{equation}
I_{3,s,t_{k}} = - \frac{v_{k}}{\psi_{s}(t_{k})}\Big( \psi(t_{k})-\frac{1}{\pi} \Big) - \frac{v_{k} \widetilde{\Lambda}_{I,k}}{\rho_{s}(t_{k})(1-t_{k}^{2})}\int_{t_{k}}^{1} \Big( \psi(x)-\frac{1}{\pi}\Big) \frac{dx}{\sqrt{1-x^{2}}}
\end{equation}
and hence, similarly to \eqref{int I3t}--\eqref{int I3t L}, as $n\to + \infty$ we have
\begin{equation}
\int_{0}^{1} I_{3,s,t_{k}} ds = -v_{k} \log (\pi \psi(t_{k})) +\bigO(n^{-1+2|\Re \beta_{k}|}).
\end{equation}
Also, from \eqref{lol14}--\eqref{lol15}, we have
\begin{equation}
I_{3,s,1} = -\frac{1-4\alpha_{m+1}^{2}}{8\psi_{s}(1)}\Big( \psi(1)-\frac{1}{\pi} \Big) \qquad \mbox{and} \qquad I_{3,s,-1} = -\frac{1-4\alpha_{0}^{2}}{8\psi_{s}(-1)} \Big( \psi(-1)-\frac{1}{\pi} \Big),
\end{equation}
and hence
\begin{equation}
\int_{0}^{1}I_{3,s,1}ds = -\frac{1-4\alpha_{m+1}^{2}}{8}\log(\pi \psi(1)) \qquad \mbox{and} \qquad \int_{0}^{1} I_{3,s,-1} ds = -\frac{1-4\alpha_{0}^{2}}{8} \log \big( \pi \psi(-1) \big).
\end{equation}
This concludes the proof of proposition \ref{prop 72}.
\section{Integration in $W$}\label{Section: integration in W}
The main result of this section is the following.
\begin{proposition}\label{81 L}
As $n \to \infty$, we have
\begin{multline}\label{81 L r}
\log \frac{D_{n}(\vec{\alpha},\vec{\beta},V,W)}{D_{n}(\vec{\alpha},\vec{\beta},V,0)} = n\int_{-1}^{1}W(x)\rho(x)dx -\frac{1}{4\pi^{2}}\int_{-1}^{1}  \frac{W(y)}{\sqrt{1-y^{2}}} \bigg(\dashint_{-1}^{1} \frac{W^{\prime}(x) \sqrt{1-x^{2}}}{x-y}dx\bigg)dy \\
	+ \frac{\mathcal{A}}{2\pi}\int_{-1}^{1} \frac{W(x)}{\sqrt{1-x^{2}}}dx - \sum_{j=0}^{m+1} \frac{\alpha_{j}}{2}W(t_{j}) +  \sum_{j=1}^{m} \frac{i \beta_{j}}{\pi} \sqrt{1-t_{j}^{2}} \dashint_{-1}^{1} \frac{W(x)}{\sqrt{1-x^{2}}(t_{j}-x)}dx + \bigO\big(n^{-1+2\beta_{\max}}\big).
\end{multline}
where $D_n$ stands for either $L_n$ or $J_n$.
\end{proposition}
\begin{remark}
The difference between Laguerre-type and Jacobi-type weights in the r.h.s. of \eqref{81 L r} is only reflected in the definitions of $\rho$ and $\mathcal{A}$. 
\end{remark}
The proof of Proposition \ref{81 L r} goes in a similar way as in \cite{Charlier}. For each $t \in [0,1]$, we define
\begin{equation}
W_{t}(z) = \log \big( 1-t+te^{W(z)} \big),
\end{equation}
where the principal branch is taken for the log. For every $t \in [0,1]$, $W_{t}$ is analytic on a neighbourhood of $[-1,1]$ (independent of $t$) and is still Hölder continuous on $\mathcal{I}$. This deformation is the same as the one used in \cite{DeiftItsKrasovsky, BerWebbWong, Charlier}. Therefore, we can and do use the steepest descent analysis of Section \ref{Section: steepest descent} applied to the weight 
\begin{equation}
w_t(x)=e^{-nV(x)}e^{W_t(x)} \om(x).
\end{equation}
From Proposition \ref{prop: general diff identity}, we have the following differential identities 
\begin{align}
& \partial_{t} \log L_{n}(\vec{\alpha},\vec{\beta},V,W_{t}) =  \frac{1}{2\pi i}\int_{-1}^{+\infty}[Y^{-1}(x)Y^{\prime}(x)]_{21}\partial_{t}w_{t}(x)dx, \label{diff identity L t}\\
& \partial_{t} \log J_{n}(\vec{\alpha},\vec{\beta},V,W_{t}) =  \frac{1}{2\pi i}\int_{-1}^{1}[Y^{-1}(x)Y^{\prime}(x)]_{21}\partial_{t}w_{t}(x)dx.  \label{diff identity J t}
\end{align}
The rest of the proof consists of inverting the transformations $Y \mapsto T \mapsto S \mapsto R$ and evaluating certain integrals by contour deformations. These computations are identical to those done in \cite[Section 7]{Charlier} for Gaussian-type weights and we omit them here.
\section{Appendix}
We recall here some well-known model RH problems: the Airy model RH problem, whose solution is denoted $\Phi_{\mathrm{Ai}}$ and the Bessel model RH problem, whose solution is denoted $\Phi_{\mathrm{Be}}(\cdot) = \Phi_{\mathrm{Be}}(\cdot;\alpha)$, where the parameter $\alpha$ is such that $\Re \alpha >-1$.
\subsection{Airy model RH problem}\label{subsection: model Airy}
\begin{itemize}
\item[(a)] $\Phi_{\mathrm{Ai}} : \mathbb{C} \setminus \Sigma_{A} \rightarrow \mathbb{C}^{2 \times 2}$ is analytic, and $\Sigma_{A}$ is shown in Figure \ref{figAiry}.
\item[(b)] $\Phi_{\mathrm{Ai}}$ has the jump relations
\begin{equation}\label{jumps P3}
\begin{array}{l l}
\Phi_{\mathrm{Ai},+}(z) = \Phi_{\mathrm{Ai},-}(z) \begin{pmatrix}
0 & 1 \\ -1 & 0
\end{pmatrix}, & \mbox{ on } \mathbb{R}^{-}, \\

\Phi_{\mathrm{Ai},+}(z) = \Phi_{\mathrm{Ai},-}(z) \begin{pmatrix}
 1 & 1 \\
 0 & 1
\end{pmatrix}, & \mbox{ on } \mathbb{R}^{+}, \\

\Phi_{\mathrm{Ai},+}(z) = \Phi_{\mathrm{Ai},-}(z) \begin{pmatrix}
 1 & 0  \\ 1 & 1
\end{pmatrix}, & \mbox{ on } e^{ \frac{2\pi i}{3} }  \mathbb{R}^{+} , \\

\Phi_{\mathrm{Ai},+}(z) = \Phi_{\mathrm{Ai},-}(z) \begin{pmatrix}
 1 & 0  \\ 1 & 1
\end{pmatrix}, & \mbox{ on }e^{ -\frac{2\pi i}{3} }\mathbb{R}^{+} . \\
\end{array}
\end{equation}
\item[(c)] As $z \to \infty$, $z \notin \Sigma_{A}$, we have
\begin{equation}\label{Asymptotics Airy}
\Phi_{\mathrm{Ai}}(z) = z^{-\frac{\sigma_{3}}{4}}N \left( I + \sum_{k=1}^{\infty} \frac{\Phi_{\mathrm{Ai,k}}}{z^{3k/2}} \right) e^{-\frac{2}{3}z^{3/2}\sigma_{3}},
\end{equation}
where $N = \frac{1}{\sqrt{2}}\begin{pmatrix}
1 & i \\ i & 1
\end{pmatrix}$ and $\Phi_{\mathrm{Ai,1}} = \frac{1}{8}\begin{pmatrix}
\frac{1}{6} & i \\ i & -\frac{1}{6}
\end{pmatrix}$.

As $z \to 0$, we have
\begin{equation}
\Phi_{\mathrm{Ai}}(z) = \bigO(1).
\end{equation} 
\end{itemize}
The Airy model RH problem was introduced and solved in \cite{Deiftetal}. We have
\begin{figure}[t]
    \begin{center}
    \setlength{\unitlength}{1truemm}
    \begin{picture}(100,55)(-5,10)
        \put(50,40){\line(1,0){30}}
        \put(50,40){\line(-1,0){30}}
        \put(50,39.9){\thicklines\circle*{1.2}}
        \put(50,40){\line(-0.5,0.866){15}}
        \put(50,40){\line(-0.5,-0.866){15}}
        \qbezier(53,40)(52,43)(48.5,42.598)
        \put(53,43){$\frac{2\pi}{3}$}
        \put(50.3,36.8){$0$}
        \put(65,39.9){\thicklines\vector(1,0){.0001}}
        \put(35,39.9){\thicklines\vector(1,0){.0001}}
        \put(41,55.588){\thicklines\vector(0.5,-0.866){.0001}}
        \put(41,24.412){\thicklines\vector(0.5,0.866){.0001}}
    \end{picture}
    \caption{\label{figAiry}The jump contour $\Sigma_{A}$ for $\Phi_{\mathrm{Ai}}$.}
\end{center}
\end{figure}
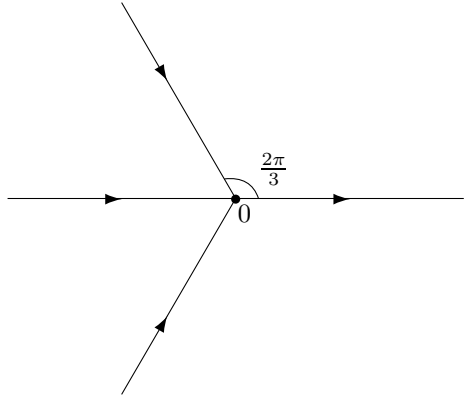
\begin{equation}
\Phi_{\mathrm{Ai}}(z) := M_{A} \times \left\{ \begin{array}{l l}
\begin{pmatrix}
\mbox{Ai}(z) & \mbox{Ai}(\omega^{2}z) \\
\mbox{Ai}^{\prime}(z) & \omega^{2}\mbox{Ai}^{\prime}(\omega^{2}z)
\end{pmatrix}e^{-\frac{\pi i}{6}\sigma_{3}}, & \mbox{for } 0 < \arg z < \frac{2\pi}{3}, \\
\begin{pmatrix}
\mbox{Ai}(z) & \mbox{Ai}(\omega^{2}z) \\
\mbox{Ai}^{\prime}(z) & \omega^{2}\mbox{Ai}^{\prime}(\omega^{2}z)
\end{pmatrix}e^{-\frac{\pi i}{6}\sigma_{3}}\begin{pmatrix}
1 & 0 \\ -1 & 1
\end{pmatrix}, & \mbox{for } \frac{2\pi}{3} < \arg z < \pi, \\
\begin{pmatrix}
\mbox{Ai}(z) & - \omega^{2}\mbox{Ai}(\omega z) \\
\mbox{Ai}^{\prime}(z) & -\mbox{Ai}^{\prime}(\omega z)
\end{pmatrix}e^{-\frac{\pi i}{6}\sigma_{3}}\begin{pmatrix}
1 & 0 \\ 1 & 1
\end{pmatrix}, & \mbox{for } -\pi < \arg z < -\frac{2\pi}{3}, \\
\begin{pmatrix}
\mbox{Ai}(z) & - \omega^{2}\mbox{Ai}(\omega z) \\
\mbox{Ai}^{\prime}(z) & -\mbox{Ai}^{\prime}(\omega z)
\end{pmatrix}e^{-\frac{\pi i}{6}\sigma_{3}}, & \mbox{for } -\frac{2\pi}{3} < \arg z < 0, \\
\end{array} \right.
\end{equation}
with $\omega = e^{\frac{2\pi i}{3}}$, Ai the Airy function and
\begin{equation}
M_{A} = \sqrt{2 \pi} e^{\frac{\pi i}{6}} \begin{pmatrix}
1 & 0 \\ 0 & -i
\end{pmatrix}.
\end{equation}
\subsection{Bessel model RH problem}\label{ApB}
\begin{itemize}
\item[(a)] $\Phi_{\mathrm{Be}} : \mathbb{C} \setminus \Sigma_{\mathrm{Be}} \to \mathbb{C}^{2\times 2}$ is analytic, where
$\Sigma_{\mathrm{Be}}$ is shown in Figure \ref{figBessel}.
\item[(b)] $\Phi_{\mathrm{Be}}$ satisfies the jump conditions
\begin{equation}\label{Jump for P_Be}
\begin{array}{l l} 
\Phi_{\mathrm{Be},+}(z) = \Phi_{\mathrm{Be},-}(z) \begin{pmatrix}
0 & 1 \\ -1 & 0
\end{pmatrix}, & z \in \mathbb{R}^{-}, \\

\Phi_{\mathrm{Be},+}(z) = \Phi_{\mathrm{Be},-}(z) \begin{pmatrix}
1 & 0 \\ e^{\pi i \alpha} & 1
\end{pmatrix}, & z \in e^{ \frac{2\pi i}{3} }  \mathbb{R}^{+}, \\

\Phi_{\mathrm{Be},+}(z) = \Phi_{\mathrm{Be},-}(z) \begin{pmatrix}
1 & 0 \\ e^{-\pi i \alpha} & 1
\end{pmatrix}, & z \in e^{ -\frac{2\pi i}{3} }  \mathbb{R}^{+}. \\
\end{array}
\end{equation}
\item[(c)] As $z \to \infty$, $z \notin \Sigma_{\mathrm{Be}}$, we have
\begin{equation}\label{large z asymptotics Bessel}
\Phi_{\mathrm{Be}}(z) = ( 2\pi z^{\frac{1}{2}} )^{-\frac{\sigma_{3}}{2}}N
\left(I+\sum_{k=1}^{\infty} \Phi_{\mathrm{Be},k} z^{-k/2}\right) e^{2z^{\frac{1}{2}}\sigma_{3}},
\end{equation}
where $\Phi_{\mathrm{Be},1} = \frac{1}{16}\begin{pmatrix}
-(1+4\alpha^{2}) & -2i \\ -2i & 1+4\alpha^{2}
\end{pmatrix}$.
\item[(d)] As $z$ tends to 0, the behaviour of $\Phi_{\mathrm{Be}}(z)$ is
\begin{equation}\label{local behaviour near 0 of P_Be}
\begin{array}{l l}
\displaystyle \Phi_{\mathrm{Be}}(z) = \left\{ \begin{array}{l l}
\begin{pmatrix}
\bigO(1) & \bigO(\log z) \\
\bigO(1) & \bigO(\log z) 
\end{pmatrix}, & |\arg z| < \frac{2\pi}{3}, \\
\begin{pmatrix}
\bigO(\log z) & \bigO(\log z) \\
\bigO(\log z) & \bigO(\log z) 
\end{pmatrix}, & \frac{2\pi}{3}< |\arg z| < \pi,
\end{array}  \right., & \displaystyle \mbox{ if } \Re \alpha = 0, \\[0.8cm]
\displaystyle \Phi_{\mathrm{Be}}(z) = \left\{ \begin{array}{l l}
\begin{pmatrix}
\bigO(1) & \bigO(1) \\
\bigO(1) & \bigO(1) 
\end{pmatrix}z^{\frac{\alpha}{2}\sigma_{3}}, & |\arg z | < \frac{2\pi}{3}, \\
\begin{pmatrix}
\bigO(z^{-\frac{\alpha}{2}}) & \bigO(z^{-\frac{\alpha}{2}}) \\
\bigO(z^{-\frac{\alpha}{2}}) & \bigO(z^{-\frac{\alpha}{2}}) 
\end{pmatrix}, & \frac{2\pi}{3}<|\arg z | < \pi,
\end{array} \right. , & \displaystyle \mbox{ if } \Re \alpha > 0, \\[0.8cm]
\displaystyle \Phi_{\mathrm{Be}}(z) = \begin{pmatrix}
\bigO(z^{\frac{\alpha}{2}}) & \bigO(z^{\frac{\alpha}{2}}) \\
\bigO(z^{\frac{\alpha}{2}}) & \bigO(z^{\frac{\alpha}{2}}) 
\end{pmatrix}, & \displaystyle \mbox{ if } \Re \alpha < 0.
\end{array}
\end{equation}
\end{itemize}
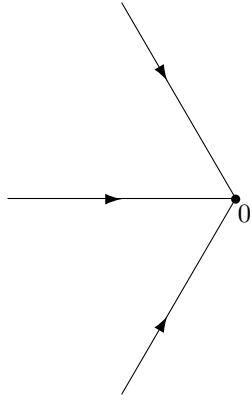
\begin{figure}[t]
    \begin{center}
    \setlength{\unitlength}{1truemm}
    \begin{picture}(100,55)(-5,10)
        \put(50,40){\line(-1,0){30}}
        \put(50,39.9){\thicklines\circle*{1.2}}
        \put(50,40){\line(-0.5,0.866){15}}
        \put(50,40){\line(-0.5,-0.866){15}}
        \put(50.3,36.8){$0$}
        \put(35,39.9){\thicklines\vector(1,0){.0001}}
        \put(41,55.588){\thicklines\vector(0.5,-0.866){.0001}}
        \put(41,24.412){\thicklines\vector(0.5,0.866){.0001}}

    \end{picture}
    \caption{\label{figBessel}The jump contour $\Sigma_{B}$ for $P_{\mathrm{Be}}(\zeta)$.}
\end{center}
\end{figure}
This RH problem was introduced and solved in \cite{KMcLVAV}. Its unique solution is given by 
\begin{equation}\label{Phi explicit}
\Phi_{\mathrm{Be}}(z)=
\begin{cases}
\begin{pmatrix}
I_{\alpha}(2z^{\frac{1}{2}}) & \frac{ i}{\pi} K_{\alpha}(2z^{\frac{1}{2}}) \\
2\pi i z^{\frac{1}{2}} I_{\alpha}^{\prime}(2 z^{\frac{1}{2}}) & -2 z^{\frac{1}{2}} K_{\alpha}^{\prime}(2 z^{\frac{1}{2}})
\end{pmatrix}, & |\arg z | < \frac{2\pi}{3}, \\

\begin{pmatrix}
\frac{1}{2} H_{\alpha}^{(1)}(2(-z)^{\frac{1}{2}}) & \frac{1}{2} H_{\alpha}^{(2)}(2(-z)^{\frac{1}{2}}) \\
\pi z^{\frac{1}{2}} \left( H_{\alpha}^{(1)} \right)^{\prime} (2(-z)^{\frac{1}{2}}) & \pi z^{\frac{1}{2}} \left( H_{\alpha}^{(2)} \right)^{\prime} (2(-z)^{\frac{1}{2}})
\end{pmatrix}e^{\frac{\pi i \alpha}{2}\sigma_{3}}, & \frac{2\pi}{3} < \arg z < \pi, \\

\begin{pmatrix}
\frac{1}{2} H_{\alpha}^{(2)}(2(-z)^{\frac{1}{2}}) & -\frac{1}{2} H_{\alpha}^{(1)}(2(-z)^{\frac{1}{2}}) \\
-\pi z^{\frac{1}{2}} \left( H_{\alpha}^{(2)} \right)^{\prime} (2(-z)^{\frac{1}{2}}) & \pi z^{\frac{1}{2}} \left( H_{\alpha}^{(1)} \right)^{\prime} (2(-z)^{\frac{1}{2}})
\end{pmatrix}e^{-\frac{\pi i \alpha}{2}\sigma_{3}}, & -\pi < \arg z < -\frac{2\pi}{3},
\end{cases}
\end{equation}
where $H_{\alpha}^{(1)}$ and $H_{\alpha}^{(2)}$ are the Hankel functions of the first and second kind, and $I_\alpha$ and $K_\alpha$ are the modified Bessel functions of the first and second kind.

	
\newpage
\section*{Acknowledgements}
C.C. acknowledges support from the Swedish Research Council, Grant No. 2015-05430, and the European Research Council, Grant Agreement No. 682537, and R.G. acknowledges support by NSF-grant DMS-1700261. R.G. also acknowledges the support he received as a graduate student from the Department of Mathematical Sciences at IUPUI where this work was completed. We are grateful to Sergey Berezin, Maurice Duits and Peter Forrester for interesting discussions, and to Sergey Berezin and Alexander Bufetov for sharing with us a preliminary version of \cite{BerezinBufetov}. We also thank the referees for the numerous constructive remarks that have helped us to improve this paper.

\end{document}